\documentclass[
  reprint,
  aps,
  prx,
  longbibliography,
  superscriptaddress
]{revtex4-2}

\usepackage{style/setup}
\usepackage{style/macros}
\makeatletter
\newcommand{\DisableTOC}{%
  \@ifundefined{orig@addcontentsline}{%
    \let\orig@addcontentsline\addcontentsline
  }{}%
  \renewcommand{\addcontentsline}[3]{}%
}
\newcommand{\EnableTOC}{%
  \@ifundefined{orig@addcontentsline}{}{%
    \let\addcontentsline\orig@addcontentsline
  }%
}
\makeatother

\begin{document}
\DisableTOC

\title{Ansatz-Free Learning of Lindbladian Dynamics In Situ}

\author{Petr Ivashkov}
\thanks{Equal contribution}
\affiliation{Department of Information Technology and Electrical Engineering, ETH Z\"urich, Z\"urich, Switzerland}
\affiliation{Department of Physics, Harvard University, Cambridge, MA 02138, USA}
\author{Nikita Romanov}
\thanks{Equal contribution}
\affiliation{Department of Physics, Harvard University, Cambridge, MA 02138, USA}
\affiliation{Quantum Science and Engineering, Harvard University, Cambridge, MA 02138, USA}
\author{Weiyuan Gong}
\affiliation{School of Engineering and Applied Sciences, Harvard University, Allston, MA 02134, USA}
\author{Andi Gu}
\affiliation{Department of Physics, Harvard University, Cambridge, MA 02138, USA}

\author{Hong-Ye Hu}
\email[]{hongyehu.physics@gmail.com}
\affiliation{Department of Physics, Harvard University, Cambridge, MA 02138, USA}
\author{Susanne F. Yelin}
\email[]{syelin@g.harvard.edu}
\affiliation{Department of Physics, Harvard University, Cambridge, MA 02138, USA}

\begin{abstract}
Characterizing the dynamics of open quantum systems at the level of microscopic interactions and error mechanisms is essential for calibrating quantum hardware, designing robust simulation protocols, and developing tailored error-correction methods. Under Markovian noise/dissipation, a natural characterization approach is to identify the full Lindbladian generator that gives rise to both coherent (Hamiltonian) and dissipative dynamics. Prior protocols for learning Lindbladians from dynamical data assumed pre-specified interaction structure, which can be restrictive when the relevant noise channels or control imperfections are not known in advance. In this paper, we present a sample-efficient protocol for learning sparse Lindbladians without assuming any a priori structure. Our protocol is ancilla-free, uses only product-state preparations and Pauli-basis measurements, admits provably stable coefficient reconstruction, and achieves near-optimal time resolution, making it compatible with near-term experimental capabilities. Together, this provides a systematic route to scalable characterization of open-system quantum dynamics, especially in settings where the error mechanisms of interest are unknown.
\end{abstract}

\maketitle
\twocolumngrid


\section{Introduction}\label{sec:introduction}

Recent advances in quantum error correction have brought fault-tolerant quantum computation closer to reality \cite{bluvstein2024logical,bluvstein2025fault,google2024threshold, google2023suppressing,besedin2026lattice,putterman2025hardware, brock2025quantum,paetznick2024demonstration,ryan2024high}. At the same time, achieving fault tolerance in practice relies on precise control at the physical level, which in turn requires a detailed understanding of device dynamics and the underlying noise sources. This need has motivated the development of scalable benchmarking and characterization protocols \cite{emerson2005scalable,knill2008randomized,dankert2009exact,flammia2011direct,da2011practical,magesan2011scalable,moussa2012practical,wallman2015estimating,wallman2016robust,boixo2018characterizing,erhard2019characterizing,proctor2022scalable, harper2021fast,Bayesian_noise}, including randomized benchmarking, cycle benchmarking, and cross-entropy benchmarking, among others. These tools have proven invaluable for assessing overall device performance. However, they typically provide only coarse-grained metrics, such as the average circuit fidelity or effective error rates, and therefore offer limited insight into the microscopic mechanisms governing coherent and incoherent errors. At the opposite extreme, quantum process tomography  can fully characterize a specified quantum channel, but its experimental cost scales exponentially with system size \cite{chuang1997prescription, flammia2012quantum}, which limits its application to current large-scale quantum devices.

A direct way to bridge this characterization gap is to identify the generator of the device's dynamics, providing a microscopic description without requiring exponentially many experimental queries. In the Markovian regime, the dynamics of a quantum device are governed by a Lindblad master equation capturing both \emph{coherent} (Hamiltonian) and \emph{dissipative} contributions \cite{gorini1976completely, lindblad1976generators}. Full knowledge of the Lindbladian yields a complete characterization of the device dynamics and enables a range of downstream applications, including device optimization \cite{boulant2003robust, shulman2014suppressing, innocenti2020supervised, samach2022lindblad, hangleiter2024robustly, dobrynin2024compressed}, error mitigation \cite{temme2017error, bravyi2021mitigating, van2023probabilistic, cai2021multi, kim2023evidence, strikis2021learning,generalized_DD}, design of tailored quantum error-correcting codes \cite{fletcher2008channel,aliferis2008fault,bonilla2021xzzx,tuckett2019tailoring,chuang1997bosonic,wu2025bias,kuehnke2025hardware, valenti2019hamiltonian,rl_qec}, and verification of analog quantum simulation \cite{kraft2025bounded}. In this work, we study the problem of learning unknown Lindbladian dynamics using only access to its time evolution, without additional quantum control during the evolution, a setting often referred to as \emph{in situ} learning \cite{liu2025optimal}. This restriction is crucial because auxiliary control (e.g., dynamical decoupling~\cite{viola1999dynamical},
Trotterization~\cite{childs2021theory}, or interleaved/continuous-control schemes~\cite{dutkiewicz2024advantage,huang2023learning,hu2025ansatz,bakshi2024structure}) can modify the effective generator being probed, thereby introducing systematic bias in the inferred Lindbladian.

A simpler and well-studied limit of this problem is Hamiltonian learning, which aims to infer the generator of unitary dynamics, typically under idealized assumptions such as noiseless evolution~\citep{da2011practical,bairey2019learning,zubida2021optimal,haah2024learning,stilck2024efficient,gu2024practical, caro2024learning, mirani2024learning, ni2024quantum, dutkiewicz2024advantage, bakshi2024structure, ma2024learning, zhao2025learning, hu2025ansatz, sinha2025improvedhamiltonianlearningsparsity, shin2026heisenberg} or the ability to prepare equilibrium states (e.g., Gibbs states) of the Hamiltonian~\citep{qi2019determining, evans2019scalable, gu2024practical, anshu2021sample, rouze2024learning, haah2024learning, bakshi2024learning}. Over the past decade, a broad class of algorithms has been developed, achieving near-optimal performance across a wide range of control models and structural assumptions \cite{haah2024learning, huang2023learning, dutkiewicz2024advantage, hu2025ansatz, zhao2025learning}. While these advances are central for learning generators of unitary dynamics, realistic quantum devices inevitably experience system–environment interactions, leading to dissipation and noise that violate the assumptions underlying most Hamiltonian-learning protocols \cite{zhou2018achieving,cotler2026noisylearning,gong2026multiparameter}.
As a result, Hamiltonian-learning guarantees do not automatically extend, and dissipation can obscure the coherent generator one aims to learn. More fundamentally, a faithful characterization of realistic evolution requires identifying \emph{both} the coherent and dissipative components, making Lindbladian learning a necessary next milestone.

Recently, significant progress has been made toward scalable protocols for learning Lindbladian generators. One line of work infers Lindbladians from their steady states \cite{bairey2020learning}. However, steady states generally do not uniquely determine the underlying generator \cite{baumgartner2008analysis, bairey2020learning} and can be challenging to prepare in practice \cite{temme2010chi}. We therefore focus on learning from real-time evolution, where the dynamics can be probed by simply letting the device evolve under its intrinsic Lindbladian. In this dynamical setting, recent works extend derivative-estimation--based Hamiltonian-learning techniques to Lindbladian generators and provide rigorous guarantees \cite{stilck2024efficient,franca2025learning,montana2025efficiently}. Despite their efficiency, these methods assume the generator structure is known \emph{a priori}, and their provable guarantees are currently limited to restricted dissipative models, e.g., dissipation generated by single-qubit jump operators. Other approaches have demonstrated promising numerical performance on systems with tens to hundreds of qubits \cite{pastori2022characterization, olsacher2025hamiltonian,berg2025large,kraft2025bounded}, but they either lack rigorous guarantees or similarly rely on prior structural assumptions. In many practical scenarios, however, the relevant Hamiltonian terms and error mechanisms are unknown and may be nonlocal~\cite{kraft2025bounded}, leading to an exponentially large space of candidate Lindbladian generators. This motivates the central question addressed in this work: can one efficiently learn the full Lindbladian without strong prior assumptions on its interaction structure, using only its native dynamics (i.e., without interleaving quantum control during the evolution)? We refer to this setting as \emph{ansatz-free} learning of Lindbladian dynamics \emph{in situ}. The ansatz-free requirement removes structural biases commonly imposed in learning algorithms, while the in situ constraint ensures that the learned generator faithfully reflects the device’s intrinsic dynamics by avoiding unintended modifications induced by control protocols.

The main result of this work is an efficient algorithm that achieves the described learning task in an experimentally practical setting, restricted to product-state preparations, time evolution, and Pauli-basis measurements. Concretely, any $n$-qubit Lindbladian governing the evolution of a density operator $\rho$ can be written as
\begin{equation} \label{eq:lindbladian-intro}
    \begin{aligned}
        \dv{}{t}\rho
    =\Lindblad(\rho) &= -i\underbrace{\sum_{P_i\in \HamiltStruct} h_i \comm{P_i}{\rho}}_{\mathrm{Hamiltonian}\,(H)} + \\ &+\underbrace{\sum_{P_i,P_j \in \DiagDissStruct} a_{ij}\left(P_i \rho P_j - \tfrac12\acomm{P_j P_i}{\rho}\right)}_{\mathrm{Dissipator}\,(\mathcal D)},
    \end{aligned}
\end{equation}
where $P_i,P_j$ are $n$-qubit traceless Pauli operators \footnote{Although Lindbladian generators admit a unitary gauge freedom in the choice of jump operators, the Pauli-basis expansion \eqref{eq:lindbladian-intro} is unique.}. Here $\HamiltStruct$ and $\DiagDissStruct$ denote the Pauli supports of the Hamiltonian and dissipator, respectively, and $\{h_i,a_{ij}\}$ are unknown coefficients. We denote by $M$ an upper bound on the number of Pauli terms appearing with nonzero coefficients in the Lindbladian, providing a natural notion of sparsity. Our algorithm reconstructs all nonzero Lindbladian coefficients without prior knowledge of which Pauli terms appear in its expansion or any assumption of bounded locality. The resulting end-to-end protocol achieves quantum query and time complexity $\widetilde{\mathcal O}(\varepsilon^{-4})$ in the target accuracy $\varepsilon$. We highlight that our coefficient reconstruction is provably stable and does not depend on the conditioning of a linear system. Finally, the time resolution of our algorithm is near-optimal and depends on the target accuracy only polylogarithmically, so higher precision does not force vanishingly short evolution times.

\emph{Note added.}
Since the first version of this manuscript was posted, several concurrent and
subsequent works have developed related approaches to Lindbladian learning
without prior knowledge of the interaction structure.
Refs.~\cite{arad2026near,mobus2026robust,lewis2026learning} consider
$k$-local settings, while
Refs.~\cite{sinha2026lindblad,zhou2026characterizing,chen2026learning} develop methods applicable
without locality assumptions.

\section{Overview of Results}
\label{sec:results}

This section provides an overview of the learning algorithm and states the main guarantees informally.

We consider the task of learning an unknown time-independent $n$-qubit Lindbladian $\mathcal{L}$ of the form presented in \cref{eq:lindbladian-intro} from dynamical data. 
Each query consists of preparing a product Pauli eigenstate, applying the Markovian evolution
\begin{equation}
    \mathcal{E}_t = e^{t\mathcal{L}}, \qquad t \ge 0,
\end{equation}
and then measuring each qubit in a chosen single-qubit Pauli basis. Throughout, we restrict to ancilla-free experiments.

As the Lindbladian global scale can be absorbed into the evolution time, we assume without loss of generality that $|h_i|\le 1$ and $|a_{ij}|\le 1$. We also let $\eta := \min\left(\min_{h_{i}\neq 0}|h_i|,\ \min_{a_{ij}\neq 0}|a_{ij}|\right)$ denote
the smallest nonzero Lindbladian coefficient \footnote{In practice, one may choose $\eta$ as a conservative lower bound satisfying
$\eta \le \min\left(\min_{h_{i}\neq 0}|h_i|,\ \min_{a_{ij}\neq 0}|a_{ij}|\right)$.
If $\eta$ is chosen larger than the true minimum nonzero magnitude, the structure-learning routine
returns the support of the $\eta$-heavy part of the dissipator and (almost all of) the $\eta$-heavy
Hamiltonian support (see \ref{remark:eta-heavy-support}).
While this can be of independent interest, our coefficient-learning guarantees
(\cref{theorem:ancilla_free_coefficient_learning}) require supersets of the \emph{full} structure as input
(hence $\eta$ below the smallest nonzero coefficient), since otherwise, residual weak terms can act as adversarial model mismatch and bias the recovered coefficients.}.

\begin{figure*}
    \centering
    \includegraphics[width=\linewidth]{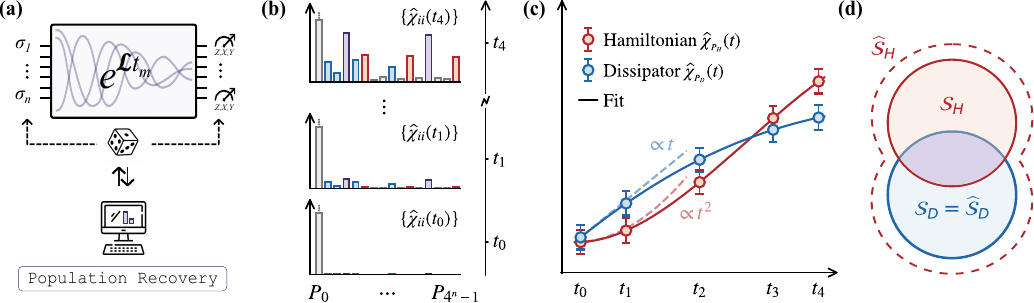}
    \caption{
    \textbf{Structure learning from Pauli error rates at short times.} The protocol outputs sets $\widehat{\mathcal{S}}_D$ and $\widehat{\mathcal{S}}_H$ containing the true dissipator and Hamiltonian structures $\mathcal{S}_D$ and $\mathcal{S}_H$, respectively.  
    (a)~We fix a set of short evolution times $\{t_m\}$ and, for each $t_m$, query the channel $e^{\mathcal{L}t_m}$ using the population-recovery protocol~\cite{flammia2021pauli} to obtain a distribution of Pauli error rates $\{\widehat{\chi}_{ii}(t_m)\}$. Each query consists of preparing a random Pauli product state, evolving for time $t_m$, and measuring in a random Pauli basis.
    (b)~Collecting these estimated distributions over $\{t_m\}$ yields time traces of Pauli error rates. For very short times, only the identity component is appreciable. As $t$ increases, the rates associated with Hamiltonian terms (red), dissipator terms (blue), and terms in the overlap $\mathcal{S}_H\cap\mathcal{S}_D$ (purple) begin to grow, spreading the rate distribution.
    (c)~For each observed trace, we fit a low-degree Chebyshev interpolant and use the first and second derivatives of the fit at $t=0$ to decide whether the Pauli enters through dissipation (first order) or through the Hamiltonian (second order): Pauli rates corresponding to dissipator terms $P_D \in \mathcal{S}_D$ exhibit linear short-time growth, while purely Hamiltonian rates $P_H \in \mathcal{S}_H \setminus \mathcal{S}_D$ enter only at second order and exhibit quadratic short-time growth.
    (d)~Finally, we threshold these derivative estimates to obtain candidate supports $\widehat{\mathcal{S}}_D$ and $\widehat{\mathcal{S}}_H$ that satisfy $\widehat{\mathcal{S}}_D = \mathcal{S}_D$ and $\widehat{\mathcal{S}}_H \supseteq \mathcal{S}_H$.
    }
    \label{fig:structure_learning}
\end{figure*}

\begin{problem}[Ansatz-free Lindbladian learning]
\label{problem:lindbladian_learning}
Let $\mathcal{L}$ be an unknown time-independent $n$-qubit Lindbladian as in~\cref{eq:lindbladian-intro}. 
Assume that the total number of nonzero coefficients among $\{h_i\}\cup\{a_{ij}\}_{i,j}$ is at most $M$, and that the supports $\HamiltStruct$ and $\DiagDissStruct$ are not known.
Given query access to $e^{\mathcal{L}t}$, output estimates $\{\widehat{h}_i,\widehat{a}_{ij}\}$ of all nonzero coefficients to within additive error $\varepsilon$, with high success probability.
\end{problem}

We summarize our solution to the ansatz-free Lindbladian learning in two parts: a structure-learning procedure that outputs candidate supports, and a coefficient-learning procedure that estimates the associated coefficients. This two-stage decomposition naturally introduces two accuracy scales: $\eta$, which governs the
resolution required to identify nonzero terms in the structure-learning stage, and $\varepsilon$, which
controls the final additive accuracy of the recovered coefficients.

\noindent\textbf{Structure learning.}
Our first result identifies candidate supports of the Hamiltonian and dissipator by estimating the first two time derivatives of the \emph{Pauli error rates} in the short-time regime. Concretely, writing the channel in the Pauli basis as $\mathcal{E}_t(\rho)=\sum_{i,j}\chi_{ij}(t)\,P_i\rho P_j$, the Pauli error rates are the diagonal entries $\{\chi_{ii}(t)\}_i$. Importantly, these quantities can be efficiently estimated experimentally within our query model \cite{flammia2021pauli}. The key insight is that the first and second time-derivatives of these error rates can be directly related to the coefficients of the dissipator and the Hamiltonian. By identifying all derivatives above an appropriate threshold, we obtain the Lindbladian structure estimates $\widehat{\mathcal S}_H$ and $\widehat{\mathcal S}_D$, as summarized in \Cref{result:structure_learning_informal} and illustrated schematically in \cref{fig:structure_learning}~(a--d). The structure-learning algorithm is described in \cref{sec:structure_learning_main_A}, with formal guarantees stated in \Cref{theorem:ancilla_free_structure_learning}.

\begin{result}[Structure Learning]
    \label{result:structure_learning_informal}
    There exists an ancilla-free quantum algorithm that with high probability identifies the dissipator structure $\widehat{\mathcal{S}}_D =\mathcal S_D$ and a superset $\widehat{\mathcal{S}}_H \supseteq \mathcal{S}_H$ of the Hamiltonian  structure using
    \begin{equation}
    m = \widetilde{\mathcal{O}}\left(M^4/\eta^4\right)
    \end{equation}
    applications of the channel $e^{\mathcal{L}t}$. The size of the Hamiltonian superset $\widehat{\mathcal{S}}_H$ is at most $M^2$. The algorithm prepares a product input state, evolves it for time $t = \widetilde{\Theta}(1/M)$, and measures each qubit in a Pauli basis. 
\end{result}

\begin{figure*}
    \centering
    \includegraphics[width=\linewidth]{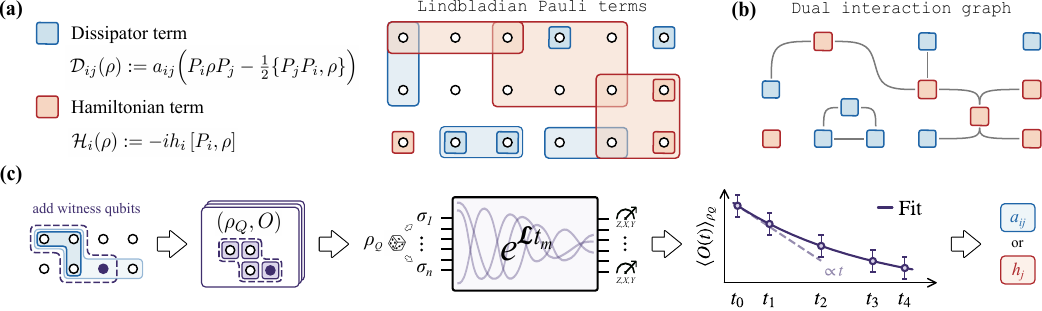}
    \caption{
    \textbf{Coefficient learning from Pauli observables at short times.}
    (a)~The (unknown) Lindbladian is a sparse sum of Pauli terms: coherent Hamiltonian terms (red) and dissipative terms (blue), each acting nontrivially on some subset of qubits, represented by white circles.
    (b)~In the dual picture, we define an interaction graph whose vertices are the nonzero Lindbladian Pauli terms and whose edges connect terms with overlapping qubit support. The maximum degree $\mathfrak d$ of the dual graph upper bounds how rapidly Heisenberg-evolved Pauli observables can spread. If $\mathfrak d$ is not known in advance, it can be upper bounded from the terms induced by the candidate structures $\widehat{\mathcal{S}}_H, \widehat{\mathcal{S}}_D$.
    (c)~For each target coefficient, we first enlarge its natural support by adding witness qubits where needed to remove local ambiguities. Using the candidate structures, we then deterministically select a small family of Pauli probes $(\rho_Q,O)$ whose unwanted contributions cancel under averaging. For each selected probe, we sample the corresponding expectation value at short Gauss--Chebyshev times $\{t_m\}$, fit a low-degree interpolant, and extract its derivative at $t=0$. Combining these derivatives according to the explicit reconstruction formula yields the target dissipative coefficient $a_{ij}$ or Hamiltonian coefficient $h_j$.
    }
    \label{fig:coefficient_learning}
\end{figure*}

\noindent\textbf{Coefficient learning.}
With $\widehat{\mathcal S}_H$ and $\widehat{\mathcal S}_D$ in hand, our second result estimates the associated coefficients to accuracy $\varepsilon$ from short-time derivatives of Pauli expectation values. Starting from the standard derivative identity~\cite{gu2024practical,stilck2024efficient,haah2024learning}, we use
\begin{equation}
\left.\dv{}{t}\expval{O(t)}\right|_{t=0}
=
\tr\left(\rho\,\mathcal L^\dagger(O)\right),
\end{equation}
which gives a linear mapping from the unknown Lindbladian coefficients to experimentally accessible derivatives of expectation values of Pauli observables. Our main contribution is an explicit probe-selection procedure that inverts this mapping for every target coefficient. Specifically, it constructs a small family of Pauli input--observable pairs whose derivative responses isolate that coefficient. After resolving possible local non-identifiability using a support-enlargement construction adapted from \citet{montana2025efficiently}, we exploit Pauli commutation relations to choose probes for which all unwanted contributions cancel exactly under averaging. The resulting reconstruction is provably stable, so errors in the measured derivatives are not amplified.

A central tool in both the probe construction and its analysis is our extension of the Hamiltonian \emph{dual interaction graph}~\cite{gu2024practical,haah2024learning} to Lindbladians. Its vertices are the individual nonzero Pauli components of $\mathcal L^\dagger$, and two vertices are connected whenever their qubit supports overlap, as illustrated in \cref{fig:coefficient_learning}~(a,b); let $\mathfrak d$ denote its maximum degree. Intuitively, $\mathfrak d$ quantifies the local interaction complexity of the generator. We use $\mathfrak d$ to bound the number of terms that can contribute to a targeted probe derivative, implying that all
$\widehat M:=|\widehat{\mathcal S}_H|+|\widehat{\mathcal S}_D|^2$
coefficients can be reconstructed using only $\mathcal O(\widehat M\mathfrak d)$ Pauli derivatives. The same quantity also controls the growth of higher time derivatives of the selected observables, allowing each required derivative to be estimated by Gauss--Chebyshev interpolation using $\widetilde{\mathcal O}(\min\{\mathfrak d^4,M^2\}/\varepsilon^2)$ channel queries. As discussed in \cref{sec:dual-interaction-graph}, many physically motivated Lindbladians have $\mathfrak d=O(1)$. More generally, $\mathfrak d$ can be upper bounded using the learned candidate structures. The full coefficient-learning procedure is presented in \cref{sec:coefficient_learning_main}, with formal guarantees in \Cref{theorem:ancilla_free_coefficient_learning}; \cref{fig:coefficient_learning}~(c) illustrates the procedure schematically.

\begin{result}[Coefficient Learning]
    \label{result:coefficient_learning_informal}
    There exists an ancilla-free quantum algorithm that, given the candidate structures $\widehat{\mathcal S}_H,\widehat{\mathcal S}_D$ returned by structure learning, outputs estimates $(\widehat{\bm h},\widehat{\bm a})$ satisfying
    \begin{equation}
        \|\widehat{\bm h}-\bm h\|_\infty\le\varepsilon
        \qq{and}
        \|\widehat{\bm a}-\bm a\|_\infty\le\varepsilon
    \end{equation}
    with high probability. The total number of applications of $e^{\mathcal Lt}$ is
    \begin{equation}
        m = 
        \widetilde{\mathcal O}\left(
            \widehat M\,\mathfrak d\,
            \frac{\min\{\mathfrak d^4,M^2\}}{\varepsilon^2}
        \right).
    \end{equation}
    The algorithm prepares product Pauli eigenstates, evolves for times $t=\mathcal O(\max\{\mathfrak d^{-2},M^{-1}\})$, and measures each qubit in a Pauli basis.
\end{result}

Running the two stages sequentially yields an end-to-end, ansatz-free Lindbladian learning algorithm. Using $\widehat M=\mathcal O(M^2)$ and targeting full resolution of all nonzero coefficients with $\varepsilon=\eta/2$, the complete protocol requires
\begin{equation}
    m_{\mathrm{tot}}
    =
    \widetilde{\mathcal O}\left(
        \frac{
            M^5
        }{\varepsilon^4}
    \right)
\end{equation}
quantum channel queries. Notably, although both stages rely on derivative estimation, they probe fundamentally different quantities: structure learning uses derivatives of global Pauli error rates, whereas coefficient learning uses derivatives of targeted Pauli expectation values constructed from the learned candidate structure.

\noindent\textbf{Optimality of the resolution time.}
Crucially, we prove that the minimum time scale required for the complete
learning protocol, $t_{\min}=\widetilde{\Theta}(M^{-1})$, is optimal at the
level of its polynomial dependence on $M$. In particular, enforcing a coarser
time resolution $t \gtrsim t_0=\Theta(M^{-(1-\theta)})$ for constant threshold
$\eta$ and $\theta>0$ can incur an exponential overhead in the required
sampling complexity, as summarized in \Cref{result:informal_lower}.
Intuitively, we show that time $t \gtrsim t_0$ suffices to drive any input
state exponentially close to a steady state of some $M$-sparse Lindbladian.
Since steady states generally do not uniquely identify generators, we then
prove that recovering the Lindbladian structure is hard. The outline of the
proof is provided in \Cref{sec:optimality_resolution_time} with a formal
statement in \Cref{thm:lower}.

\begin{result}[Optimality of Time Resolution]\label{result:informal_lower}
Consider a (possibly adaptive) protocol that can prepare any single-copy input state, evolve it for time $t\geq t_0$ with a time resolution $t_0=\Theta(M^{-(1-\theta)})$ for some constant $\theta>0$, and perform any single-copy measurement. There always exists a Lindbladian that only contains $M=\Theta(n^{\kappa})$ nonzero terms with $\kappa=O(1/\theta)$ and $\eta = 1/3$ such that identifying its structure requires at least $\Omega(\exp(n))$ channel queries.
\end{result}

\section{Structure learning}\label{sec:structure_learning_main_A}
The core challenge of structure learning is that even if the Lindbladian $\mathcal{L}$ is sparse, the time-evolution channel $e^{\mathcal{L}t}$ can become dense once $t$ is sufficiently large. To reveal the Lindbladian structure, we therefore focus on the short-time regime, where the growth of the operator structure of $e^{\mathcal{L}t}$ is governed by the first few terms of the Taylor expansion
\begin{equation}\label{eq:channel_taylor_expansion}
	e^{\mathcal{L}t}(\rho)
	\;=\;
	\rho
	\;+\;
	t\mathcal{L}(\rho)
	\;+\;
	\tfrac{t^2}{2}\mathcal{L}^2(\rho)
	\;+\;
	\mathcal{O}(t^3).
\end{equation} 
It is instructive to consider the $\chi$-matrix representation of $e^{\mathcal{L}t}$,
\begin{equation}\label{eq:chi_matrix_representation}
    e^{\mathcal{L}t}(\rho)
    \;=\;
    \sum_{i,j}\chi_{ij}(t)\,P_i\rho P_j.
\end{equation}
The diagonal entries of the (time-dependent) $\chi$-matrix are known as the Pauli error rates $\{\chi_{ii}(t)\}_i$ and form a probability distribution. At $t=0$, all rates are $0$ except $\chi_{00}(0)=1$, corresponding to the identity channel. As $t$ increases, the weight of $\mathcal{L}$ and $\mathcal{L}^2$ in the Taylor expansion grows, giving rise to new Pauli error rates at the expense of the decreasing identity component, as shown in \cref{fig:structure_learning}~(b). The key observation that enables our Lindbladian structure learning algorithm is that the first two derivatives of these Pauli error rates at $t=0$ can be directly related to the Lindbladian terms with non-vanishing coefficients. Indeed, by comparing terms in \cref{eq:channel_taylor_expansion,eq:chi_matrix_representation}, we can notice that $\mathcal{L}(\rho) \;=\; \sum_{i,j} \chi_{ij}^{(1)}\, P_i \rho P_j$, where we defined $\chi_{ij}^{(1)}=\dv{t}\chi_{ij}(0)$. Similarly, $\mathcal{L}^2(\rho) \;=\; \sum_{i,j} \chi_{ij}^{(2)}\, P_i \rho P_j$ with $\chi_{ij}^{(2)}=\dv[2]{t}\chi_{ij}(0)$. In \cref{sec:structure_learning} we show that
\begin{align}
    \chi_{ii}^{(1)} &= a_{ii},
    && \text{if } a_{ii} > 0, \label{eq:chi_first_deriv}\\[4pt]
    \chi_{ii}^{(1)} &= 0,\quad
    \chi_{ii}^{(2)} \ge 2h_i^{2},
    && \text{if } a_{ii} = 0. \label{eq:chi_second_deriv}
\end{align}
Therefore, the presence of a Pauli $P_i$ in the dissipator structure $\DiagDissStruct$ leads to a non-vanishing first derivative $\chi_{ii}^{(1)}$, while the presence of a Pauli $P_i$ in the Hamiltonian-only structure $\HamiltStruct\setminus \DiagDissStruct$ leads to a non-vanishing second derivative $\chi_{ii}^{(2)}$. Hence, by identifying all first and second derivatives above appropriate thresholds, one is guaranteed to identify a superset of all Pauli terms present in the Lindbladian, as shown in \cref{fig:structure_learning}~(c,d).

The key subroutine that allows the estimation of Pauli rate derivatives is the protocol of \citet{flammia2021pauli} (see \cref{fig:structure_learning}~(a)), who showed that, nontrivially, all Pauli rates can be estimated efficiently using the classical population recovery algorithm. The protocol uses only product-state preparations and single-qubit Pauli measurements to estimate all diagonal rates simultaneously, producing estimators $\widehat{\chi_{ii}}(t)$ satisfying
\begin{equation} \label{eq:pauli-error-rates-estimation-error_A}
    \bigl|\widehat{\chi_{ii}}(t)-\chi_{ii}(t)\bigr|\le \varepsilon_s,\qquad \forall\, i,
\end{equation}
with probability at least $1-\delta$, using
$\mathcal{O}\left(\varepsilon_s^{-2}\log\frac{n}{\varepsilon_s\delta}\right)$ queries to the $n$-qubit channel $\mathcal{E}_t$. Importantly, the protocol outputs at most $\mathcal{O}(1/\varepsilon_s)$ nonzero estimates $\widehat\chi_{ii}(t)$ and treats all other rates as zero.

Naively, one could estimate the derivatives $\chi^{(1)}_{ii}$ and $\chi^{(2)}_{ii}$ via finite differences by $\chi^{(1)}_{ii}\;\approx\;\tfrac{1}{t}\bigl[\chi_{ii}(t)-\chi_{ii}(0)\bigr]$, and similarly for $\chi^{(2)}_{ii}$ using a second-order stencil. However, identifying whether $\chi^{(1)}_{ii} \gtrsim \eta$ requires estimating $\chi^{(1)}_{ii}$ to additive precision $\Theta(\eta)$, and controlling the finite-difference bias forces $t\propto \eta$. In other words, it would require resolving vanishingly short dynamics as the desired precision increases, leading to infeasible experimental requirements. 

To alleviate impractical time resolution, we use Chebyshev interpolation, which achieves the desired derivative estimates with exponentially better time resolution in terms of $\eta$. Specifically, we fix a set of $r{+}1$ carefully chosen time points $\{t_m\}_{m=0}^r$ within a short-time window $[0,\tau_{\max}]$. At each time $t_m$, we estimate all Pauli error rates to a fixed accuracy $\varepsilon_s$ using the population recovery protocol. Finally, for each observed rate, we fit a low-degree Chebyshev interpolant to the noisy time series $\{\widehat{\chi}_{ii}(t_m)\}_{m=0}^r$ and then take the first and second derivatives of the fitted polynomial at $t=0$ as the estimators $\widehat{\chi}^{(1)}_{ii}$ and $\widehat{\chi}^{(2)}_{ii}$, as shown in \cref{fig:structure_learning}~(c).

The resulting derivative-estimation error has two contributions: a systematic interpolation bias controlled by higher-order derivatives of $\chi_{ii}(t)$, and the nodewise sampling noise $\varepsilon_s$. For an $M$-sparse Lindbladian, as in \cref{eq:lindbladian-intro}, we show in \cref{sec:deriv-bounds-pauli-error-rates} the uniform derivative bound
\begin{equation}
    \Big|\frac{d^\ell}{dt^\ell}\chi_{ii}(t)\Big|
    \;\le\;
    (2M)^\ell,\qquad \ell\ge 0,
\end{equation}
which implies that the Pauli error rates are uniformly smooth on time scales of order $1/M$. Using this bound, the interpolation bias can be kept below the target scale by choosing the time window $\tau_{\max}=\Theta(1/M)$ and polynomial degree $r=\Theta(\log(M/\eta))$. To control the contribution of sampling noise, we enforce that the Pauli-rate estimates returned by population recovery satisfy
$\varepsilon_s=\widetilde{\Theta}(\eta^2/M^2)$. 

With the choice of parameters above, the differentiation step returns $\widehat{\chi}^{(1)}_{ii}$ to additive accuracy $\eta/2$ and $\widehat{\chi}^{(2)}_{ii}$ to additive accuracy $\eta^2$, which suffices to resolve the nonzero Lindbladian terms according to \cref{eq:chi_first_deriv,eq:chi_second_deriv}. In particular, \cref{eq:chi_first_deriv} guarantees that
\begin{equation}
    \widehat{\mathcal S}_D
    \;:=\;
    \bigl\{\,P_i:\ \widehat{\chi}^{(1)}_{ii}\ge\eta/2\,\bigr\}
\end{equation}
includes all (and only) dissipator terms with $a_{ii}\ge \eta$. Similarly, according to \cref{eq:chi_second_deriv},
\begin{equation}
    \widehat{\mathcal S}_H
    \;:=\;
    \widehat{\mathcal S}_D
    \ \cup\
    \bigl\{\,P_i:\ \widehat{\chi}^{(2)}_{ii}\ge\eta^2\,\bigr\}
\end{equation}
contains all Hamiltonian terms with $|h_i| \ge \eta$, as shown in \cref{fig:structure_learning}~(d). Evidently, $\widehat{\mathcal S}_H$ may contain dissipator terms that are not present in the Hamiltonian. Moreover, it may further include spurious Pauli terms arising from second-order dissipative couplings that induce $\widehat{\chi}^{(2)}_{ii}>\eta^2$. However, the worst-case number of such false positives is bounded by the total number of Pauli-rate estimates returned across all sampled times, i.e., $|\widehat{\mathcal S}_H| \;=\; {\mathcal{O}}\left(M^2\right)$, and therefore remains polynomial in the problem parameters. Moreover, such false positives will be eliminated in the subsequent coefficient-learning stage, which accurately estimates all coefficients corresponding to $\widehat{\mathcal S}_H$ and $\widehat{\mathcal S}_D$ and filters out any below the desired threshold~$\eta$.

\section{Coefficient learning}\label{sec:coefficient_learning_main}

Having identified candidate Pauli structures
$\widehat{\mathcal S}_H$ and $\widehat{\mathcal S}_D$ that contain
all nonvanishing Hamiltonian and dissipative terms, the remaining
task is to estimate their associated coefficients to a target
accuracy $\varepsilon$.

Our coefficient-learning protocol builds on a standard
linear-response relation: short-time derivatives of expectation
values provide linear access to the Lindbladian coefficients
\cite{haah2024learning,gu2024practical,stilck2024efficient}. For any
observable $O$ and initial state $\rho$,
\begin{equation}
\label{eq:derivative_linear_access}
    \left.\dv{}{t}\expval{O(t)}\right|_{t=0}
    =
    \tr{\rho\,\mathcal L^\dagger(O)}.
\end{equation}
Substituting the Pauli expansion of $\mathcal L^\dagger$ makes
this dependence explicit:
\begin{equation}
\label{eq:coefficient_learning_linear_relation_main}
\begin{aligned}
    &\left.\dv{}{t}\expval{O(t)}\right|_{t=0}
    =
    \sum_{P_i\in\widehat{\mathcal S}_H}
        h_i\,\tr{i[P_i,O]\rho}
    \\
    &\quad+
    \sum_{P_i,P_j\in\widehat{\mathcal S}_D}
        a_{ij}
        \left(
            \tr{P_jOP_i\rho}
            -
            \frac12\tr{\{P_jP_i,O\}\rho}
        \right).
\end{aligned}
\end{equation}
Thus experimentally accessible derivatives are known linear
functions of the unknown coefficients. The question is how to
choose a small collection of probes whose derivatives allow us to
recover the individual coefficients. We work with Pauli probes of the form
\begin{equation}
\label{eq:pauli_input_derivative_main}
    d_{O,Q}
    :=
    2^{-n}\tr{Q\,\mathcal L^\dagger(O)},
    \qquad
    O,Q\in\mathcal P_n .
\end{equation}
These derivatives can be measured experimentally from the two
Pauli-eigenstate preparations
$\rho_{\pm Q}=2^{-n}(I\pm Q)$.

The first difficulty in constructing suitable probes is \emph{local identifiability}: distinct global Lindbladian terms can be locally indistinguishable. Consider, for example, estimating a dissipative coefficient $a_{ij}$ by local tomography, using Pauli preparations and measurements supported on its natural region
$T_{ij}:=\supp{P_i}\cup\supp{P_j}$. A different dissipative pair $(P_u,P_v)$ can be locally indistinguishable from the target whenever
\begin{equation}
\label{eq:local_alias_main}
P_u=P_iR,\qquad
P_v=P_jR,\qquad
\supp{R}\cap T_{ij}=\emptyset .
\end{equation}
Indeed, since $R$ is supported outside $T_{ij}$, every local observable
commutes with it, so $P_vOP_u=P_jOP_i$. Hence the two global terms contribute to the same local
action, so no experiment restricted to $T_{ij}$ can distinguish
them. A minimal example is
\begin{equation}
\label{eq:coefficient_alias_example_main}
    \mathcal L(\rho)
    =
    a(Z_1\rho Z_1-\rho)
    +
    b(Z_1Z_2\rho Z_1Z_2-\rho).
\end{equation}
If we probe only qubit $1$, the two dissipative terms are
indistinguishable; for example,
$d_{X_1,X_1}=-2(a+b)$.
Thus the coefficient $a$ cannot be identified from data supported
only on its natural region.

We resolve such ambiguities using a support-enlargement procedure
adapted from \citet{montana2025efficiently}. For every conflicting
pair of the form \cref{eq:local_alias_main}, we add one
\emph{witness qubit} on which the exterior factor $R$ acts
nontrivially. On the enlarged region, the conflicting term and the
target have different local Pauli factors and are therefore
distinguishable.

The number of required witnesses is controlled by the
\emph{dual interaction graph} of the Lindbladian
(\cref{sec:dual-interaction-graph}). Its vertices are the nonzero
Pauli components of $\mathcal L^\dagger$, and two vertices are
connected whenever their qubit supports overlap, as illustrated in
\cref{fig:coefficient_learning}~(b); let
$\mathfrak d$ denote its maximum degree. Each conflict necessarily overlaps the target component, so the
dual-graph degree $\mathfrak d$ bounds how many such conflicts can
occur. Hence at most $\mathfrak d$ witness qubits are needed, and the size of the enlarged region is bounded as
\[
    |\widetilde T_{ij}|
    \le
    |T_{ij}|+\mathfrak d .
\]
In \cref{eq:coefficient_alias_example_main}, the exterior factor is
$R=Z_2$, so adding qubit $2$ resolves the ambiguity; indeed
$d_{X_1X_2,X_1X_2}=-2a$.
The same idea applies to Hamiltonian coefficient, which is detailed in
\cref{sec:stable_local_reconstruction}.

After support enlargement, every coefficient is identifiable on the enlarged region. The remaining question is how to select a
small number of probes on this enlarged region without performing
full local process tomography. Here we take advantage of the  structure already identified during
structure learning. Consider again a dissipative target $a_{ij}$
on its enlarged region $\widetilde T_{ij}$. For a Pauli observable
$O\in\mathcal P_{\widetilde T_{ij}}$, choose a phaseless Pauli $Q_{ij}(O)$ and a phase
$\omega_{ij}(O)$ such that
$P_jOP_i=\omega_{ij}(O)Q_{ij}(O)$. 
By Pauli orthogonality, a term $P_vOP_u$ contributes only if
$P_vP_u\propto P_jP_i$, so
\begin{equation}
\label{eq:targeted_sign_structure_main}
    \omega_{ij}(O)^*d_{O,Q_{ij}(O)}
    =
    a_{ij}
    +
    \sum_{(u,v)}
        \zeta_{uv}\,s_{S_{uv}}(O),
\end{equation}
where the sum runs over $N_{ij}$ unwanted pairs $(u,v)$. Here, we defined $S_{uv}\propto P_uP_i$ for every unwanted pair $(u,v)$, $s_{S_{uv}}(O)\in\{\pm1\}$ is its commutation sign with $O$,
and $\zeta_{uv}$ is its unknown amplitude. Thus the target $a_{ij}$ contribution is independent of $O$, whereas each
unwanted contribution varies according to a known sign pattern. This allows the unwanted terms to be canceled directly. If a Pauli $G$ anticommutes with $S_{uv}$, then
\[
    s_{S_{uv}}(GO)=-s_{S_{uv}}(O),
\]
so averaging the responses for $O$ and $GO$ cancels the
corresponding unwanted term. We show in
\cref{sec:deterministic_targeted_probes} that one can always choose
$G$ so that at least half of the remaining unwanted terms are
canceled by averaging the responses for $O$ and $GO$. After at most
$s=\lceil\log_2(N_{ij}+1)\rceil$ generators $G_1,\dots,G_s$, averaging over the phaseless probes $O_{\vb b}\propto G_1^{b_1}\cdots G_s^{b_s}O$ cancels all unwanted terms and yields
\begin{equation}
\label{eq:targeted_reconstruction_main}
    a_{ij}
    =
    \frac{1}{2^s}
    \sum_{\vb b\in\{0,1\}^s}
        \omega_{ij}(O_{\vb b})^*
        d_{O_{\vb b},Q_{ij}(O_{\vb b})}.
\end{equation}
The same construction applies to Hamiltonian coefficients; see
\cref{sec:deterministic_targeted_probes}.

The reconstruction weights in
\cref{eq:targeted_reconstruction_main} have an absolute sum of one since $|\omega_{ij}(O_{\vb b})| = 1$, so
derivative-estimation errors are not amplified. Moreover, the
dual-graph degree bounds how many nuisance terms can appear in these
targeted responses. Summed over all candidate coefficients, the
construction requires only
\begin{equation}
\label{eq:number_targeted_probes_main}
    \mathcal O(\widehat M\,\mathfrak d),
    \qquad
    \widehat M
    :=
    |\widehat{\mathcal S}_H|
    +
    |\widehat{\mathcal S}_D|^2,
\end{equation}
distinct Pauli derivatives.

Finally, it remains to estimate these derivatives from finite-time quantum
queries. Similarly to the derivative estimation used for structure
learning (\cref{sec:structure_learning_main_A}), we measure $f_{O,Q}(t) := \tr{O\,e^{t\mathcal L}(2^{-n}Q)}$ at several Gauss--Chebyshev time nodes and recover
$f'_{O,Q}(0)=d_{O,Q}$ by analytically differentiating the fitted interpolant at $t=0$.

The key additional input is again the maximum degree
$\mathfrak d$ of the dual interaction graph. Intuitively,
$\mathfrak d$ quantifies how quickly the support of
$(\mathcal L^\dagger)^\ell(O)$ can grow: each application of
$\mathcal L^\dagger$ only mixes $O$ with Pauli components whose
supports overlap it, and $\mathfrak d$ bounds the number of such
neighbors. For the enlarged probes constructed above, this gives
\begin{equation}
\label{eq:coefficient_derivative_bound_main}
    \left|\dv[\ell]{t}f_{O,Q}(t)\right|
    \le
    \mathcal O\!\left(
        \min\{\mathfrak d^2,M\}
    \right)^\ell
    \ell!,
\end{equation}
as shown in \cref{sec:pauli_derivative_estimation_main}.
Hence $f_{O,Q}(t)$ is well approximated by a low-degree polynomial
for $t\in[0,\tau_{\max}]$ with
$\tau_{\max} =\Theta\!\left(1/\min\{\mathfrak d^2,M\}\right)$. Using Gauss--Chebyshev
interpolation, each required derivative can therefore be estimated
to accuracy $\varepsilon$ with
$\widetilde{\mathcal O}\!\left(
\min\{\mathfrak d^4,M^2\}/\varepsilon^2
\right)$ queries per probe. Applying this for each of the
$\mathcal O(\widehat M\mathfrak d)$ selected derivatives gives the
overall query complexity $\widetilde{\mathcal O}\!\left(
\widehat M\mathfrak d\min\{\mathfrak d^4,M^2\}/\varepsilon^2
\right)$, saturating the standard quantum limit up to polylogarithmic factors.

The derivative-estimation step can optionally be parallelized using
the shadow process tomography method of
\citet{stilck2024efficient}; this yields an improved query complexity
when the candidate Lindbladian locality $k$ is sufficiently small.
Since every selected probe is supported on at most
$k+\mathfrak d$ qubits, all required probe values at a fixed
Gauss--Chebyshev time node can be estimated simultaneously using
$\widetilde{\mathcal O}\left(
9^{k+\mathfrak d}\min\{\mathfrak d^4,M^2\}/\varepsilon^2
\right)$ samples. Full details and proofs are given
in \cref{sec:coefficient_learning}.

\section{Optimality of time resolution}\label{sec:optimality_resolution_time}

The minimum evolution time required by our structure-learning algorithm is $t_{\min}=\widetilde{\Theta}(M^{-1})$. A natural question is whether one can work with larger minimum times (i.e., a coarser time resolution), ideally approaching a constant, which would be more experimentally convenient~\citep{shin2026heisenberg}. Unfortunately, we show that this is impossible even for the simpler task of learning only the dissipator support $\DiagDissStruct$. In particular, for dissipator-structure learning, enforcing a time resolution $t \gtrsim \Theta(M^{2/\kappa-1})$ (for any fixed integer $\kappa\ge2$) can incur an exponential-in-$n$ sample-complexity overhead (see \cref{sec:appendix-optimality-of-time-resoltion} for more details).

To show the $\Omega(\exp(n))$ sample-complexity lower bound for dissipator-structure learning when
$t \gtrsim \Theta(M^{2/\kappa-1})$, we first specify the access model. We consider protocols that perform experiments of the following form:
(i) prepare an arbitrary single-copy $n$-qubit input state $\rho_{\mathrm{in}}$,
(ii) evolve it under the unknown Lindbladian $\mathcal L$ for some time $t\ge t_0 = \Theta(M^{2/\kappa-1})$ (where $t_0$ is the assumed time resolution),
and (iii) perform an arbitrary single-copy measurement on the output state
$\rho_{\mathrm{out}} = e^{t\mathcal L}(\rho_{\mathrm{in}})$. This access model strictly strengthens the experimental capabilities assumed in
\Cref{theorem:ancilla_free_structure_learning} (product-state preparations and single-qubit Pauli measurements).
Hence, the lower bound proved here also holds for our structure-learning setting.

The high-level idea behind our lower bound is as follows. We construct two purely dissipative
Lindbladians, $\mathcal{L}_{0,\kappa}$ and $\mathcal{L}_{1,\kappa}$, with distinct dissipator supports.
Each instance has $M=\Theta(n^\kappa)$ nonzero terms of magnitude at least $\eta=1/3$.
Since any procedure that learns the dissipator support would, in particular, distinguish between these
two cases, it suffices to lower bound the sample complexity of the corresponding Lindbladian
distinguishing task. A key feature of our construction is rapid mixing. For any input state $\rho_{\mathrm{in}}$ and any
$t \ge t_0=\Theta(M^{2/\kappa-1})$, the evolution under either hypothesis drives the system
exponentially close to the maximally mixed state. Specifically, for all $\rho_{\mathrm{in}}$ and both
$\mathcal{L}\in\{\mathcal{L}_{0,\kappa},\mathcal{L}_{1,\kappa}\}$,
\begin{equation}
    \left\|e^{t\mathcal{L}}(\rho_{\mathrm{in}})-\frac{I}{2^n}\right\|_1 \le 2^{-n}.
\end{equation}
By the triangle inequality, the two hypotheses therefore produce output states whose trace distance is
at most $\delta_{\mathrm{tr}}\coloneq \mathcal O(2^{-n})$, regardless of the chosen inputs. A sequential hybrid argument over the adaptive classical transcript shows that any $m$-sample protocol can achieve a distinguishing advantage of at most $\mathcal{O}(m\delta_{\mathrm{tr}})$; equivalently, constant success probability
requires $m=\Omega(1/\delta_{\mathrm{tr}})=\Omega(2^n)$ samples. Consequently, at time resolution $t_0 \ge \Theta(M^{2/\kappa-1})$ (for any fixed integer $\kappa\ge 2$), there exists a
family of Lindbladians with $M=\Theta(n^\kappa)$ nonzero dissipator terms for which dissipator-structure learning
requires an exponential number of samples.

\section{Conclusion and Outlook}
In this work, we presented an ansatz-free Lindbladian learning algorithm that assumes no prior knowledge about the generator's locality or structure. Our protocol efficiently reconstructs both the Hamiltonian and dissipative components of an $M$-sparse Lindbladian in an experimentally motivated, ancilla-free setting restricted to product-state preparations, time evolution, and Pauli-basis measurements. The coefficient reconstruction is provably stable and does not depend on the conditioning of a linear system. Together, these results provide a practical route to identifying unknown error mechanisms and controlling non-idealities, with potential relevance to tailored error correction and mitigation, noisy quantum metrology, and analog simulation.

Looking ahead, we view the minimum evolution time as a key obstacle to the experimental realization of Lindbladian learning. While our results show that a minimum evolution time scaling as $1/M$ is optimal in the general sparse setting, substantially longer evolution times are possible under additional structure, such as bounded degree or locality~\cite{lewis2026learning,arad2026near,mobus2026robust,chen2026learning}. Although recent works have made significant progress toward improving this time resolution, the available guarantees can still incur unfavorable dependence on interaction strength, locality, or large constant prefactors. Developing Lindbladian-learning algorithms with experimentally accessible minimum evolution times therefore remains an important open problem.

\begin{acknowledgments}
We are thankful for insightful discussions with
John Martyn, Ishaan Kannan, Andrei C. Diaconu, and Sabina Dragoi.
We especially thank Jordi Montañà-López for pointing us to the support-enlargement construction of \citet{montana2025efficiently}, which we adapted to our coefficient-learning procedure.
We acknowledge the support from DOE through
the QUACQ program (DE-SC0025572). N.R. acknowledges support from the NTT research fellowship. W.G. acknowledges support from NSF Grant CCF-2430375 and the Von Neumann Award from Harvard Computer Science.

We acknowledge the use of large language models in preparing this work. The structure learning and the original coefficient learning protocols were conceived and developed by the authors. During the revision of the manuscript, GPT-5.6 Sol helped develop the probe-selection idea used in the revised coefficient-learning stage. Different models assisted with writing and proofreading. The authors reviewed and revised all AI-assisted material and take full responsibility for the manuscript.
\end{acknowledgments}

\bibliography{main}

@book{Watrous_2018, 
    place={Cambridge}, 
    title={The Theory of Quantum Information}, 
    publisher={Cambridge University Press}, 
    author={Watrous, John}, 
    year={2018}
}

@article{lindblad1976generators,
  title={On the generators of quantum dynamical semigroups},
  author={Lindblad, Goran},
  journal={Communications in mathematical physics},
  volume={48},
  pages={119--130},
  year={1976},
  publisher={Springer},
  doi = {10.1007/BF01608499}
}

@article{gorini1976completely,
  title={Completely positive dynamical semigroups of N-level systems},
  author={Gorini, Vittorio and Kossakowski, Andrzej and Sudarshan, Ennackal Chandy George},
  journal={Journal of Mathematical Physics},
  volume={17},
  number={5},
  pages={821--825},
  year={1976},
  publisher={American Institute of Physics},
  doi = {10.1063/1.522979},
  url = {https://doi.org/10.1063/1.522979}
}

@article{baumgartner2008analysis,
doi = {10.1088/1751-8113/41/39/395303},
url = {https://doi.org/10.1088/1751-8113/41/39/395303},
year = {2008},
month = {sep},
publisher = {},
volume = {41},
number = {39},
pages = {395303},
author = {Baumgartner, Bernhard and Narnhofer, Heide},
title = {Analysis of quantum semigroups with GKS–Lindblad generators: {II}. General},
journal = {Journal of Physics A: Mathematical and Theoretical}
}

@article{temme2010chi,
  title={The $\chi^2$-divergence and mixing times of quantum Markov processes},
  author={Temme, Kristan and Kastoryano, Michael James and Ruskai, Mary Beth and Wolf, Michael Marc and Verstraete, Frank},
  journal={Journal of Mathematical Physics},
  volume={51},
  number={12},
  year={2010},
  publisher={AIP Publishing},
  doi={10.1063/1.3511335}
}

@article{boulant2003robust,
  title={Robust method for estimating the Lindblad operators of a dissipative quantum process from measurements of the density operator at multiple time points},
  author={Boulant, Nicolas and Havel, Timothy F and Pravia, Marco A and Cory, David G},
  journal={Physical Review A},
  volume={67},
  number={4},
  pages={042322},
  year={2003},
  publisher={APS},
  doi = {10.1103/PhysRevA.67.042322},
}

@article{bairey2020learning,
  title={Learning the dynamics of open quantum systems from their steady states},
  author={Bairey, Eyal and Guo, Chu and Poletti, Dario and Lindner, Netanel H and Arad, Itai},
  journal={New Journal of Physics},
  volume={22},
  number={3},
  pages={032001},
  year={2020},
  publisher={IOP Publishing},
  doi={10.1088/1367-2630/ab73cd}
}

@article{samach2022lindblad,
  title={Lindblad tomography of a superconducting quantum processor},
  author={Samach, Gabriel O and Greene, Ami and Borregaard, Johannes and Christandl, Matthias and Barreto, Joseph and Kim, David K and McNally, Christopher M and Melville, Alexander and Niedzielski, Bethany M and Sung, Youngkyu and others},
  journal={Physical Review Applied},
  volume={18},
  number={6},
  pages={064056},
  year={2022},
  publisher={APS},
  doi = {10.1103/PhysRevApplied.18.064056},
}

@article{dobrynin2024compressed,
  title={Compressed-sensing Lindbladian quantum tomography with trapped ions},
  author={Dobrynin, Dmitrii and Cardarelli, Lorenzo and M{\"u}ller, Markus and Bermudez, Alejandro},
  journal = {Quantum Science and Technology},
  year = {2025},
  month = {sep},
  publisher = {IOP Publishing},
  volume = {10},
  number = {4},
  pages = {045041},
  doi = {10.1088/2058-9565/ae0363},
  url = {https://doi.org/10.1088/2058-9565/ae0363},
}

@article{stilck2024efficient,
  title={Efficient and robust estimation of many-qubit Hamiltonians},
  author={Stilck Fran{\c{c}}a, Daniel and Markovich, Liubov A and Dobrovitski, VV and Werner, Albert H and Borregaard, Johannes},
  journal={Nature Communications},
  volume={15},
  number={1},
  pages={311},
  year={2024},
  publisher={Nature Publishing Group UK London},
  doi={10.1038/s41467-023-44012-5}
}

@misc{franca2025learning,
      title={Learning and certification of local time-dependent quantum dynamics and noise}, 
      author={Daniel Stilck França and Tim Möbus and Cambyse Rouzé and Albert H. Werner},
      year={2025},
      eprint={2510.08500},
      archivePrefix={arXiv},
      primaryClass={quant-ph},
      url={https://arxiv.org/abs/2510.08500}, 
}

@article{olsacher2025hamiltonian,
  title={Hamiltonian and Liouvillian learning in weakly-dissipative quantum many-body systems},
  author={Olsacher, Tobias and Kraft, Tristan and Kokail, Christian and Kraus, Barbara and Zoller, Peter},
  journal={Quantum Science and Technology},
  volume={10},
  number={1},
  pages={015065},
  year={2025},
  publisher={IOP Publishing},
  doi={10.1088/2058-9565/ad9ed5}
}

@article{pastori2022characterization,
  title={Characterization and verification of Trotterized digital quantum simulation via Hamiltonian and Liouvillian learning},
  author={Pastori, Lorenzo and Olsacher, Tobias and Kokail, Christian and Zoller, Peter},
  journal={PRX Quantum},
  volume={3},
  number={3},
  pages={030324},
  year={2022},
  publisher={APS},
  doi={10.1103/PRXQuantum.3.030324}
}

@misc{berg2025large,
      title={Large-scale Lindblad learning from time-series data}, 
      author={Ewout van den Berg and Brad Mitchell and Ken Xuan Wei and Moein Malekakhlagh},
      year={2025},
      eprint={2512.08165},
      archivePrefix={arXiv},
      primaryClass={quant-ph},
      url={https://arxiv.org/abs/2512.08165}, 
}

@misc{kraft2025bounded,
      title={Bounded-Error Quantum Simulation via Hamiltonian and Lindbladian Learning}, 
      author={Tristan Kraft and Manoj K. Joshi and William Lam and Tobias Olsacher and Florian Kranzl and Johannes Franke and Lata Kh Joshi and Rainer Blatt and Augusto Smerzi and Daniel Stilck França and Benoît Vermersch and Barbara Kraus and Christian F. Roos and Peter Zoller},
      year={2025},
      eprint={2511.23392},
      archivePrefix={arXiv},
      primaryClass={quant-ph},
      url={https://arxiv.org/abs/2511.23392}, 
}

@article{montana2025efficiently,
  title={Efficiently learning non-Markovian noise in many-body quantum simulators},
  author={Monta{\~n}{\`a}-L{\'o}pez, Jordi A and Elben, Andreas and Choi, Joonhee and Trivedi, Rahul},
  journal={arXiv preprint arXiv:2511.16772},
  year={2025}
}

@misc{arad2026near,
      title={Near-Optimal Learning of Local Lindbladians}, 
      author={Itai Arad and Zhili Chen and Naixu Guo and Patrick Rebentrost and Zhan Yu},
      year={2026},
      eprint={2606.20535},
      archivePrefix={arXiv},
      primaryClass={quant-ph},
      url={https://arxiv.org/abs/2606.20535}, 
}

@misc{mobus2026robust,
      title={Robust Structure Learning of $k$-local Lindbladians}, 
      author={Tim Möbus and Thiago Bergamaschi and Daniel Stilck França and Cambyse Rouzé},
      year={2026},
      eprint={2606.23652},
      archivePrefix={arXiv},
      primaryClass={quant-ph},
      url={https://arxiv.org/abs/2606.23652}, 
}

@misc{lewis2026learning,
      title={Learning the structure of open quantum systems}, 
      author={Laura Lewis and Ewin Tang and John Wright},
      year={2026},
      eprint={2606.30358},
      archivePrefix={arXiv},
      primaryClass={quant-ph},
      url={https://arxiv.org/abs/2606.30358}, 
}

@misc{zhou2026characterizing,
      title={Characterizing Arbitrary Lindbladian Dynamics with a Few Pauli Measurements}, 
      author={Taiqi Zhou and Weiyuan Gong},
      year={2026},
      eprint={2607.23044},
      archivePrefix={arXiv},
      primaryClass={quant-ph},
      url={https://arxiv.org/abs/2607.23044}, 
}

@misc{chen2026learning,
      title={Learning Arbitrary Lindbladians from Time Evolution}, 
      author={Zhili Chen and Zhan Yu},
      year={2026},
      eprint={2607.28610},
      archivePrefix={arXiv},
      primaryClass={quant-ph},
      url={https://arxiv.org/abs/2607.28610}, 
}

@misc{sinha2026lindblad,
  doi = {10.7907/HPRR-FQ68},
  url = {https://resolver.caltech.edu/CaltechTHESIS:04142026-214248334},
  author = {Sinha, Savar Dayal},
  title = {Efficient and SPAM-Robust Ansatz-Free Lindbladian Learning},
  publisher = {California Institute of Technology},
  year = {2026},
  copyright = {No commercial reproduction, distribution, display or performance rights in this work are provided.}
}

@article{shulman2014suppressing,
  title={Suppressing qubit dephasing using real-time Hamiltonian estimation},
  author={Shulman, Michael D and Harvey, Shannon P and Nichol, John M and Bartlett, Stephen D and Doherty, Andrew C and Umansky, Vladimir and Yacoby, Amir},
  journal={Nature communications},
  volume={5},
  number={1},
  pages={5156},
  year={2014},
  publisher={Nature Publishing Group UK London},
  doi = {10.1038/ncomms6156}
}

@article{hangleiter2024robustly,
  title={Robustly learning the Hamiltonian dynamics of a superconducting quantum processor},
  author={Hangleiter, Dominik and Roth, Ingo and Fuksa, Jon{\'a}{\v{s}} and Eisert, Jens and Roushan, Pedram},
  journal={Nature Communications},
  volume={15},
  number={1},
  pages={9595},
  year={2024},
  publisher={Nature Publishing Group UK London},
  doi = {10.1038/s41467-024-52629-3}
}

@article{valenti2019hamiltonian,
  title = {Hamiltonian learning for quantum error correction},
  author = {Valenti, Agnes and van Nieuwenburg, Evert and Huber, Sebastian and Greplova, Eliska},
  journal = {Phys. Rev. Res.},
  volume = {1},
  issue = {3},
  pages = {033092},
  numpages = {14},
  year = {2019},
  month = {Nov},
  publisher = {American Physical Society},
  doi = {10.1103/PhysRevResearch.1.033092},
  url = {https://link.aps.org/doi/10.1103/PhysRevResearch.1.033092}
}

@article{innocenti2020supervised,
  title={Supervised learning of time-independent Hamiltonians for gate design},
  author={Innocenti, Luca and Banchi, Leonardo and Ferraro, Alessandro and Bose, Sougato and Paternostro, Mauro},
  journal={New Journal of Physics},
  volume={22},
  number={6},
  pages={065001},
  year={2020},
  publisher={IOP Publishing},
  doi = {10.1088/1367-2630/ab8aaf},
}

@article{hu2025ansatz,
  title = {Ansatz-Free Hamiltonian Learning with Heisenberg-Limited Scaling},
  author = {Hu, Hong-Ye and Ma, Muzhou and Gong, Weiyuan and Ye, Qi and Tong, Yu and Flammia, Steven T. and Yelin, Susanne F.},
  journal = {PRX Quantum},
  volume = {6},
  issue = {4},
  pages = {040315},
  numpages = {30},
  year = {2025},
  month = {Oct},
  publisher = {American Physical Society},
  doi = {10.1103/j7b8-pb77},
  url = {https://link.aps.org/doi/10.1103/j7b8-pb77}
}

@article{dutkiewicz2024advantage,
  title={The advantage of quantum control in many-body Hamiltonian learning},
  author={Dutkiewicz, Alicja and O'Brien, Thomas E and Schuster, Thomas},
  journal={Quantum},
  volume={8},
  pages={1537},
  year={2024},
  publisher={Verein zur F{\"o}rderung des Open Access Publizierens in den Quantenwissenschaften},
  doi= {10.22331/q-2024-11-26-1537}
}

@article{anshu2021sample,
  title={Sample-efficient learning of interacting quantum systems},
  author={Anshu, Anurag and Arunachalam, Srinivasan and Kuwahara, Tomotaka and Soleimanifar, Mehdi},
  journal={Nature Physics},
  volume={17},
  number={8},
  pages={931--935},
  year={2021},
  publisher={Nature Publishing Group UK London},
  doi={10.1038/s41567-021-01232-0}
}

@inproceedings{bakshi2024structure,
  title={Structure learning of Hamiltonians from real-time evolution},
  author={Bakshi, Ainesh and Liu, Allen and Moitra, Ankur and Tang, Ewin},
  booktitle={2024 IEEE 65th Annual Symposium on Foundations of Computer Science (FOCS)},
  pages={1037--1050},
  year={2024},
  organization={IEEE},
  doi={10.1109/FOCS61266.2024.00069}
}

@article{caro2024learning,
  title={Learning quantum processes and Hamiltonians via the Pauli transfer matrix},
  author={Caro, Matthias C},
  journal={ACM Transactions on Quantum Computing},
  volume={5},
  number={2},
  pages={1--53},
  year={2024},
  publisher={ACM New York, NY}
}

@article{gu2024practical,
  title={Practical Hamiltonian learning with unitary dynamics and Gibbs states},
  author={Gu, Andi and Cincio, Lukasz and Coles, Patrick J},
  journal={Nature Communications},
  volume={15},
  number={1},
  pages={312},
  year={2024},
  publisher={Nature Publishing Group UK London}, 
  doi={10.1038/s41467-023-44008-1}
}

@inproceedings{zhao2025learning,
author = {Zhao, Andrew},
title = {Learning the Structure of Any Hamiltonian from Minimal Assumptions},
year = {2025},
isbn = {9798400715105},
publisher = {Association for Computing Machinery},
address = {New York, NY, USA},
url = {https://doi.org/10.1145/3717823.3718115},
doi = {10.1145/3717823.3718115},
booktitle = {Proceedings of the 57th Annual ACM Symposium on Theory of Computing},
pages = {1201–1211},
numpages = {11},
location = {Prague, Czechia},
series = {STOC '25}
}

@misc{sinha2025improvedhamiltonianlearningsparsity,
      title={Improved Hamiltonian learning and sparsity testing through Bell sampling}, 
      author={Savar D. Sinha and Yu Tong},
      year={2025},
      eprint={2509.07937},
      archivePrefix={arXiv},
      primaryClass={quant-ph},
      url={https://arxiv.org/abs/2509.07937}, 
}

@article{haah2024learning,
  title={Learning quantum Hamiltonians from high-temperature Gibbs states and real-time evolutions},
  author={Haah, Jeongwan and Kothari, Robin and Tang, Ewin},
  journal={Nature Physics},
  volume={20},
  number={6},
  pages={1027--1031},
  year={2024},
  publisher={Nature Publishing Group UK London},
  doi={10.1038/s41567-023-02376-x}
}

@article{huang2023learning,
  title = {Learning Many-Body Hamiltonians with Heisenberg-Limited Scaling},
  author = {Huang, Hsin-Yuan and Tong, Yu and Fang, Di and Su, Yuan},
  journal = {Phys. Rev. Lett.},
  volume = {130},
  issue = {20},
  pages = {200403},
  numpages = {7},
  year = {2023},
  month = {May},
  publisher = {American Physical Society},
  doi = {10.1103/PhysRevLett.130.200403},
  url = {https://link.aps.org/doi/10.1103/PhysRevLett.130.200403}
}

@misc{ma2024learning,
      title={Learning $k$-body Hamiltonians via compressed sensing}, 
      author={Muzhou Ma and Steven T. Flammia and John Preskill and Yu Tong},
      year={2024},
      eprint={2410.18928},
      archivePrefix={arXiv},
      primaryClass={quant-ph},
      url={https://arxiv.org/abs/2410.18928}, 
}

@article{mirani2024learning,
  title={Learning interacting fermionic Hamiltonians at the Heisenberg limit},
  author={Mirani, Arjun and Hayden, Patrick},
  journal={Physical Review A},
  volume={110},
  number={6},
  pages={062421},
  year={2024},
  publisher={APS},
  doi={10.1103/PhysRevA.110.062421}
}

@article{ni2024quantum,
  title={Quantum hamiltonian learning for the fermi-hubbard model},
  author={Ni, Hongkang and Li, Haoya and Ying, Lexing},
  journal={Acta Applicandae Mathematicae},
  volume={191},
  number={1},
  pages={2},
  year={2024},
  publisher={Springer},
  doi={10.1007/s10440-024-00651-4}
}

@misc{zubida2021optimal,
      title={Optimal short-time measurements for Hamiltonian learning}, 
      author={Assaf Zubida and Elad Yitzhaki and Netanel H. Lindner and Eyal Bairey},
      year={2021},
      eprint={2108.08824},
      archivePrefix={arXiv},
      primaryClass={quant-ph},
      url={https://arxiv.org/abs/2108.08824}, 
}

@misc{evans2019scalable,
      title={Scalable Bayesian Hamiltonian learning}, 
      author={Tim J. Evans and Robin Harper and Steven T. Flammia},
      year={2019},
      eprint={1912.07636},
      archivePrefix={arXiv},
      primaryClass={quant-ph},
      url={https://arxiv.org/abs/1912.07636}, 
}

@article{bairey2019learning,
  title={Learning a local Hamiltonian from local measurements},
  author={Bairey, Eyal and Arad, Itai and Lindner, Netanel H},
  journal={Physical review letters},
  volume={122},
  number={2},
  pages={020504},
  year={2019},
  publisher={APS},
  doi={10.1103/PhysRevLett.122.020504}
}

@article{qi2019determining,
  title={Determining a local Hamiltonian from a single eigenstate},
  author={Qi, Xiao-Liang and Ranard, Daniel},
  journal={Quantum},
  volume={3},
  pages={159},
  year={2019},
  publisher={Verein zur F{\"o}rderung des Open Access Publizierens in den Quantenwissenschaften},
  doi={10.22331/q-2019-07-08-159}
}

@article{rouze2024learning,
  title={Learning quantum many-body systems from a few copies},
  author={Rouz{\'e}, Cambyse and Fran{\c{c}}a, Daniel Stilck},
  journal={Quantum},
  volume={8},
  pages={1319},
  year={2024},
  publisher={Verein zur F{\"o}rderung des Open Access Publizierens in den Quantenwissenschaften},
  doi={10.22331/q-2024-04-30-1319}
}

@inproceedings{bakshi2024learning,
author = {Bakshi, Ainesh and Liu, Allen and Moitra, Ankur and Tang, Ewin},
title = {Learning Quantum Hamiltonians at Any Temperature in Polynomial Time},
year = {2024},
isbn = {9798400703836},
publisher = {Association for Computing Machinery},
address = {New York, NY, USA},
url = {https://doi.org/10.1145/3618260.3649619},
doi = {10.1145/3618260.3649619},
booktitle = {Proceedings of the 56th Annual ACM Symposium on Theory of Computing},
pages = {1470–1477},
numpages = {8},
location = {Vancouver, BC, Canada},
series = {STOC 2024}
}

@article{liu2025optimal,
  title={Optimal and Robust In-situ Quantum Hamiltonian Learning through Parallelization},
  author={Liu, Suying and Wu, Xiaodi and Niu, Murphy Yuezhen},
  journal={arXiv preprint arXiv:2510.07818},
  year={2025},
  url={https://arxiv.org/abs/2510.07818}
}

@article{shin2026heisenberg,
  title={Heisenberg-limited Hamiltonian learning without short-time control},
  author={Shin, Myeongjin and Lee, Junseo and Oh, Changhun},
  journal={arXiv preprint arXiv:2604.27838},
  year={2026}
}

@article{huelga1997improvement,
  title={Improvement of frequency standards with quantum entanglement},
  author={Huelga, Susanna F and Macchiavello, Chiara and Pellizzari, Thomas and Ekert, Artur K and Plenio, Martin B and Cirac, J Ignacio},
  journal={Physical Review Letters},
  volume={79},
  number={20},
  pages={3865},
  year={1997},
  publisher={APS},
  doi={10.1103/PhysRevLett.79.3865}
}

@article{escher2011general,
  title={General framework for estimating the ultimate precision limit in noisy quantum-enhanced metrology},
  author={Escher, BM and de Matos Filho, Ruynet Lima and Davidovich, Luiz},
  journal={Nature Physics},
  volume={7},
  number={5},
  pages={406--411},
  year={2011},
  publisher={Nature Publishing Group UK London},
  doi = {10.1038/nphys1958}
}

@article{demkowicz2014using,
  title={Using entanglement against noise in quantum metrology},
  author={Demkowicz-Dobrza{\'n}ski, Rafal and Maccone, Lorenzo},
  journal={Physical review letters},
  volume={113},
  number={25},
  pages={250801},
  year={2014},
  publisher={APS},
  doi={10.1103/PhysRevLett.113.250801}
}

@article{zhou2018achieving,
  title={Achieving the Heisenberg limit in quantum metrology using quantum error correction},
  author={Zhou, Sisi and Zhang, Mengzhen and Preskill, John and Jiang, Liang},
  journal={Nature communications},
  volume={9},
  number={1},
  pages={78},
  year={2018},
  publisher={Nature Publishing Group UK London},
  doi={10.1038/s41467-017-02510-3}
}

@book{mason2002chebyshev,
  title={Chebyshev polynomials},
  author={Mason, John C and Handscomb, David C},
  year={2002},
  publisher={Chapman and Hall/CRC}
}

@book{trefethen2019approximation,
author = {Trefethen, Lloyd N.},
title = {Approximation Theory and Approximation Practice, Extended Edition},
publisher = {Society for Industrial and Applied Mathematics},
year = {2019},
doi = {10.1137/1.9781611975949},
address = {Philadelphia, PA},
edition   = {},
URL = {https://epubs.siam.org/doi/abs/10.1137/1.9781611975949},
}

@book{boyd2001chebyshev,
  title={Chebyshev and Fourier spectral methods},
  author={Boyd, John P},
  year={2001},
  publisher={Courier Corporation}
}

@misc{marshall2024chebyshev,
  title={Chebyshev interpolation},
  author={Marshall, NICHOLAS F},
  year={2024}
}

@article{chen2024tight,
  title={Tight bounds on Pauli channel learning without entanglement},
  author={Chen, Senrui and Oh, Changhun and Zhou, Sisi and Huang, Hsin-Yuan and Jiang, Liang},
  journal={Physical Review Letters},
  volume={132},
  number={18},
  pages={180805},
  year={2024},
  publisher={APS},
  doi = {10.1103/PhysRevLett.132.180805}
}

@article{chen2022quantum,
  title={Quantum advantages for Pauli channel estimation},
  author={Chen, Senrui and Zhou, Sisi and Seif, Alireza and Jiang, Liang},
  journal={Physical Review A},
  volume={105},
  number={3},
  pages={032435},
  year={2022},
  publisher={APS},
  doi={10.1103/PhysRevA.105.032435}
}

@article{flammia2021pauli,
  title={Pauli error estimation via population recovery},
  author={Flammia, Steven T and O'Donnell, Ryan},
  journal={Quantum},
  volume={5},
  pages={549},
  year={2021},
  publisher={Verein zur F{\"o}rderung des Open Access Publizierens in den Quantenwissenschaften},
  doi={10.22331/q-2021-09-23-549}
}

@misc{o2025spam,
      title={SPAM Tolerance for Pauli Error Estimation}, 
      author={Ryan O'Donnell and Samvitti Sharma},
      year={2025},
      eprint={2510.00230},
      archivePrefix={arXiv},
      primaryClass={quant-ph},
      url={https://arxiv.org/abs/2510.00230}, 
}

@article{temme2017error,
  title = {Error Mitigation for Short-Depth Quantum Circuits},
  author = {Temme, Kristan and Bravyi, Sergey and Gambetta, Jay M.},
  journal = {Phys. Rev. Lett.},
  volume = {119},
  issue = {18},
  pages = {180509},
  numpages = {5},
  year = {2017},
  month = {Nov},
  publisher = {American Physical Society},
  doi = {10.1103/PhysRevLett.119.180509},
  url = {https://link.aps.org/doi/10.1103/PhysRevLett.119.180509}
}

@article{bravyi2021mitigating,
  title = {Mitigating measurement errors in multiqubit experiments},
  author = {Bravyi, Sergey and Sheldon, Sarah and Kandala, Abhinav and Mckay, David C. and Gambetta, Jay M.},
  journal = {Phys. Rev. A},
  volume = {103},
  issue = {4},
  pages = {042605},
  numpages = {12},
  year = {2021},
  month = {Apr},
  publisher = {American Physical Society},
  doi = {10.1103/PhysRevA.103.042605},
  url = {https://link.aps.org/doi/10.1103/PhysRevA.103.042605}
}

@article{van2023probabilistic,
  title={Probabilistic error cancellation with sparse Pauli--Lindblad models on noisy quantum processors},
  author={Van Den Berg, Ewout and Minev, Zlatko K and Kandala, Abhinav and Temme, Kristan},
  journal={Nature physics},
  volume={19},
  number={8},
  pages={1116--1121},
  year={2023},
  publisher={Nature Publishing Group UK London},
  doi = {10.1038/s41567-023-02042-2}
}

@article{cai2021multi,
  title={Multi-exponential error extrapolation and combining error mitigation techniques for NISQ applications},
  author={Cai, Zhenyu},
  journal={npj Quantum Information},
  volume={7},
  number={1},
  pages={80},
  year={2021},
  publisher={Nature Publishing Group UK London}, 
  doi = {10.1038/s41534-021-00404-3}
}

@article{kim2023evidence,
  title={Evidence for the utility of quantum computing before fault tolerance},
  author={Kim, Youngseok and Eddins, Andrew and Anand, Sajant and Wei, Ken Xuan and Van Den Berg, Ewout and Rosenblatt, Sami and Nayfeh, Hasan and Wu, Yantao and Zaletel, Michael and Temme, Kristan and others},
  journal={Nature},
  volume={618},
  number={7965},
  pages={500--505},
  year={2023},
  publisher={Nature Publishing Group UK London},
  doi={10.1038/s41586-023-06096-3}
}

@article{harper2021fast,
  title = {Fast Estimation of Sparse Quantum Noise},
  author = {Harper, Robin and Yu, Wenjun and Flammia, Steven T.},
  journal = {PRX Quantum},
  volume = {2},
  issue = {1},
  pages = {010322},
  numpages = {26},
  year = {2021},
  month = {Feb},
  publisher = {American Physical Society},
  doi = {10.1103/PRXQuantum.2.010322},
  url = {https://link.aps.org/doi/10.1103/PRXQuantum.2.010322}
}

@article{strikis2021learning,
  title = {Learning-Based Quantum Error Mitigation},
  author = {Strikis, Armands and Qin, Dayue and Chen, Yanzhu and Benjamin, Simon C. and Li, Ying},
  journal = {PRX Quantum},
  volume = {2},
  issue = {4},
  pages = {040330},
  numpages = {30},
  year = {2021},
  month = {Nov},
  publisher = {American Physical Society},
  doi = {10.1103/PRXQuantum.2.040330},
  url = {https://link.aps.org/doi/10.1103/PRXQuantum.2.040330}
}

@article{childs2021theory,
  title={Theory of trotter error with commutator scaling},
  author={Childs, Andrew M and Su, Yuan and Tran, Minh C and Wiebe, Nathan and Zhu, Shuchen},
  journal={Physical Review X},
  volume={11},
  number={1},
  pages={011020},
  year={2021},
  publisher={APS},
  doi={PhysRevX.11.011020}
}

@article{viola1999dynamical,
  title={Dynamical decoupling of open quantum systems},
  author={Viola, Lorenza and Knill, Emanuel and Lloyd, Seth},
  journal={Physical Review Letters},
  volume={82},
  number={12},
  pages={2417},
  year={1999},
  publisher={APS},
  doi={10.1103/PhysRevLett.82.2417}
}

@misc{cotler2026noisylearning,
      title={Noisy Quantum Learning Theory}, 
      author={Jordan Cotler and Weiyuan Gong and Ishaan Kannan},
      year={2026},
      eprint={2512.10929},
      archivePrefix={arXiv},
      primaryClass={quant-ph},
      url={https://arxiv.org/abs/2512.10929}, 
}

@misc{gong2026multiparameter,
      title={Robust multiparameter estimation using quantum scrambling}, 
      author={Wenjie Gong and Bingtian Ye and Daniel Mark and Soonwon Choi},
      year={2026},
      eprint={2601.23283},
      archivePrefix={arXiv},
      primaryClass={quant-ph},
      url={https://arxiv.org/abs/2601.23283}, 
}

@article{emerson2005scalable,
doi = {10.1088/1464-4266/7/10/021},
url = {https://doi.org/10.1088/1464-4266/7/10/021},
year = {2005},
month = {sep},
publisher = {},
volume = {7},
number = {10},
pages = {S347},
author = {Emerson, Joseph and Alicki, Robert and Życzkowski, Karol},
title = {Scalable noise estimation with random unitary operators},
journal = {Journal of Optics B: Quantum and Semiclassical Optics},
}

@article{knill2008randomized,
  title = {Randomized benchmarking of quantum gates},
  author = {Knill, E. and Leibfried, D. and Reichle, R. and Britton, J. and Blakestad, R. B. and Jost, J. D. and Langer, C. and Ozeri, R. and Seidelin, S. and Wineland, D. J.},
  journal = {Phys. Rev. A},
  volume = {77},
  issue = {1},
  pages = {012307},
  numpages = {7},
  year = {2008},
  month = {Jan},
  publisher = {American Physical Society},
  doi = {10.1103/PhysRevA.77.012307},
  url = {https://link.aps.org/doi/10.1103/PhysRevA.77.012307}
}

@article{dankert2009exact,
  title = {Exact and approximate unitary 2-designs and their application to fidelity estimation},
  author = {Dankert, Christoph and Cleve, Richard and Emerson, Joseph and Livine, Etera},
  journal = {Phys. Rev. A},
  volume = {80},
  issue = {1},
  pages = {012304},
  numpages = {6},
  year = {2009},
  month = {Jul},
  publisher = {American Physical Society},
  doi = {10.1103/PhysRevA.80.012304},
  url = {https://link.aps.org/doi/10.1103/PhysRevA.80.012304}
}

@article{flammia2011direct,
  title = {Direct Fidelity Estimation from Few Pauli Measurements},
  author = {Flammia, Steven T. and Liu, Yi-Kai},
  journal = {Phys. Rev. Lett.},
  volume = {106},
  issue = {23},
  pages = {230501},
  numpages = {4},
  year = {2011},
  month = {Jun},
  publisher = {American Physical Society},
  doi = {10.1103/PhysRevLett.106.230501},
  url = {https://link.aps.org/doi/10.1103/PhysRevLett.106.230501}
}

@article{da2011practical,
  title = {Practical Characterization of Quantum Devices without Tomography},
  author = {da Silva, Marcus P. and Landon-Cardinal, Olivier and Poulin, David},
  journal = {Phys. Rev. Lett.},
  volume = {107},
  issue = {21},
  pages = {210404},
  numpages = {5},
  year = {2011},
  month = {Nov},
  publisher = {American Physical Society},
  doi = {10.1103/PhysRevLett.107.210404},
  url = {https://link.aps.org/doi/10.1103/PhysRevLett.107.210404}
}

@article{magesan2011scalable,
  title = {Scalable and Robust Randomized Benchmarking of Quantum Processes},
  author = {Magesan, Easwar and Gambetta, J. M. and Emerson, Joseph},
  journal = {Phys. Rev. Lett.},
  volume = {106},
  issue = {18},
  pages = {180504},
  numpages = {4},
  year = {2011},
  month = {May},
  publisher = {American Physical Society},
  doi = {10.1103/PhysRevLett.106.180504},
  url = {https://link.aps.org/doi/10.1103/PhysRevLett.106.180504}
}

@article{moussa2012practical,
  title = {Practical Experimental Certification of Computational Quantum Gates Using a Twirling Procedure},
  author = {Moussa, Osama and da Silva, Marcus P. and Ryan, Colm A. and Laflamme, Raymond},
  journal = {Phys. Rev. Lett.},
  volume = {109},
  issue = {7},
  pages = {070504},
  numpages = {5},
  year = {2012},
  month = {Aug},
  publisher = {American Physical Society},
  doi = {10.1103/PhysRevLett.109.070504},
  url = {https://link.aps.org/doi/10.1103/PhysRevLett.109.070504}
}

@article{wallman2015estimating,
doi = {10.1088/1367-2630/17/11/113020},
url = {https://doi.org/10.1088/1367-2630/17/11/113020},
year = {2015},
month = {nov},
publisher = {IOP Publishing},
volume = {17},
number = {11},
pages = {113020},
author = {Wallman, Joel and Granade, Chris and Harper, Robin and Flammia, Steven T},
title = {Estimating the coherence of noise},
journal = {New Journal of Physics},
}

@article{wallman2016robust,
doi = {10.1088/1367-2630/18/4/043021},
url = {https://doi.org/10.1088/1367-2630/18/4/043021},
year = {2016},
month = {apr},
publisher = {IOP Publishing},
volume = {18},
number = {4},
pages = {043021},
author = {Wallman, Joel J and Barnhill, Marie and Emerson, Joseph},
title = {Robust characterization of leakage errors},
journal = {New Journal of Physics},
}

@article{boixo2018characterizing,
  title        = {Characterizing quantum supremacy in near-term devices},
  author       = {Boixo, Sergio and Isakov, Sergei V. and Smelyanskiy, Vadim N. and Babbush, Ryan and Ding, Nan and Jiang, Zhang and Bremner, Michael J. and Martinis, John M. and Neven, Hartmut},
  journal      = {Nature Physics},
  volume       = {14},
  number       = {6},
  pages        = {595--600},
  year         = {2018},
  doi          = {10.1038/s41567-018-0124-x},
  url          = {https://doi.org/10.1038/s41567-018-0124-x}
}

@article{erhard2019characterizing,
  title={Characterizing large-scale quantum computers via cycle benchmarking},
  author={Erhard, Alexander and Wallman, Joel J and Postler, Lukas and Meth, Michael and Stricker, Roman and Martinez, Esteban A and Schindler, Philipp and Monz, Thomas and Emerson, Joseph and Blatt, Rainer},
  journal={Nature communications},
  volume={10},
  number={1},
  pages={5347},
  year={2019},
  publisher={Nature Publishing Group UK London},
  doi          = {10.1038/s41467-019-13068-7},
  url          = {https://doi.org/10.1038/s41467-019-13068-7}
}

@article{proctor2022scalable,
  title={Scalable randomized benchmarking of quantum computers using mirror circuits},
  author={Proctor, Timothy and Seritan, Stefan and Rudinger, Kenneth and Nielsen, Erik and Blume-Kohout, Robin and Young, Kevin},
  journal={Physical Review Letters},
  volume={129},
  number={15},
  pages={150502},
  year={2022},
  publisher={APS},
  doi = {10.1103/PhysRevLett.129.150502},
  url = {https://link.aps.org/doi/10.1103/PhysRevLett.129.150502}
}

@article{chuang1997prescription,
  title={Prescription for experimental determination of the dynamics of a quantum black box},
  author={Chuang, Isaac L},
  journal={Journal of Modern Optics},
  volume={44},
  number={11-12},
  pages={2455--2467},
  year={1997},
  publisher={Taylor \& Francis},
  doi = {10.1080/09500349708231894},
  url = {https://www.tandfonline.com/doi/abs/10.1080/09500349708231894},
}

@article{flammia2012quantum,
  title={Quantum tomography via compressed sensing: error bounds, sample complexity and efficient estimators},
  author={Flammia, Steven T and Gross, David and Liu, Yi-Kai and Eisert, Jens},
  journal={New Journal of Physics},
  volume={14},
  number={9},
  pages={095022},
  year={2012},
  publisher={IOP Publishing},
  doi = {10.1088/1367-2630/14/9/095022},
  url = {https://doi.org/10.1088/1367-2630/14/9/095022},
}

@article{fletcher2008channel,
  title={Channel-adapted quantum error correction for the amplitude damping channel},
  author={Fletcher, Andrew S and Shor, Peter W and Win, Moe Z},
  journal={IEEE Transactions on Information Theory},
  volume={54},
  number={12},
  pages={5705--5718},
  year={2008},
  publisher={IEEE},
  doi={10.1109/TIT.2008.2006458}
}

@article{aliferis2008fault,
  title = {Fault-tolerant quantum computation against biased noise},
  author = {Aliferis, Panos and Preskill, John},
  journal = {Phys. Rev. A},
  volume = {78},
  issue = {5},
  pages = {052331},
  numpages = {9},
  year = {2008},
  month = {Nov},
  publisher = {American Physical Society},
  doi = {10.1103/PhysRevA.78.052331},
  url = {https://link.aps.org/doi/10.1103/PhysRevA.78.052331}
}

@article{bonilla2021xzzx,
  title={The XZZX surface code},
  author={Bonilla Ataides, J Pablo and Tuckett, David K and Bartlett, Stephen D and Flammia, Steven T and Brown, Benjamin J},
  journal={Nature communications},
  volume={12},
  number={1},
  pages={2172},
  year={2021},
  publisher={Nature Publishing Group UK London},
  doi={10.1038/s41467-021-22274-1}
}

@article{tuckett2019tailoring,
  title = {Tailoring Surface Codes for Highly Biased Noise},
  author = {Tuckett, David K. and Darmawan, Andrew S. and Chubb, Christopher T. and Bravyi, Sergey and Bartlett, Stephen D. and Flammia, Steven T.},
  journal = {Phys. Rev. X},
  volume = {9},
  issue = {4},
  pages = {041031},
  numpages = {22},
  year = {2019},
  month = {Nov},
  publisher = {American Physical Society},
  doi = {10.1103/PhysRevX.9.041031},
  url = {https://link.aps.org/doi/10.1103/PhysRevX.9.041031}
}

@article{chuang1997bosonic,
  title = {Bosonic quantum codes for amplitude damping},
  author = {Chuang, Isaac L. and Leung, Debbie W. and Yamamoto, Yoshihisa},
  journal = {Phys. Rev. A},
  volume = {56},
  issue = {2},
  pages = {1114--1125},
  numpages = {0},
  year = {1997},
  month = {Aug},
  publisher = {American Physical Society},
  doi = {10.1103/PhysRevA.56.1114},
  url = {https://link.aps.org/doi/10.1103/PhysRevA.56.1114}
}

@article{wu2025bias,
  title={Bias-tailored single-shot quantum LDPC codes},
  author={Wu, Shixin and Brun, Todd A and Lidar, Daniel A},
  journal={arXiv preprint arXiv:2507.02239},
  year={2025},  
  url = {https://arxiv.org/abs/2507.02239}
}

@article{kuehnke2025hardware,
  title={Hardware-tailored logical Clifford circuits for stabilizer codes},
  author={Kuehnke, Eric J and Levi, Kyano and Roffe, Joschka and Eisert, Jens and Miller, Daniel},
  journal={arXiv preprint arXiv:2505.20261},
  year={2025},
  url={https://arxiv.org/abs/2505.20261}
}

@article{bluvstein2024logical,
  title={Logical quantum processor based on reconfigurable atom arrays},
  author={Bluvstein, Dolev and Evered, Simon J and Geim, Alexandra A and Li, Sophie H and Zhou, Hengyun and Manovitz, Tom and Ebadi, Sepehr and Cain, Madelyn and Kalinowski, Marcin and Hangleiter, Dominik and others},
  journal={Nature},
  volume={626},
  number={7997},
  pages={58--65},
  year={2024},
  publisher={Nature Publishing Group UK London},
  doi = {10.1038/s41586-023-06927-3}
}

@article{bluvstein2025fault,
  title={A fault-tolerant neutral-atom architecture for universal quantum computation},
  author={Bluvstein, Dolev and Geim, Alexandra A and Li, Sophie H and Evered, Simon J and Bonilla Ataides, J Pablo and Baranes, Gefen and Gu, Andi and Manovitz, Tom and Xu, Muqing and Kalinowski, Marcin and others},
  journal={Nature},
  pages={1--3},
  year={2025},
  publisher={Nature Publishing Group UK London},
  doi = {10.1038/s41586-025-09848-5}
}

@article{google2024threshold,
  title={Quantum error correction below the surface code threshold},
  author={{Google Quantum AI and Collaborators}},
  journal={Nature},
  volume={638},
  pages={920},
  year={2024},
  doi={10.1038/s41586-024-08449-y}
}

@article{google2023suppressing,
  author       = {{Google Quantum AI}},
  title        = {Suppressing quantum errors by scaling a surface code logical qubit},
  journal      = {Nature},
  volume       = {614},
  number       = {7949},
  pages        = {676--681},
  year         = {2023},
  doi          = {10.1038/s41586-022-05434-1}
}

@article{besedin2026lattice,
  title={Lattice surgery realized on two distance-three repetition codes with superconducting qubits},
  author={Besedin, Ilya and Kerschbaum, Michael and Knoll, Jonathan and Hesner, Ian and B{\"o}deker, Lukas and Colmenarez, Luis and Hofele, Luca and Lacroix, Nathan and Hellings, Christoph and Swiadek, Fran{\c{c}}ois and others},
  journal={Nature Physics},
  pages={1--6},
  year={2026},
  publisher={Nature Publishing Group UK London},
  doi={10.1038/s41567-025-03090-6}
}

@article{putterman2025hardware,
  title={Hardware-efficient quantum error correction via concatenated bosonic qubits},
  author={Putterman, Harald and Noh, Kyungjoo and Hann, Connor T and MacCabe, Gregory S and Aghaeimeibodi, Shahriar and Patel, Rishi N and Lee, Menyoung and Jones, William M and Moradinejad, Hesam and Rodriguez, Roberto and others},
  journal={Nature},
  volume={638},
  number={8052},
  pages={927--934},
  year={2025},
  publisher={Nature Publishing Group UK London},
  doi={10.1038/s41586-025-08642-7}
}

@article{brock2025quantum,
  title={Quantum error correction of qudits beyond break-even},
  author={Brock, Benjamin L and Singh, Shraddha and Eickbusch, Alec and Sivak, Volodymyr V and Ding, Andy Z and Frunzio, Luigi and Girvin, Steven M and Devoret, Michel H},
  journal={Nature},
  volume={641},
  number={8063},
  pages={612--618},
  year={2025},
  publisher={Nature Publishing Group UK London},
  doi ={10.1038/s41586-025-08899-y}
}

@misc{paetznick2024demonstration,
      title={Demonstration of logical qubits and repeated error correction with better-than-physical error rates}, 
      author={Paetznick, A and Da Silva, MP and Ryan-Anderson, C and Bello-Rivas, JM and Campora III, JP and Chernoguzov, A and Dreiling, JM and Foltz, C and Frachon, F and Gaebler, JP and others},
      year={2024},
      eprint={2404.02280},
      archivePrefix={arXiv},
      primaryClass={quant-ph},
      url={https://arxiv.org/abs/2404.02280}, 
}

@article{ryan2024high,
  title={High-fidelity teleportation of a logical qubit using transversal gates and lattice surgery},
  author={Ryan-Anderson, Ciaran and Brown, NC and Baldwin, CH and Dreiling, JM and Foltz, C and Gaebler, JP and Gatterman, TM and Hewitt, N and Holliman, C and Horst, CV and others},
  journal={Science},
  volume={385},
  number={6715},
  pages={1327--1331},
  year={2024},
  publisher={American Association for the Advancement of Science},
  doi={10.1126/science.adp6016}
}

@article{generalized_DD,
  title = {Syncopated dynamical decoupling to suppress crosstalk in quantum circuits},
  author = {Evert, Bram and Gonzalez Izquierdo, Zoe and Sud, James and Hu, Hong-Ye and Grabbe, Shon and Rieffel, Eleanor G. and Reagor, Matthew J. and Wang, Zhihui},
  journal = {Phys. Rev. Appl.},
  volume = {24},
  issue = {4},
  pages = {044025},
  numpages = {15},
  year = {2025},
  month = {Oct},
  publisher = {American Physical Society},
  doi = {10.1103/8lxc-lvv1},
  url = {https://link.aps.org/doi/10.1103/8lxc-lvv1}
}

@article{rl_qec,
  title = {Discovery of optimal quantum codes via reinforcement learning},
  author = {Su, Vincent Paul and Cao, ChunJun and Hu, Hong-Ye and Yanay, Yariv and Tahan, Charles and Swingle, Brian},
  journal = {Phys. Rev. Appl.},
  volume = {23},
  issue = {3},
  pages = {034048},
  numpages = {19},
  year = {2025},
  month = {Mar},
  publisher = {American Physical Society},
  doi = {10.1103/PhysRevApplied.23.034048},
  url = {https://link.aps.org/doi/10.1103/PhysRevApplied.23.034048}
}

@article{Bayesian_noise,
	author = {Hu, Hong-Ye and Gu, Andi and Majumder, Swarnadeep and Ren, Hang and Zhang, Yipei and Wang, Derek S. and You, Yi-Zhuang and Minev, Zlatko and Yelin, Susanne F. and Seif, Alireza},
	date = {2025/03/26},
	doi = {10.1038/s41467-025-57349-w},
	id = {Hu2025},
	isbn = {2041-1723},
	journal = {Nature Communications},
	number = {1},
	pages = {2943},
	title = {Demonstration of robust and efficient quantum property learning with shallow shadows},
	url = {https://doi.org/10.1038/s41467-025-57349-w},
	volume = {16},
	year = {2025}}

\clearpage
\onecolumngrid

\appendix
\EnableTOC

\begingroup
\renewcommand{\tocname}{Contents of Appendix}
\tableofcontents
\endgroup

\makeatletter
\markboth{\MakeUppercase{\@shorttitle}}{\MakeUppercase{\@shorttitle}}
\makeatother
\clearpage

\section{Notation}
\label{sec:appendix_notation_prelim}

\noindent We consider $n$ qubits with Hilbert space $\mathcal H=(\mathbb C^2)^{\otimes n}$ and dimension $d = 2^n$ (not to be confused with the dual-graph degree $\mathfrak d$ or the probe derivatives $d_{O,Q}$). We write $\mathcal B(\mathcal H)$ for the space of linear operators on $\mathcal H$.
A linear map $\Phi:\mathcal B(\mathcal H)\to\mathcal B(\mathcal H)$ is called a superoperator. We use the Hilbert--Schmidt inner product $\langle A,B\rangle \;\coloneq\; \tr(A^\dagger B)$. The adjoint of a superoperator $\Phi$ is defined with respect to $\langle\cdot,\cdot\rangle$ by
\begin{equation}
    \langle \Phi(A),B\rangle \;=\; \langle A,\Phi^\dagger(B)\rangle
    \qquad \forall\,A,B\in\mathcal B(\mathcal H).
\end{equation}
For an observable $O$ and a state $\rho$, we write $\langle O\rangle_\rho \coloneq \tr(O\rho)$.

\medskip
\noindent\textbf{Norms.}
For $x\in\mathbb C^m$ and $p\in[1,\infty)$ we set
\begin{equation}
    \|x\|_p \;\coloneq\; \Big(\sum_{i=1}^m |x_i|^p\Big)^{1/p},
    \qquad
    \|x\|_\infty \;\coloneq\; \max_{1\le i\le m}|x_i|.
\end{equation}
We use $\|\cdot\|_\infty$ for entrywise control of noise and coefficient-estimation error vectors (e.g.\ in \cref{theorem:ancilla_free_coefficient_learning}). For a matrix $A\in\mathbb C^{m\times m'}$ and $p,q\in[1,\infty]$, the induced $(p\to q)$ norm is
\begin{equation}
    \|A\|_{p\to q}\;\coloneq\;\sup_{x\neq 0}\frac{\|Ax\|_q}{\|x\|_p}.
\end{equation}
The only case we use explicitly is the $\ell_\infty$-induced norm
\begin{equation}
    \|A\|_{\infty\to\infty}
    \;=\;
    \max_{1\le i\le m}\sum_{j=1}^{m'}|A_{ij}|,
\end{equation}
i.e.\ the maximum absolute row sum. For $X\in\mathcal B(\mathcal H)$ with singular values $\{s_i\}_{i=1}^d$, the Schatten-$p$ norm is
\begin{equation}
    \|X\|_{p}\;\coloneq\;\Big(\sum_{i=1}^d s_i^p\Big)^{1/p}\quad(1\le p<\infty),
    \qquad
    \|X\|_{\infty}\;\coloneq\;\max_i s_i(X).
\end{equation}
Thus $\|X\|_1$ is the trace norm, $\|X\|_2$ is the Hilbert--Schmidt norm, and $\|X\|_\infty$ is the operator norm.
We use repeatedly that
\begin{equation}
    \|X\|_\infty \le \|X\|_2 \le \|X\|_1,
    \qquad
    \|X\|_1 \le \sqrt d\,\|X\|_2,
    \qquad
    \|X\|_2 \le \sqrt d\,\|X\|_\infty.
\end{equation}
For a density operator $\rho\succeq 0$ with $\tr(\rho)=1$, one has $\|\rho\|_1=1$ and $\|\rho\|_2\le 1$. For a superoperator $\Phi:\mathcal B(\mathcal H)\to\mathcal B(\mathcal H)$ and $p,q\in[1,\infty]$, we define the induced Schatten $(p\to q)$ norm 
\begin{equation}
    \|\Phi\|_{p\to q}\;\coloneq\;\sup_{X\neq 0}\frac{\|\Phi(X)\|_q}{\|X\|_p},
\end{equation}
where $\|\cdot\|_p$ and $\|\cdot\|_q$ are Schatten norms on $\mathcal B(\mathcal H)$.
The induced Schatten-$2$ norm $\|\mathcal L\|_{2\to2}$ is used in the $\chi$-derivative bounds (e.g.\ \cref{cor:chi_derivative_bound_M_sparse})
and is bounded in \cref{sec:induced_norm_bound}.
Induced norms are submultiplicative:
\begin{equation}
    \|\Phi\circ\Psi\|_{p\to q}\le \|\Phi\|_{r\to q}\,\|\Psi\|_{p\to r}.
\end{equation}
We also use H\"older's inequality for Schatten norms:
\begin{equation}
    |\tr(A^\dagger B)|\le \|A\|_p\,\|B\|_q,
    \qquad \frac1p+\frac1q=1.
\end{equation}
Finally, for a scalar function $f:[0,\infty)\to\mathbb C$ we write
\begin{equation}
    \|f\|_\infty \;\coloneq\; \sup_{t\ge 0}|f(t)|.
\end{equation}
This convention is used throughout the Chebyshev derivative analysis (e.g.\ bounds of the form
$\|f^{(k)}\|_\infty \le B\,\Lambda^k k!$ in \cref{sec:chebyshev_interpolation_derivative_analysis}) and in the uniform
$\chi$-derivative bounds.

\medskip
\noindent\textbf{Pauli group and Pauli operator basis.}
We use $I, X, Y, Z$ to denote the identity and the three Pauli matrices:
\begin{equation}
    I=\begin{pmatrix}1&0\\0&1\end{pmatrix},\quad
    X=\begin{pmatrix}0&1\\1&0\end{pmatrix},\quad
    Y=\begin{pmatrix}0&-i\\ i&0\end{pmatrix},\quad
    Z=\begin{pmatrix}1&0\\0&-1\end{pmatrix}.
\end{equation}
The single-qubit Pauli group $\mathbb{P}_1$ is the group generated by the three Pauli matrices, along with the identity:
\begin{equation}
    \mathbb{P}_1 = \langle I, X,Y,Z \rangle = \{\pm I,\pm iI,\pm X,\pm iX,\pm Y,\pm iY,\pm Z,\pm iZ \}.
\end{equation}
The $n$-qubit Pauli group $\mathbb{P}_n$ of order $|\mathbb{P}_n| = 4\cdot4^{n}$ is generated by single-qubit Pauli operators acting on each of the $n$ qubits:
\begin{equation}
    \mathbb{P}_n = \langle \sigma_1\otimes\cdots\otimes \sigma_n \;:\; \sigma_i \in \{I,X,Y,Z \} \rangle.
\end{equation}
For computational purposes, a convenient Hermitian operator basis for the space of linear operators $\mathcal{B}(\mathcal{H})$ is given by the set of all $n$-qubit Pauli strings:
\begin{equation}
  \mathcal{P}_n = \{I, X,Y,Z\}^{\otimes n}.
\end{equation}
In other words, $\mathcal{P}_n$ is obtained by stripping off the overall phases from $\mathbb{P}_n$, such that any two elements $P_i,P_j \in \mathcal{P}_n$ are orthogonal to each other under the Hilbert--Schmidt inner product: $\tr{P_i P_j} = 2^n\delta_{ij}$. The set $\mathcal{P}_n$ has cardinality $|\mathcal{P}_n| = 4^n$. When multiplying two Pauli strings $P_i,P_j\in\mathcal P_n$, the product is another Pauli up to a phase:
\begin{equation}
    P_i P_j = \omega_{ij}\, P_{i\oplus j}, \qquad \omega_{ij}\in\{\pm1,\pm i\}.
\end{equation}
These phases record the commutation relations of Paulis and satisfy $\omega_{ij}=\omega_{ji}^*$, so that $(\omega_{ij})$ is Hermitian. Keeping track of $\omega_{ij}$ allows us to work with the phase-stripped basis $\mathcal P_n$, which, unlike $\mathbb P_n$, is not closed under multiplication.

We adopt the convention of labeling $P_0 = I^{\otimes n}$ as the $n$-qubit identity operator. Every other element $P_i \in \mathcal{P}_n$ with $i \neq 0$ is traceless: $\tr{P_i} = 2^n \delta_{i0}$. The orthogonality and completeness of $\mathcal{P}_n$ means that any operator $X \in \mathcal{B}(\mathcal{H})$ can be uniquely expanded in the Pauli basis as:
\begin{equation}
X = \sum_{i=0}^{4^n-1} x_i P_i,
\end{equation}
where the expansion coefficients are given by $x_i = \frac{1}{2^n}\tr{P_i X}$.

\medskip
\noindent\textbf{Supports and patches.}
For an $n$-qubit Pauli string $P$, its support $\supp{P}\subseteq[n]$ is the set of qubits on which $P$ acts nontrivially,
and its weight is $w(P)\coloneq|\supp{P}|$. A patch (or region) $T\subseteq[n]$ is a subset of the $n$ qubits. For a patch $T$, we write
\begin{equation}
    \mathcal P_T := \{P\in\mathcal P_n : \supp{P}\subseteq T\}
\end{equation}
to denote all $n$-qubit Pauli strings that are supported entirely on this patch.
This notation is used for the enlarged regions $\widetilde T$ on which the targeted probes of the coefficient-learning stage are supported (\cref{sec:stable_local_reconstruction,sec:deterministic_targeted_probes}).

\noindent\textbf{Lindbladian locality.}
Given a Lindbladian superoperator $\mathcal L$ expressed in the Pauli basis,
\begin{equation*}
    \mathcal L(\rho)
    =-i\sum_{P_i\in \mathcal S_H} h_i [P_i,\rho]
    +\sum_{P_i,P_j\in \mathcal S_D} a_{ij}\!\left(P_i \rho P_j-\tfrac12\{P_jP_i,\rho\}\right),
\end{equation*}
we define the Hamiltonian locality $k_H$ and dissipator locality $k_D$ by
\begin{equation}\label{eq:lindbladian-pauli-locality}
    k_H \coloneq \max_{P\in\mathcal S_H} |\supp{P}|,
    \qquad
    k_D \coloneq \max_{P_i,P_j\in\mathcal S_D} |\supp{P_i}\cup\supp{P_j}|,
\end{equation}
and set $k\coloneq\max\{k_H,k_D\}$. These localities are invariant under the unitary gauge freedom in the choice of jump operators. In the coefficient-learning stage, $k$ enters through the optional parallelized derivative estimation via shadow process tomography (\cref{sec:coefficient_learning}), where every probe is supported on at most $k+\mathfrak d$ qubits.

\medskip
\noindent\textbf{Asymptotic notation.}
We use standard Landau notation and write $\widetilde{\mathcal O}(\cdot)$ to suppress polylogarithmic factors in the relevant
problem parameters (e.g.\ $n$, $M$, $1/\eta$, $1/\varepsilon$, $1/\delta$).
\section{Chebyshev interpolation}\label{sec:chebyshev_interpolation}

Chebyshev interpolation is instrumental to our structure and coefficient learning algorithms. We therefore provide a brief overview based on standard references~\citep{boyd2001chebyshev, mason2002chebyshev, trefethen2019approximation, marshall2024chebyshev}.

\subsection{Overview and intuition}\label{sec:chebyshev_interpolation_overview}
Given a function $f:[a,b]\to\mathbb{R}$, we would like to approximate it by a polynomial. The Weierstrass approximation theorem guarantees that any continuous function on a bounded interval can be approximated arbitrarily well by a polynomial. The question is, how do we construct such a polynomial systematically and with good numerical properties.

A natural idea is to pick $r{+}1$ points $t_0 < t_1 < \dots <t_r \in [a,b]$, evaluate $f(t)$ at those points, and then find a polynomial $p(t)$ of degree at most $r$ that exactly fits the data:
\begin{equation}
    p(t_m) = f(t_m), \qquad m=0,\dots,r.
\end{equation}
Such a polynomial, called the interpolant, always exists and is unique. It can be written explicitly in the Lagrange form
\begin{equation}
    p(t) = \sum_{m=0}^{r} f(t_m)\,\ell_m(t),
    \qquad
    \ell_m(t) = \prod_{\substack{n \neq m}}^{r} \frac{t - t_n}{t_m - t_n},
    \label{eq:lagrange_form}
\end{equation}
where $\ell_m(t)$ are the cardinal basis polynomials satisfying $\ell_m(t_n)=\delta_{mn}$. If $f$ is $(r{+}1)$-times continuously differentiable, the remainder theorem guarantees that the pointwise interpolation error can be bounded as follows:
\begin{equation}
    |f(t)-p(t)|
    \;\le\;
    \frac{\|f^{(r{+}1)}\|_\infty}{(r+1)!}
    \;\prod_{m=0}^{r}|t-t_m|,
\end{equation}
where $\|f^{(r{+}1)}\|_\infty : = \sup_{t\in[a,b]}|f^{(r+1)}(t)|$. The second term, $\prod_{m=0}^{r}|t-t_m|$, depends entirely on the choice of nodes. For example, the naïve choice of equispaced nodes often leads to large oscillations near the boundaries--the Runge phenomenon--and a rapidly growing error term. To obtain the best possible set of nodes for a fixed degree $r$, one can minimize the worst-case size of this node-dependent product. On the standard interval $[-1,1]$, the unique minimizer is given by the set of Gauss-Chebyshev nodes
\begin{equation}
    z_m = \cos\left(\frac{2m+1}{r+1} \, \frac{\pi}{2}\right),
    \qquad m=0,\dots,r,
    \label{eq:chebyshev_nodes}
\end{equation}
which are precisely the roots of the $(r{+}1)$-st Chebyshev polynomial of the first kind $T_{r+1}(z)$. Chebyshev polynomials can be defined recursively as
\begin{equation}
    T_0(z)=1, \qquad
    T_1(z)=z, \qquad
    T_{r+1}(z)=2zT_r(z)-T_{r-1}(z).
\end{equation}
Gauss-Chebyshev nodes $\{Z_m\}_{m=0}^r$ minimize the node-dependent factor for a fixed $r$. Mapping back to the physical interval $[a,b]$ via $t_m := \tfrac{b-a}{2}\, z_m + \tfrac{a+b}{2}$, the polynomial $p(t)$ fitted through the data $\{(t_m, f(t_m)\}_{m=0}^r$ obeys
\begin{equation}
    \max_{t\in[a,b]} |f(t)-p(t)|
    \;\le\;
    \frac{2\|f^{(r{+}1)}\|_\infty}{(r+1)!}\,
    \Big(\frac{b-a}{4}\Big)^{r+1}.
    \label{eq:chebyshev_error_bound}
\end{equation}
Although the resulting polynomial has guaranteed approximation accuracy, the Lagrange form \eqref{eq:lagrange_form} is numerically unstable to evaluate $p(t)$ for large $r$. An equivalent but more stable and computationally efficient representation is obtained by expanding $p(t)$ in the Chebyshev basis,
\begin{equation}
    p(t) = \sum_{n=0}^{r} c_n\,T_n\left(z(t)\right),
\end{equation}
where we map $[a,b] \mapsto [-1,1]$ via $z(t) = \tfrac{2t - (a+b)}{b-a}$ and the coefficients $\{c_n\}$ are called the Chebyshev coefficients. These coefficients can readily be computed in $O(r\log r)$ time by applying the Fast Fourier Transform (FFT) to the function values $\{f(t_m)\}$. Equivalently, by the discrete orthogonality of the Chebyshev polynomials on Chebyshev nodes, the coefficients admit an explicit form
\begin{equation}
    c_0 \;=\; \frac{1}{r+1}\sum_{m=0}^{r} f(t_m),
    \qquad
    c_n \;=\; \frac{2}{r+1}\sum_{m=0}^{r} f(t_m)\,T_n(z_m),
    \quad n=1,\ldots,r .
    \label{eq:cheb_coeffs_explicit}
\end{equation}

If our only goal were to approximate $f(t)$ on $[a,b]$, we would be done here. However, we are interested not in $f(t)$ itself, but in its first and second derivatives at $x=0$. In addition, the data points $\{f(t_i)\}$ are obtained from noisy measurements with an uncertainty $\pm\varepsilon_s$. In the following, we will show that the Chebyshev interpolant can be used to estimate the first and second derivatives of $f$ at $t=0$ from noisy function evaluations. The analysis follows the approach of~\citet{gu2024practical}, which we generalize to include the estimation of the second derivative.

\subsection{Derivative analysis of the Chebyshev interpolant}\label{sec:chebyshev_interpolation_derivative_analysis}

From now on, the interval of interest corresponds to a finite time window $[0,\tau_{\max}]$, and the map to the standard domain $[-1,1]$ is given by $z(t) = \frac{2t}{\tau_{\max}} - 1$. Let 
\begin{equation}
    \widehat p(t)=\sum_{n=0}^r \widehat c_n\,T_n\big(z(t)\big)
\end{equation}
denote the degree-$r$ Chebyshev interpolant constructed from the noisy function estimates 
$\{\widehat f(t_m)\}_{m=0}^r$ at the Gauss-Chebyshev nodes $\{t_m\}_{m=0}^r$. We assume that the measurement noise is uniformly bounded and controllable:
\begin{equation}
    |\widehat f(t_m)-f(t_m)|\le\varepsilon_s.
\end{equation}
The corresponding derivative estimators at $t=0$ (i.e.\ $z=-1$) are
\begin{align}
    \label{eq:chebyshev_derivative_estimator_first}
    \widehat{p}^{\,\,\prime}(0)
    &= \sum_{n=0}^r \widehat c_n\,\dv{t}T_n(z(t))\Big|_{t=0}
     = \frac{2}{\tau_{\max}}\sum_{n=1}^r (-1)^{n+1} n^2\,\widehat c_n,\\[3pt]
    \label{eq:chebyshev_derivative_estimator_second}
    \widehat{p}^{\,\,\prime\prime}(0)
    &= \sum_{n=0}^r \widehat c_n\,\dv[2]{t}T_n(z(t))\Big|_{t=0}
     = \frac{4}{3\tau_{\max}^2}\sum_{n=2}^r (-1)^n n^2(n^2-1)\,\widehat c_n,
\end{align}
where we used $T_n'(-1)=(-1)^{n+1}n^{2}$ and $T_n''(-1)=(-1)^n n^2(n^2-1)/3$ (see, e.g.,~\cite[§2.4.5]{mason2002chebyshev}). By substituting the explicit form of the Chebyshev coefficients from \cref{eq:cheb_coeffs_explicit}, the derivative estimators can be expressed as weighted sums of the sampled function values:
\begin{equation} \label{eq:chebyshev-weights}
    \widehat{p}^{\,\prime}(0)
    \;=\;
    \sum_{m=0}^{r} \alpha^{(1)}_m\,\widehat f(t_m),
    \qquad
    \widehat{p}^{\,\prime\prime}(0)
    \;=\;
    \sum_{m=0}^{r} \alpha^{(2)}_m\,\widehat f(t_m),
\end{equation}
where the weights $\{\alpha^{(1)}_m\}$ and $\{\alpha^{(2)}_m\}$ can be precomputed and reused for all derivative estimators on the same Chebyshev grid:
\begin{align}
    \label{eq:chebyshev_derivative_weights}
    \alpha^{(1)}_m
    \;&=\;
    \frac{4}{\tau_{\max}(r{+}1)}\,
    \sum_{n=1}^{r} (-1)^{n+1} n^2\,T_n(z_m)\\
    \alpha^{(2)}_m
    \;&=\;
    \frac{8}{3\tau_{\max}^2(r{+}1)}\,
    \sum_{n=2}^{r} (-1)^{n} n^2(n^2{-}1)\,T_n(z_m),
\end{align}
Then, the following result quantifies the total error of the derivative estimators defined above.

\begin{theorem}[Noisy derivative estimators]
    \label{theorem:chebyshev_derivative_total_error}
    Let $\tau_{\max}>0$ and $f$ be $(r{+}1)$-times continuously differentiable on $[0,\tau_{\max}]$. Suppose the available noisy samples $\{\widehat f(t_m)\}_{m=0}^r$ on the set of Gauss-Chebyshev nodes $\{t_m\}_{m=0}^r \subset [0,\tau_{\max}]$ have bounded uniform measurement noise: $|\widehat f(t_m)-f(t_m)|\le \varepsilon_s$. Then, the estimators $\widehat{p}^{\,\prime}(0)$ and $\widehat{p}^{\,\prime\prime}(0)$ defined in \cref{eq:chebyshev_derivative_estimator_first} and \cref{eq:chebyshev_derivative_estimator_second} satisfy
    \begin{align}
        \label{eq:chebyshev_total_error_first}
        \bigl|\widehat{p}^{\,\prime}(0)-f'(0)\bigr|
        &\le
        \underbrace{\frac{5r^3}{2\tau_{\max}}\,\varepsilon_s}_{\text{\em noise term}}
        \;+\;
        \underbrace{2(r+1)^2\tau_{\max}^r\frac{\|f^{(r+1)}\|_\infty}{r!}}_{\text{\em bias term}}, \\[4pt]
        \label{eq:chebyshev_total_error_second}
        \bigl|\widehat{p}^{\,\prime\prime}(0)-f''(0)\bigr|
        &\le
        \underbrace{\frac{r^5}{\tau_{\max}^2}\,\varepsilon_s}_{\text{\em noise term}}
        \;+\;
        \underbrace{\tfrac{2}{3}(r+1)^2((r+1)^2-1) \tau_{\max}^{r-1}\frac{\|f^{(r+1)}\|_\infty}{(r-1)!}}_{\text{\em bias term}}.
    \end{align}    
    where $\|f^{(r+1)}\|_\infty : = \sup_{t\in[0,\tau_{\max}]}|f^{(r+1)}(t)|$ upper bounds the $(r+1)$-st derivative of $f$ on $[0,\tau_{\max}]$.
\end{theorem}

\begin{proof}
    The total error can be decomposed into bias and noise:
    \begin{equation}
        \label{eq:bias_noise_decomposition_derivatives}
        |\widehat{p}^{\,\prime}(0)-f'(0)| \;\le\; |p^{\,\prime}(0)-f'(0)| + |\widehat{p}^{\,\prime}(0)-p^{\,\prime}(0)|.
    \end{equation}
    The decomposition for the second derivative is analogous. We first bound the bias term. By construction, the Chebyshev interpolant exactly fits $r+1$ data points $\{(t_m, f(t_m))\}_{m=0}^{r}$. Therefore, it classifies as a Lagrange interpolating polynomial, and we can import derivative bounds for Lagrange interpolants:
    \begin{theorem}[Derivative error bound for Lagrange interpolants~{\cite[Theorem~3]{mason2002chebyshev}}]
        \label{theorem:lagrange_derivative_bound}
        Let $f$ be $(r{+}1)$-times continuously differentiable on $[a,b]$ and let $p(t)$ be the degree-$r$ Lagrange interpolant of $f$ at distinct nodes $\{t_m\}_{m=0}^r\subset[a,b]$.  
        Then, for any $j\le r$ and $t \in [a,b]$
        \begin{equation}
            |p^{(j)}(t) - f^{(j)}(t)|
            \;\le\;
            \|\omega^{(j)}\|_\infty \,\frac{\|f^{(r{+}1)}\|_\infty}{j!\,(r{+}1{-}j)!},
            \qquad
            \omega(t)=\prod_{m=0}^{r}(t-t_m).
            \label{eq:lagrange_derivative_error_bound}
        \end{equation}
    \end{theorem}
    The terms $\|\omega'\|_\infty$ and $\|\omega''\|_\infty$ can be bounded using the Markov brothers' inequality on the interval $[0,\tau_{\max}]$: 
    \begin{equation}
        \|\omega'\|_\infty \;\le\; \frac{2(r+1)^2}{\tau_{\max}} \|\omega\|_\infty
        \qq{and}
        \|\omega''\|_\infty \;\le\; \frac{4(r+1)^2((r+1)^2-1)}{3\tau_{\max}^2} \|\omega\|_\infty
    \end{equation}
    For $\|\omega\|_\infty$, a straightforward upper bound is $\|\omega\|_\infty \le \tau_{\max}^{r+1}$ since $|t-t_m|\le \tau_{\max}$. Putting everything together, we obtain the claimed upper bounds for the two bias terms $\bigl|{p}^{\,\prime}(0)-f'(0)\bigr|$ and $\bigl|{p}^{\,\prime\prime}(0)-f''(0)\bigr|$. 

    To bound the noise term, consider the explicit form of the derivative estimator  in \cref{eq:chebyshev_derivative_estimator_first} and \cref{eq:chebyshev_derivative_estimator_second} and plug in the expansion of Chebyshev coefficients in \cref{eq:cheb_coeffs_explicit}
    \begin{align}
        \widehat{p}^{\,\,\prime}(0)
        &\;=\; \frac{2}{\tau_{\max}}\frac{2}{r+1}\sum_{m=1}^r (-1)^{m+1} m^2 \sum_{i=0}^{r} \widehat f(t_i)\,T_m(z_i) \\[3pt]
        \widehat{p}^{\,\,\prime\prime}(0)
        &\;=\; \frac{4}{3\tau_{\max}^2}\frac{2}{r+1} \sum_{m=2}^r (-1)^m m^2(m^2-1) \sum_{i=0}^{r} \widehat f(t_i)\,T_m(z_i)
    \end{align}
    Then, the noise can be bounded as follows:
    \begin{align}
        |\widehat{p}^{\,\prime}(0)-p^{\,\prime}(0)| 
        \;&\le\; \frac{2}{\tau_{\max}}\frac{2}{r+1}\sum_{m=1}^r m^2 \sum_{i=0}^{r} |\widehat f(t_i) - f(t_i)|\,|T_m(z_i)| \\
        \;&\le\; \frac{5r^3}{2\tau_{\max}}\,\varepsilon_s  \\[3pt]
        |\widehat{p}^{\,\prime\prime}(0)-p^{\,\prime\prime}(0)| 
        \;&\le\; \frac{4}{3\tau_{\max}^2}\frac{2}{r+1}\sum_{m=2}^r m^2(m^2-1) \sum_{i=0}^{r} |\widehat f(t_i) - f(t_i)|\,|T_m(z_i)| \\
        \;&\le\; \frac{r^5}{\tau_{\max}^2}\,\varepsilon_s
    \end{align}
    where we used $|\widehat f(t_i) - f(t_i)| \le \varepsilon_s$ by assumption, and $|T_m(z_i)| \le 1$. We also applied the following bounds for $r\ge 2$: $\sum_{m=1}^r m^2 \le 5r^3/8$ and $\sum_{m=2}^r m^2(m^2-1) \le 3r^5/8$. Putting the bias and noise terms together, as in \cref{eq:bias_noise_decomposition_derivatives}, we obtain the claimed bounds.    
\end{proof}

Now, equipped with \cref{theorem:lagrange_derivative_bound}, we can provide a recipe for the choice of $r$, $\tau_{\max}$, and $\varepsilon_s$ to meet a target approximation accuracy of $\varepsilon$ for the first and second derivative estimators. We first consider the regime where the derivatives of $f$ satisfy the uniform bound $\|f^{(k)}\|_\infty \;\le\; B\,\Lambda^k$, which is the relevant regime in \cref{sec:structure_learning}.

\begin{corollary}[First derivative: sufficient $\tau_{\max}$, $r$, and $\varepsilon_s$]
    \label{cor:first_chebyshev_derivative_parameters} 
    Suppose the derivatives of $f$ satisfy the uniform bound
    \[
        \|f^{(k)}\|_\infty \;\le\; B\,\Lambda^k , \qquad (k \;\ge\; 1).
    \]
    Fix a target accuracy $\varepsilon>0$ and set
    \begin{equation}
    \tau_{\max}=\frac{1}{2\Lambda},
    \qquad
    r \;=\; \big\lceil \log\big(\tfrac{18\,B\Lambda}{\varepsilon}\big)\big\rceil,
    \qquad
    \varepsilon_s \;=\; \frac{\varepsilon}{10\,\Lambda\,r^{3}}
    \end{equation}
    Then the derivative estimator $\widehat{p}^{\,\prime}(0)$ in \cref{eq:chebyshev_derivative_estimator_first} satisfies $\bigl|\widehat p^{\,\prime}(0)-f'(0)\bigr|\le \varepsilon$.
\end{corollary}

\begin{proof}
    We split the error budget equally, requiring (bias) $\le \varepsilon/2$ and (noise) $\le \varepsilon/2$.
    Set $\tau_{\max}=\tfrac{1}{2\Lambda}$, so $1/\tau_{\max}=2\Lambda$ and $(\Lambda\tau_{\max})^{r+1}=2^{-(r+1)}$.
    By \cref{theorem:chebyshev_derivative_total_error},
    \begin{equation}
        \text{bias}=2(r+1)^2\tau_{\max}^r\frac{\|f^{(r+1)}\|_\infty}{r!}
        \le
        \frac{9B}{\tau_{\max}}(\Lambda\tau_{\max})^{r+1}
        =18\,B\Lambda\,2^{-(r+1)}
        \le \frac{\varepsilon}{2}
        \quad\Longrightarrow\quad
        r\ \ge\ \log\Bigl(\tfrac{18\,B\Lambda}{\varepsilon}\Bigr)
    \end{equation}
    where we used that $(r+1)^2/r! \le 9/2$ and $\|f^{(r+1)}\|_\infty \;\le\; B\,\Lambda^{r+1}$. 
    For the noise term, again by \cref{theorem:chebyshev_derivative_total_error},
    \begin{equation}
        \text{noise}=\frac{5r^{3}}{2\tau_{\max}}\,\varepsilon_s
        =5\Lambda\,r^{3}\,\varepsilon_s
        \le \frac{\varepsilon}{2}
        \quad\Longrightarrow\quad
        \varepsilon_s\ \le\ \frac{\varepsilon}{10\,\Lambda\,r^{3}}.
    \end{equation}
    With these two conditions, $|\widehat p^{\,\prime}(0)-f'(0)|\le \varepsilon$.
\end{proof}

\begin{corollary}[Second derivative: sufficient $\tau_{\max}$, $r$, and $\varepsilon_s$]
    \label{cor:second_chebyshev_derivative_parameters}
    Suppose the derivatives of $f$ satisfy the uniform bound
    \[
        \|f^{(k)}\|_\infty \;\le\; B\,\Lambda^k , \qquad (k \;\ge\; 1).
    \]
    Fix a target accuracy $\varepsilon>0$ and set
    \begin{equation}
        \tau_{\max}=\frac{1}{2\Lambda},
        \qquad
        r \;=\; \big\lceil \log\big(\tfrac{320\,B\Lambda^{2}}{\varepsilon}\big)\big\rceil,
        \qquad
        \varepsilon_s \;=\; \frac{\varepsilon}{8\,\Lambda^{2}\,r^{5}} .
    \end{equation}
    Then the second-derivative estimator $\widehat{p}^{\,\prime\prime}(0)$ in \cref{eq:chebyshev_derivative_estimator_second} satisfies $\bigl|\widehat p^{\,\prime\prime}(0)-f''(0)\bigr|\le \varepsilon$.
\end{corollary}

\begin{proof}
    We split the error budget equally, requiring (bias) $\le \varepsilon/2$ and (noise) $\le \varepsilon/2$.
    Set $\tau_{\max}=\tfrac{1}{2\Lambda}$, so $1/\tau_{\max}^2 = 4\Lambda^2$ and $(\Lambda\tau_{\max})^{r+1}=2^{-(r+1)}$.
    By \cref{theorem:chebyshev_derivative_total_error},
    \begin{align*}
        \text{bias}&=\tfrac{2}{3}(r+1)^2((r+1)^2-1) \tau_{\max}^{r-1}\frac{\|f^{(r+1)}\|_\infty}{(r-1)!}
        \le \frac{80B}{\tau_{\max}^2}(\Lambda\tau_{\max})^{r+1} \\
        &=320\,B\Lambda^{2}\,2^{-(r+1)}
        \le \frac{\varepsilon}{2}
        \quad\Longrightarrow\quad
        r\ \ge\ \log\Bigl(\tfrac{320\,B\Lambda^{2}}{\varepsilon}\Bigr).
    \end{align*}
    where we used that $(r+1)^2((r+1)^2-1)/(r-1)! \le 120$ and $\|f^{(r+1)}\|_\infty \;\le\; B\,\Lambda^{r+1}$. 
    For the noise term, again by \cref{theorem:chebyshev_derivative_total_error},
    \begin{equation}
        \text{noise}=\frac{r^{5}}{\tau_{\max}^{2}}\,\varepsilon_s
        =4\Lambda^{2}r^{5}\varepsilon_s
        \le \frac{\varepsilon}{2}
        \quad\Longrightarrow\quad
        \varepsilon_s\ \le\ \frac{\varepsilon}{8\,\Lambda^{2}\,r^{5}}.
    \end{equation}
    With these two conditions, $|\widehat p^{\,\prime\prime}(0)-f''(0)|\le \varepsilon$.
\end{proof}

We now consider the regime where the derivatives of $f$ satisfy  the uniform bound $\|f^{(k)}\|_\infty \;\le\; B\,\Lambda^k k!$, which is the relevant regime in \cref{sec:coefficient_learning}

\begin{corollary}[First derivative: sufficient $\tau_{\max}$, $r$, and $\varepsilon_s$]
    \label{cor:first_chebyshev_derivative_parameters_factorial} 
    Suppose the derivatives of $f$ satisfy the uniform bound
    \[
        \|f^{(k)}\|_\infty \;\le\; B\,\Lambda^k k! , \qquad (k \;\ge\; 1).
    \]
    Fix a target accuracy $\varepsilon>0$ and set
    \begin{equation}
    \tau_{\max}=\frac{1}{2\Lambda},
    \qquad
    r \;=\; \max\{16,\lceil 4 \log(\tfrac{8\Lambda B}{\varepsilon})\rceil\},
    \qquad
    \varepsilon_s \;=\; \frac{\varepsilon}{10\,\Lambda\,r^{3}}
    \end{equation}
    Then the derivative estimator $\widehat{p}^{\,\prime}(0)$ in \cref{eq:chebyshev_derivative_estimator_first} satisfies $\bigl|\widehat p^{\,\prime}(0)-f'(0)\bigr|\le \varepsilon$.
\end{corollary}

\begin{proof}
    We split the error budget equally, requiring (bias) $\le \varepsilon/2$ and (noise) $\le \varepsilon/2$.
    Set $\tau_{\max}=\tfrac{1}{2\Lambda}$, so $1/\tau_{\max}=2\Lambda$ and $(\Lambda\tau_{\max})^{r+1}=2^{-(r+1)}$.
    By \cref{theorem:chebyshev_derivative_total_error},
    \begin{equation}
        \text{bias}=2(r+1)^2\tau_{\max}^r\frac{\|f^{(r+1)}\|_\infty}{r!}
        \le \frac{2B(r+1)^3}{\tau_{\max}} (\Lambda \tau_{\max})^{r+1}
        = 4\Lambda B(r+1)^3 2^{-(r+1)},
    \end{equation}
    where we used that $\|f^{(r+1)}\|_\infty \;\le\; B\,\Lambda^{r+1}(r+1)!$. Let $C := \log(\tfrac{8\Lambda B}{\varepsilon})$. To achieve (bias) $\le \varepsilon/2$, we must choose $r$ such that $r+1 \ge C + 3\log(r+1)$. Choosing $r = \max\{16, 4C\}$, one can verify that the function $g(r) = r+1 - C - 3\log(r+1)$ is always nonnegative; therefore, the bias term is bounded by $\varepsilon/2$. For the noise term, again by \cref{theorem:chebyshev_derivative_total_error},
    \begin{equation}
        \text{noise}=\frac{5r^{3}}{2\tau_{\max}}\,\varepsilon_s
        =5\Lambda\,r^{3}\,\varepsilon_s
        \le \frac{\varepsilon}{2}
        \quad\Longrightarrow\quad
        \varepsilon_s\ \le\ \frac{\varepsilon}{10\,\Lambda\,r^{3}}.
    \end{equation}
    With these two conditions, $|\widehat p^{\,\prime}(0)-f'(0)|\le \varepsilon$.
\end{proof}

\subsection{Time-resolution of sampling times}\label{sec:time_resolution_chebyshev_nodes}

Finally, we prove that the spacing between the Gauss--Chebyshev is not too small. This is important for experimental feasibility, where one might have a bounded time-resolution:

\begin{lemma}
    \label{lemma:time_resolution_chebyshev_nodes}
    Let $r\ge 1$. Given $r+1$ Gauss--Chebyshev nodes $z_i=\cos\big(\tfrac{(2i+1)\pi}{2(r+1)}\big)$ for $i=0,\dots,r$ and the corresponding sampling times $t_i = \frac{z_i+1}{2}\tau_{\max}$, the smallest time resolution satisfies:
    \begin{equation}
        \min_{0\le i\le r-1} (t_i - t_{i+1}) \ge \frac{2\tau_{\max}}{(r+1)^2}.
    \end{equation}
\end{lemma}
\begin{proof}
    \begin{equation}
        \frac{2(t_i - t_{i+1})}{\tau_{\max}}
        = \cos\Big(\frac{(2i+1)\pi}{2(r+1)}\Big) - \cos\Big(\frac{(2i+3)\pi}{2(r+1)}\Big)
        = -2\sin\Big(\frac{(2i+2)\pi}{2(r+1)}\Big)\sin\Big(\frac{-\pi}{2(r+1)}\Big)
    \end{equation}
    Note the sequence $\{z_i\}$ is decreasing. For $r\ge 1$ and for any $0 \le i \le r-1$:
    \begin{equation}
        \frac{2(t_i - t_{i+1})}{\tau_{\max}}
        \;\ge\; 2\sin\Big(\frac{\pi}{r+1}\Big)\sin\Big(\frac{\pi}{2(r+1)}\Big)
        \;\ge\; \frac{4}{(r+1)^2},
    \end{equation}
    where we used $\sin{x} \ge 2x/\pi$ on $[0,\pi/2]$. Since this holds for any nearest sampling times, it also holds for the minimum. Rearranging concludes the proof.
\end{proof}
\section{Derivative bounds of Lindbladian dynamics }
\label{sec:lindbladian-derivative-bounds}

This section derives uniform bounds on time derivatives under Lindbladian dynamics, which underlie the performance guarantees for Chebyshev polynomial interpolation established in \cref{sec:chebyshev_interpolation}. Section~\cref{sec:deriv-bounds-pauli-error-rates} provides derivative bounds for Pauli error rates used in structure learning, while \cref{sec:dual-interaction-graph} provides corresponding derivative bounds for Pauli observables used in coefficient learning.

Use the Lindbladian expansion in the Pauli basis:
\begin{equation}
\label{eq:lindbladian_derivative_bounds_appendix}
\mathcal{L}(\rho)
=
-i\sum_{P_k\in \mathcal{S}_H} h_k\,[P_k,\rho]
\;+\;
\sum_{P_n,P_m\in \mathcal{S}_D} a_{mn}
\Big(P_n\rho P_m-\tfrac12\{P_mP_n,\rho\}\Big),
\end{equation}
where $\mathcal{S}_H$ and $\mathcal{S}_D$ are sets of Pauli operators indexing the Hamiltonian and dissipator structures, respectively. We use $M$ to denote an upper bound on the number of non-zero Lindbladian coefficients.

Without loss of generality, we assume the coefficients are normalized: $\max\Big\{\max_{P_k\in\mathcal{S}_H}|h_k|,\ \max_{\substack{P_m,P_n\in\mathcal{S}_D}}|a_{mn}|\Big\}=1.$

\subsection{Derivative bounds of Pauli error rates} \label{sec:deriv-bounds-pauli-error-rates}

In the structure learning stage (\cref{sec:structure_learning}), we estimate derivatives of the Pauli error rates $\chi_{ii}(t)$, hence the following lemma:

\begin{lemma}[Uniform bound on $k$-th derivatives of $\chi$-entries]
    \label{lemma:chi_derivative_uniform_bound}
    Let $t \ge 0$ and let the $n$-qubit channel $\mathcal{E}_t=e^{\mathcal{L}t}$ have $\chi$-matrix representation
    $\mathcal{E}_t(\rho)=\sum_{i,j}\chi_{ij}(t)\,P_i\rho P_j$. Then for all $k\ge 0$,
    \begin{equation}
        \abs{\frac{d^k\chi_{ij}(t)}{dt^k}} \equiv |\chi_{ij}^{(k)}(t)|
        \;\le\;
        \|\mathcal{L}\|_{2\rightarrow2}^{\,k},
    \end{equation}
    where $\|\mathcal{L}\|_{2\rightarrow2}
    := \sup_{X\neq 0} \| \mathcal{L}(X) \|_2 / \| X \|_2$
    denotes the induced Schatten-$2$ norm of the superoperator $\mathcal{L}$.
\end{lemma}
    
\begin{proof}
    Let $|\psi_0\rangle=\tfrac{1}{\sqrt{d}}\sum_x |x\rangle\otimes|x\rangle$ with $d=2^n$, and define $|\psi_i\rangle=(\mathbb{I}\otimes P_i)|\psi_0\rangle$. Let $\Lambda_t=(\mathrm{id}\otimes \mathcal{E}_t)\big(|\psi_0\rangle\langle\psi_0|\big)$ be the Choi state. Then
    \begin{equation}
        \chi_{ij}(t)=\langle\psi_i|\Lambda_t|\psi_j\rangle
        \;=\; \tr\big[\,|\psi_j\rangle\langle\psi_i|\,\Lambda_t\,\big].
    \end{equation}
    Since $\dv*{t}\mathcal{E}_t=\mathcal{L}\circ\mathcal{E}_t$ for a time-independent generator, we have
    \begin{equation}
        \dv[k]{t}\Lambda_t \;=\; (\mathrm{id}\otimes \mathcal{L}^k)(\Lambda_t),
        \qquad
        \chi_{ij}^{(k)}(t)
        \;=\;
        \tr\Big[\,|\psi_j\rangle\langle\psi_i|\,(\mathrm{id}\otimes \mathcal{L}^k)(\Lambda_t)\,\Big].
    \end{equation}
    Apply Cauchy--Schwarz for the Hilbert--Schmidt inner product:
    \begin{equation}
        \big|\chi_{ij}^{(k)}(t)\big|
        \;\le\;
        \big\||\psi_j\rangle\langle\psi_i|\big\|_2\;
        \big\|(\mathrm{id}\otimes \mathcal{L}^k)(\Lambda_t)\big\|_2
        \;=\;
        \big\|(\mathrm{id}\otimes \mathcal{L}^k)(\Lambda_t)\big\|_2,
    \end{equation}
    since $\big\||\psi_j\rangle\langle\psi_i|\big\|_2=1$. Using the definition of the induced norm $\norm{\cdot}_{2\rightarrow2}$,
    \begin{equation}
        \big\|(\mathrm{id}\otimes \mathcal{L}^k)(\Lambda_t)\big\|_2
        \;\le\;
        \norm{\mathrm{id}\otimes \mathcal{L}^k}_{2\rightarrow2}\;\|\Lambda_t\|_2.
    \end{equation}
    By stability of the induced Schatten-2 norm under tensoring with the identity, $\norm{\mathrm{id}\otimes \mathcal{L}^k}_{2\rightarrow2} = \norm{\mathcal{L}^k}_{2\rightarrow2}$. Moreover, $\Lambda_t$ is a density operator, so $\|\Lambda_t\|_2=\sqrt{\tr\{\Lambda_t^2\}}\le 1$. Therefore,
    \begin{equation}
        \big|\chi_{ij}^{(k)}(t)\big|
        \;\le\;
        \norm{\mathcal{L}^k}_{2\rightarrow2}
        \;\le\;
        \norm{\mathcal{L}}^k_{2\rightarrow2},
    \end{equation}
    by submultiplicativity of the induced norm.
\end{proof}

\begin{corollary}[$\chi$-derivative bound for $M$-sparse Lindbladians]
\label{cor:chi_derivative_bound_M_sparse}
Let $\mathcal{L}$ be of the form~\cref{eq:lindbladian_derivative_bounds_appendix}. Assume the coefficients are scaled so that
$|h_k|\le 1$ and $|a_{mn}|\le 1$ for all indices in~\cref{eq:lindbladian_derivative_bounds_appendix} and that there are at most $M$ non-zero coefficients. Then, for all
$i,j$, all integers $k\ge 0$, and all $t\ge 0$,
\begin{equation}
\label{eq:chi_derivative_bound_2M}
\big|\chi_{ij}^{(k)}(t)\big|
\;\le\;
\|\mathcal{L}\|_{2\to 2}^{\,k}
\;\le\; \left(\Delta(H)\;+\;2\sum_{m,n}|a_{mn}|\right)^k \;\le\;
(2M)^k,
\end{equation}
where $H:=\sum_{P_k\in\HamiltStruct} h_k P_k$ and $\Delta(H)=E_{\max}-E_{\min}$ is the spectral range of $H$. The second and third inequalities above follow directly from \cref{lemma:operator_norm_bounds}. 
\end{corollary}

\subsection{Dual interaction graph and derivative bounds of Pauli observables } \label{sec:dual-interaction-graph}

To provide performance guarantees for the polynomial interpolation used in the coefficient-learning stage
(\cref{sec:coefficient_learning}), we upper bound the time derivatives of Pauli observables in
\cref{thm:observable_derivative_bounds}. In contrast to Pauli error rates, which encode global channel-level information, coefficient learning
accesses derivatives of selected observables. When the interaction pattern is taken into account, this
permits derivative bounds that improve over worst-case sparsity-based estimates (e.g., \cref{lemma:chi_derivative_uniform_bound}). To formalize this, we generalize the dual interaction
graph introduced for Hamiltonians in~\cite{haah2024learning} to the Lindbladian setting.

Here, we use the adjoint master equation
\begin{equation} \label{eq:adjoint-master-eq}
    \Lindblad^\dagger(O)
    =
    i \sum_{P_k\in \mathcal{S}_H} h_k \comm{P_k}{O}
    + \sum_{P_k,P_m \in \mathcal{S}_D} a_{km}\left(P_m O P_k - \tfrac12\acomm{P_mP_k}{O}\right),
\end{equation}
where $\Lindblad^\dagger$ is the adjoint of $\Lindblad$ with respect to the Hilbert--Schmidt inner product: $
    \tr\big(O^\dagger \Lindblad(Q)\big)
    =
    \tr\big((\Lindblad^\dagger(O))^\dagger Q\big)$ for all $O,Q$.

\begin{definition}[Pauli Lindbladian components]
\label{def:pauli_lindbladian_components}
Let $\Lindblad^\dagger$ be an adjoint Lindbladian super-operator expanded in the Pauli basis from \cref{eq:adjoint-master-eq}. For each $P_k\in\HamiltStruct$ with $h_k\neq 0$, define the Pauli Hamiltonian component
\begin{equation*}
\label{eq:def_hamiltonian_component}
\mathcal{A}^{H}_k(O) := i h_k\,[P_k,O].
\end{equation*}
For each pair $(k,m)$ with $P_k,P_m\in\DiagDissStruct$ and $a_{km}\neq 0$, define the Pauli dissipative component
\begin{equation*}
\label{eq:def_dissipator_component}
\mathcal{A}^{D}_{km}(O) := a_{km}\Big(P_m O P_k-\tfrac12\{P_mP_k,O\}\Big).
\end{equation*}
Let $\mathcal{A}$ denote the collection of all such components:
\begin{equation}
\label{eq:def_action_set}
\mathcal{A} := \{\mathcal{A}^{H}_k : P_k\in\HamiltStruct,\; h_k\neq 0\}\ \cup\ \{\mathcal{A}^{D}_{km} : P_k,P_m\in\DiagDissStruct,\; a_{km}\neq 0\}.
\end{equation}

This allows to re-write $\Lindblad^\dagger$ as:
\begin{equation*}
    \Lindblad^\dagger (O) = \sum_{\mathcal{A}_k \in \mathcal A} \mathcal{A}_k(O)
\end{equation*}
\end{definition}

\begin{definition}[Pauli component support]
\label{def:component_support}
For a Pauli Hamiltonian component $\mathcal{A}^{H}_k$, define its support as
\begin{equation}
\supp{\mathcal{A}^{H}_k} := \supp{P_k}.
\end{equation}
For a Pauli dissipative component $\mathcal{A}^{D}_{km}$, define its support as
\begin{equation} \label{eq:dissipatove-component-support}
\supp{\mathcal{A}^{D}_{km}} := \supp{P_k}\cup\supp{P_m}.
\end{equation}
where $\supp{P}$ denotes the set of sites on which the Pauli $P$ acts nontrivially;
\end{definition}

\begin{definition}[Dual interaction graph via Pauli components]
\label{def:dual_interaction_graph}
The \emph{dual interaction graph} associated with $\Lindblad$ is the simple graph $\mathfrak{G}=(V,E)$ with vertex set representing the collection of all Pauli components
\begin{equation*}
V := \mathcal{A},
\end{equation*}
and whose edge set is defined by overlap of component supports:
\begin{equation*}
E \;:=\; \Bigl\{\{v_1,v_2\}\subseteq V \;:\; v_1\neq v_2,\ \supp{v_1}\cap\supp{v_2}\neq\emptyset\Bigr\}.
\end{equation*}
That is, vertices correspond to Lindbladian Pauli components, and two vertices are adjacent if and only if their component supports overlap. Denote the maximum degree of this dual interaction graph by
\begin{equation} \label{eq:dual-graph-degree}
\mathfrak d := \max_{v\in V}\deg_\mathfrak{G}(v).
\end{equation}
\end{definition}

\begin{definition}[Sparsely interacting Lindbladian]
\label{def:sparsely_interacting_lindbladian}
A Lindbladian $\mathcal{L}$ of the form~\eqref{eq:lindbladian_derivative_bounds_appendix} is \emph{sparsely interacting}
if the maximum degree $\mathfrak d$ of the dual interaction graph $\mathfrak{G}$ (\cref{def:dual_interaction_graph})
satisfies $\mathfrak d=O(1)$, i.e. is bounded by a constant independent of the system size $n$.
\end{definition}

Although our results (\cref{thm:observable_derivative_bounds,algorithm:coefficient_learning}) hold for general
maximum degree $\mathfrak d$, of particular interest is the case of \emph{sparsely interacting} Lindbladians
(\cref{def:sparsely_interacting_lindbladian}), where $\mathfrak d=O(1)$ is independent of the system size $n$.
This captures, for example, a broad class of \emph{locally generated} Markovian dynamics admitting a GKLS
representation
\begin{equation}\label{eq:diagonilized-adjoint-master-eq}
    \Lindblad^\dagger(O)
    =
    i \sum_{k} \alpha_k [E_k,O]
    + \sum_{j} \gamma_j\left(L_j^\dagger O L_j - \tfrac12\{L_j^\dagger L_j,O\}\right),
\end{equation}
in which each Hamiltonian term $E_k$ is at most $k_H$-local and each jump operator $L_j$ is at most $k_D$-local,
and moreover each qubit participates in at most $q_H$ of the $E_k$’s and at most $q_D$ of the $L_j$’s. Assume in addition that each $E_k$ is at most $m_H$-sparse in the Pauli basis and each $L_j$ is at most $m_D$-sparse in the Pauli basis. Such locality/bounded-participation assumptions are standard in local noise models (e.g.\ dephasing or amplitude
damping), locally monitored or reservoir-engineered systems, and local-contact open-system models. To bound the maximum degree $\mathfrak d$ of the dual interaction graph $\mathfrak G$
(\cref{def:dual_interaction_graph}), consider the qubit hypergraph $G$ with vertex set $[n]$ and one hyperedge
$e(\mathcal C):=\supp{\mathcal C}$ for each component $\mathcal C\in\mathcal A$
(\cref{def:pauli_lindbladian_components}). In \eqref{eq:diagonilized-adjoint-master-eq}, each $E_k$ contributes at most $m_H$ Pauli strings and hence at most $m_H$ Hamiltonian components supported on $\supp{E_k}$, so each qubit participates in at most $q_H m_H$ Hamiltonian components. Likewise, each $L_j$ contributes at most $m_D$ Pauli strings, and hence yields at most $m_D^2$ dissipative components $\mathcal A^D_{km}$ supported on $\supp{L_j}$ (cf.~\cref{eq:dissipatove-component-support}), so each qubit participates in at most $q_D m_D^2$ dissipative components. Thus each qubit is contained in at most $q_H m_H+q_D m_D^2$ component supports (and thus hyper-edges). Any component $\mathcal C$ satisfies $|\supp{\mathcal C}|\le \max(k_H,k_D)$, and for each qubit in $\supp{\mathcal
C}$ there are at most $(q_H m_H+q_D m_D^2-1)$ other components overlapping it. Therefore
\begin{equation}\label{eq:dual_degree_local_bound}
\mathfrak d \le \max(k_H,k_D)\bigl(q_H m_H+q_D m_D^2-1\bigr)
\le \max(k_H,k_D)\bigl(q_H4^{k_H}+q_D16^{k_D}-1\bigr).
\end{equation}

Another important example is the class of Lindbladians that are geometrically local on Euclidean space $\mathbb{R}^d$.
Assume a constant number of qubits per lattice site of $\mathbb{Z}^d\subset\mathbb{R}^d$, and that each Hamiltonian term
$E_k$ and jump operator $L_j$ is supported within a single unit hypercube.
Additionally, assume that each $E_k$ is at most $m_H$-sparse in the Pauli basis and each $L_j$ is at most $m_D$-sparse in
the Pauli basis. Since a unit hypercube intersects only hypercubes within $\ell_\infty$-distance $1$ (at most $3^d$ choices),
and each hypercube contributes at most $m_H$ Hamiltonian components and at most $m_D^2$ dissipative components, the dual
interaction graph has maximum degree
\begin{equation}\label{eq:dual_degree_lattice_bound}
\mathfrak d \;\le\; 3^d\,(m_H+m_D^2)-1.
\end{equation}

\begin{lemma}[Single-step Pauli growth on the dual interaction graph]
\label{lemma:single-step-pauli-growth-on-graph}
Let $\Lindblad^\dagger$ be an $n$-qubit adjoint Lindbladian of the form~\eqref{eq:adjoint-master-eq} where $\max(|h_k|, |a_{jm}|) \le 1$. Let
$\mathfrak{G}=(V,E)$ be its dual interaction graph from
\cref{def:dual_interaction_graph} with maximum degree $\mathfrak d$.
Let $\mathcal{A}$ be the collection of adjoint Lindbladian components from
\cref{def:pauli_lindbladian_components}, so that
\[
\Lindblad^\dagger(O)=\sum_{\mathcal{A}_k\in\mathcal{A}} \mathcal{A}_k(O).
\]

Let $O$ be an $n$-qubit Pauli observable and define its neighborhood
\begin{equation}
\label{eq:def_neighborhood_O_components}
\mathcal{N}(O)
:=
\{\mathcal{A}_k\in V : \supp{\mathcal{A}_k}\cap\supp{O}\neq\emptyset\},
\qquad
\mathfrak d_O:=|\mathcal{N}(O)|.
\end{equation}

Then $\Lindblad^\dagger(O)$ admits a decomposition
\begin{equation}
\label{eq:formal_single_step_decomp_components}
\Lindblad^\dagger(O) = \sum_{\mathcal{A}_k\in\mathcal{N}(O)} \mathcal{A}_k(O),
\end{equation}
with the following properties:
\begin{enumerate}
    \item (\textbf{Number of summands})
    The decomposition~\eqref{eq:formal_single_step_decomp_components} contains at most $\mathfrak d_O$ summands; each a scalar multiple of a Pauli.
    \item (\textbf{Summand magnitude})
    Each non-zero summand $\mathcal{A}_k(O)$ satisfies $\|\mathcal{A}_k(O)\|_\infty \le 2$
    \item (\textbf{Neighborhood growth})
    For each non-zero summand $\mathcal{A}_k(O)$, the following holds: $ |\mathcal{N}(\mathcal{A}_k(O))|
        \le \mathfrak d_O+\mathfrak d.$
\end{enumerate}
\end{lemma}

\begin{proof}
We prove each item in turn.

\paragraph{Decomposition and number of summands.}
By definition of the neighborhood,
if $\mathcal{A}_k\notin\mathcal{N}(O)$, then $\supp{\mathcal{A}_k}\cap\supp{O}=\emptyset$.
In this case $\mathcal{A}_k(O)=0$:
for a Hamiltonian component $\mathcal{A}^H_k(O)=-i h_k[P_k,O]$, disjoint support implies $[P_k,O]=0$; for a dissipative component
$\mathcal{A}^D_{km}(O)=a_{km}\big(P_mOP_k-\tfrac12\{P_mP_k,O\}\big)$, disjoint support implies $P_k$ and $P_m$ both commute with $O$, and hence
$P_mOP_k=OP_mP_k$ and $\{P_mP_k,O\}=2OP_mP_k$, so the bracket cancels and $\mathcal{A}^D_{km}(O)=0$.
Therefore,
\[
\Lindblad^\dagger(O)
=\sum_{\mathcal{A}_k\in\mathcal{A}}\mathcal{A}_k(O)
=\sum_{\mathcal{A}_k\in\mathcal{N}(O)}\mathcal{A}_k(O),
\]
which is the decomposition~\eqref{eq:formal_single_step_decomp_components}.
The number of summands is at most $|\mathcal{N}(O)|=\mathfrak d_O$ by construction.

Now, since a product of Paulis is a Pauli operator,  $\mathcal{A}_k^H(O) \propto P_kO$ and $\mathcal{A}_{km}^D(O)\propto P_kP_mO$, we get that each summand is indeed a scalar multiple of a Pauli. 

\paragraph{Summand magnitude.}
We bound the operator norm of each nonzero summand. Since $O$ is Pauli, $\|O\|_\infty=1$ and $\|P\|_\infty=1$ for all Paulis $P$.
First consider a Hamiltonian component:
\[
\|\mathcal{A}^H_k(O)\|_\infty
=|h_k|\,\|[P_k,O]\|_\infty
\le |h_k|\big(\|P_kO\|_\infty+\|OP_k\|_\infty\big)
=2|h_k|\le 2.
\]
Next consider a dissipative component:
\begin{align*}
\|\mathcal{A}^D_{km}(O)\|_\infty
&\le |a_{km}|\left(\|P_mOP_k\|_\infty+\tfrac12\|P_mP_kO\|_\infty+\tfrac12\|OP_mP_k\|_\infty\right) \\
&= |a_{km}|\left(1+\tfrac12+\tfrac12\right)
=2|a_{km}|
\le 2,
\end{align*}
using $\max(|h_k|,|a_{km}|)\le 1$. This proves the claimed bound $\|\mathcal{A}_k(O)\|_\infty\le 2$ for every nonzero summand.

\paragraph{Neighborhood growth.}
Fix a nonzero summand $\mathcal{A}_k(O)$ (either Hamiltonian or dissipative). We first show that its Pauli support is contained in
$\supp{\mathcal{A}_k}\cup\supp{O}$ using \cref{def:component_support}.

If $\mathcal{A}_k=\mathcal{A}^H_\ell$, then $\mathcal{A}^H_\ell(O)$ is proportional to $P_\ell O$, hence
\[
\supp{\mathcal{A}^H_\ell(O)} = \supp{P_\ell O}\subseteq \supp{P_\ell}\cup\supp{O}
= \supp{\mathcal{A}^H_\ell}\cup\supp{O}.
\]
If $\mathcal{A}_k=\mathcal{A}^D_{\ell m}$, then each term in $\mathcal{A}^D_{\ell m}(O)$ is proportional to either $P_m O P_\ell$ or $P_m P_\ell O$ (or $O P_m P_\ell$),
so in all cases
\[
\supp{\mathcal{A}^D_{\ell m}(O)} \subseteq  \supp{P_m}\cup\supp{P_\ell}\cup\supp{O}
= \supp{\mathcal{A}^D_{\ell m}}\cup\supp{O},
\]

Now let $v\in V$ satisfy
$\supp{v}\cap\supp{\mathcal{A}_k(O)}\neq\emptyset$, i.e.,
$v\in\mathcal{N}(\mathcal{A}_k(O))$.
Since
$\supp{\mathcal{A}_k(O)}
\subseteq
\supp{\mathcal{A}_k}\cup\supp{O}$,
any such $v$ must overlap either $\supp{O}$ or
$\supp{\mathcal{A}_k}$. Let
\[
\operatorname{Nbr}_{\mathfrak G}(\mathcal A_k)
:=
\{v\in V:\{v,\mathcal A_k\}\in E\}
\]
denote the graph neighbors of $\mathcal A_k$, excluding $\mathcal A_k$
itself. Since $\mathcal A_k(O)\neq 0$, we necessarily have
$\mathcal A_k\in\mathcal N(O)$. Therefore
\[
\mathcal{N}(\mathcal{A}_k(O))
\subseteq
\mathcal{N}(O)
\cup
\operatorname{Nbr}_{\mathfrak G}(\mathcal A_k).
\]
Taking cardinalities and using
$|\mathcal N(O)|=\mathfrak d_O$ and
$|\operatorname{Nbr}_{\mathfrak G}(\mathcal A_k)|
\le \mathfrak d$ gives
\[
|\mathcal{N}(\mathcal{A}_k(O))|
\le
\mathfrak d_O+\mathfrak d,
\]
which proves the neighborhood growth claim.
\end{proof}

\begin{corollary}[Iterated Pauli growth]
\label{cor:iterated-pauli-growth}
Using the setting of \cref{lemma:single-step-pauli-growth-on-graph}, for any integer $k\ge 1$ we may write
\begin{equation}
\label{eq:iterated_uncollected_decomposition}
(\Lindblad^\dagger)^k(O)
=
\sum_{\mathcal{A}_{j_1}\in\mathcal{A}}
\sum_{\mathcal{A}_{j_2}\in\mathcal{A}}
\cdots
\sum_{\mathcal{A}_{j_k}\in\mathcal{A}}
\mathcal{A}_{j_1}\big(\mathcal{A}_{j_2}(\cdots \mathcal{A}_{j_k}(O))\big),
\end{equation}
This uncollected decomposition contains at most
\[
\mathfrak d_O\,(\mathfrak d_O+\mathfrak d)\cdots(\mathfrak d_O+(k-1)\mathfrak d)
\]
nonzero Pauli summands, and each summand has operator norm at most $2^k$.
\end{corollary}

\begin{proof}
The decomposition follows by iterating the expansion of \cref{lemma:single-step-pauli-growth-on-graph}.
At the first step there are at most $\mathfrak d_O$ summands, and by the neighborhood-growth property of the lemma,
after $\ell$ steps each summand has at most $\mathfrak d_O+\ell\mathfrak d$ descendants.
Multiplying these bounds yields the stated count.
Finally, each application of a component increases the operator norm by at most a factor of $2$,
so every $k$-fold nested term has norm at most $2^k$.
\end{proof}

\begin{theorem}[Uniform bound on Pauli observable derivatives]
\label{thm:observable_derivative_bounds}
Let $\mathcal{L}$ be an $n$-qubit Lindbladian of the form~\eqref{eq:lindbladian_derivative_bounds_appendix} and let
$\mathfrak{G}=(V,E)$ be its dual interaction graph from \cref{def:dual_interaction_graph} with maximum degree $\mathfrak d$.
Let $O$ be an $n$-qubit Pauli observable with neighborhood size
$\mathfrak d_O:=|\mathcal{N}(O)|$ (cf.~\cref{eq:def_neighborhood_O_components}).
For any Hermitian input operator $X$ satisfying $\|X\|_1\le 1$, define
\[
f_{O,X}(t):=\tr\bigl(O\,e^{t\mathcal L}(X)\bigr).
\]
Then for every integer $k\ge 1$ and all $t\ge 0$,
\begin{equation}
\label{eq:observable_derivative_bound_statement}
\left|
\frac{d^k}{dt^k} f_{O,X}(t)
\right|
\le
\left(2\max\{\mathfrak d_O,\mathfrak d\}\right)^k k!.
\end{equation}
\end{theorem}

\begin{proof}
Let $X_t:=e^{t\mathcal L}(X)$. Since $X$ is Hermitian and $\|X\|_1\le1$, trace-norm contractivity under CPTP maps gives
$\|X_t\|_1\le1$. For $k\ge1$,
\begin{equation}
\label{eq:observable_derivative_trace_adjoint_thm}
\frac{d^k}{dt^k} f_{O,X}(t)
=
\tr\bigl((\mathcal L^\dagger)^k(O)\,X_t\bigr),
\end{equation}
and hence H\"older's inequality implies
\begin{equation}
\label{eq:observable_derivative_holder_thm}
\left|
\frac{d^k}{dt^k} f_{O,X}(t)
\right|
\le
\bigl\|(\mathcal L^\dagger)^k(O)\bigr\|_\infty.
\end{equation}
By \cref{cor:iterated-pauli-growth}, $(\mathcal L^\dagger)^k(O)$ admits an uncollected decomposition into at most
\[
\mathfrak d_O
(\mathfrak d_O+\mathfrak d)
\cdots
(\mathfrak d_O+(k-1)\mathfrak d)
\]
nonzero Pauli summands, each of operator norm at most $2^k$. Therefore
\begin{equation}
\label{eq:Ldaggerk_triangle_bound_thm}
\bigl\|(\mathcal L^\dagger)^k(O)\bigr\|_\infty
\le
2^k
\mathfrak d_O
(\mathfrak d_O+\mathfrak d)
\cdots
(\mathfrak d_O+(k-1)\mathfrak d).
\end{equation}
For every
$\ell=0,\ldots,k-1$,
\[
\mathfrak d_O+\ell\mathfrak d
\le
(\ell+1)\max\{\mathfrak d_O,\mathfrak d\},
\]
and hence
\begin{equation}
\label{eq:factorial_bound_thm}
\mathfrak d_O
(\mathfrak d_O+\mathfrak d)
\cdots
(\mathfrak d_O+(k-1)\mathfrak d)
\le
\max\{\mathfrak d_O,\mathfrak d\}^{\,k} k!.
\end{equation}
Combining the above bounds proves \eqref{eq:observable_derivative_bound_statement} for $k\ge1$.
For $k=0$, H\"older's inequality gives
$|f_{O,X}(t)|\le \|O\|_\infty\|X_t\|_1\le1$, completing the proof.
\end{proof}

\begin{remark}[Derivative bound for $M$-sparse Lindbladians]
\label{remark:observable_derivative_bounds_M_sparse}
If $\mathcal L$ contains at most $M$ nonzero Pauli components, then
for all $k\ge1$ and $t\ge0$,
\[
\left|
\frac{d^k}{dt^k}f_{O,X}(t)
\right|
\le (2M)^k,
\]
since the uncollected expansion of
$(\mathcal L^\dagger)^k(O)$ contains at most $M^k$
nonzero summands, each of operator norm at most $2^k$.
\end{remark}

\subsection{Lindbladian Induced Schatten-2 norm bound}\label{sec:induced_norm_bound}

In this subsection, we upper bound the induced Lindbladian Schatten-$2$ norm $\|\mathcal{L}\|_{2\to2}$ as defined in \cref{sec:appendix_notation_prelim}. This norm is used by the Pauli error rates derivative bound in \cref{cor:chi_derivative_bound_M_sparse}.

We briefly recall the operator--vector correspondence (``vectorization'') following \cite[Eq.~(1.127)]{Watrous_2018}, except we use a tensor-factor ordering corresponding to column-stacking (\cref{eq:vec_def}) rather than row-stacking.
Let $\{|a\rangle\}_{a=1}^d$ be the computational basis and $E_{ab}:=|a\rangle\langle b|$ the matrix units.
We define the linear map $\vectorized{\cdot}:\mathbb{C}^{d\times d}\to \mathbb{C}^{d}\otimes\mathbb{C}^{d}$ by
\begin{equation}
\label{eq:vec_def}
\vectorized{E_{ab}} := |b\rangle\otimes|a\rangle,
\qquad\text{extended by linearity.}
\end{equation}

With this choice one has, for all $A,X,B\in\mathbb{C}^{d\times d}$,
\begin{equation}
\label{eq:vec_AXB}
\vectorized{AXB} \;=\; (B^{T}\otimes A)\,\vectorized{X}.
\end{equation}
Moreover, $\vectorized{\cdot}$ is an isometry for the Hilbert--Schmidt inner product:
\begin{equation}
\label{eq:vec_isometry}
\langle \vectorized{X},\vectorized{Y}\rangle \;=\; \tr(X^\dagger Y),
\qquad\text{so in particular}\qquad
\|\vectorized{X}\|_2=\|X\|_2 .
\end{equation}
Here $\|\cdot\|_2$ on $\mathbb{C}^d\otimes\mathbb{C}^d$ denotes the Euclidean ($\ell_2$) norm, while $\|X\|_2$ on $\mathbb{C}^{d\times d}$ denotes the Schatten-$2$ (Hilbert--Schmidt) norm.

Given a linear map (superoperator) $\Phi$, we define its matrix (Liouville) representation $\widehat{\Phi}$ by
\begin{equation}
\label{eq:liouville_def}
\vectorized{\Phi(X)} \;=\; \widehat{\Phi}\,\vectorized{X}\qquad\forall\,X.
\end{equation}
Then the induced Schatten-$2$ norm of $\Phi$ coincides with the operator norm of its Liouville matrix representation, i.e., the Schatten-$\infty$ norm:
\begin{equation}
\label{eq:induced_equals_matrix}
\|\Phi\|_{2\to 2}
:=\sup_{X\neq 0}\frac{\|\Phi(X)\|_2}{\|X\|_2}
\;=\;
\sup_{v\neq 0}\frac{\|\widehat{\Phi}v\|_2}{\|v\|_2}
\;=\;\|\widehat{\Phi}\|_{\infty}.
\end{equation}
In particular, for the Lindbladian $\Lindblad$ we upper bound $\|\Lindblad\|_{2\to 2}$ via $\|\widehat{\Lindblad}\|_{\infty}$.

\begin{lemma}[Induced Schatten-$2$ norm bound for Lindbladians]
\label{lemma:operator_norm_bounds}
Recall the Lindbladian expansion in the Pauli basis from \cref{eq:lindbladian_derivative_bounds_appendix} and write
$\Lindblad=\mathcal{H}+\mathcal{D}$ for its Hamiltonian and dissipative parts, respectively.
Assume the coefficients have magnitude less than $1$ as in \cref{eq:lindbladian_derivative_bounds_appendix},
and assume there are at most M non-zero Lindbladian coefficients.

Then
\begin{equation}
\label{eq:phys_bound}
\|\Lindblad\|_{2\to 2}
\;\le\;
\Delta(H)\;+\;2\sum_{m,n}|a_{mn}|
\;\le\;
2M,
\end{equation}
where $H:=\sum_{P_k\in\HamiltStruct} h_k P_k$ and $\Delta(H)=E_{\max}-E_{\min}$ is the spectral range of $H$.
\end{lemma}

\begin{proof}
By \cref{eq:induced_equals_matrix}, for any superoperator $\Phi$ we have
\[
\|\Phi\|_{2\to2}=\|\widehat{\Phi}\|_{\infty},
\]
where $\widehat{\Phi}$ is the Liouville matrix defined by
$\vectorized{\Phi(X)}=\widehat{\Phi}\,\vectorized{X}$.
Using subadditivity of the Schatten-$\infty$ norm,
\begin{equation}
\label{eq:triangle_liouville}
\|\Lindblad\|_{2\to2}
=\|\widehat{\Lindblad}\|_{\infty}
\le \|\widehat{\mathcal{H}}\|_{\infty}+\|\widehat{\mathcal{D}}\|_{\infty}.
\end{equation}

\paragraph{Hamiltonian part.}
In an eigenbasis $H\ket{j}=E_j\ket{j}$,
\[
\mathcal{H}(\ketbra{j}{k})=-i(E_j-E_k)\ketbra{j}{k}.
\]
Since $\{\ketbra{j}{k}\}_{j,k}$ is orthonormal in the Hilbert--Schmidt inner product, $\mathcal{H}$ is diagonal in this orthonormal operator basis and hence
\begin{equation}
\label{eq:H_norm}
\|\widehat{\mathcal{H}}\|_{\infty}=\|\mathcal{H}\|_{2\to2}=\max_{j,k}|E_j-E_k|=\Delta(H).
\end{equation}

\paragraph{Dissipator part.}
Using \cref{eq:vec_AXB}, each term satisfies
\[
\vectorized{P_m X P_n}=(P_n^T\otimes P_m)\vectorized{X},\qquad
\vectorized{(P_nP_m)X}=(I\otimes P_nP_m)\vectorized{X},\qquad
\vectorized{X(P_nP_m)}=((P_nP_m)^T\otimes I)\vectorized{X}.
\]
Therefore the Liouville representation of $\mathcal{D}$ is
\begin{equation}
\label{eq:D_liouville}
\widehat{\mathcal{D}}
=
\sum_{m,n} a_{mn}\left(
P_n^{T}\otimes P_m
-\tfrac12\, I\otimes P_nP_m
-\tfrac12\, (P_nP_m)^{T}\otimes I
\right).
\end{equation}
Taking the Schatten-$\infty$ norm and using the triangle inequality gives
\begin{align}
\|\widehat{\mathcal{D}}\|_{\infty}
&\le
\sum_{m,n}|a_{mn}|\,
\left\|
P_n^{T}\otimes P_m
-\tfrac12\, I\otimes P_nP_m
-\tfrac12\, (P_nP_m)^{T}\otimes I
\right\|_{\infty} \nonumber\\
&\le
\sum_{m,n}|a_{mn}|\left(
\|P_n^{T}\otimes P_m\|_{\infty}
+\tfrac12\|I\otimes P_nP_m\|_{\infty}
+\tfrac12\|(P_nP_m)^{T}\otimes I\|_{\infty}
\right) \nonumber\\
&=
\sum_{m,n}|a_{mn}|\left(1+\tfrac12+\tfrac12\right)
=
2\sum_{m,n}|a_{mn}|.
\label{eq:D_norm_bound}
\end{align}
Here we used that Paulis are unitary, hence have Schatten-$\infty$ norm $1$, and
$\|A\otimes B\|_{\infty}=\|A\|_{\infty}\|B\|_{\infty}$.

Combining \cref{eq:triangle_liouville,eq:H_norm,eq:D_norm_bound} yields the first inequality in \cref{eq:phys_bound}.

\paragraph{Crude sparsity bound.}
For the Hamiltonian part, use the Liouville picture: for $\mathcal{H}_k(X):=-i[P_k,X]$ one has
\[
\widehat{\mathcal{H}}_k=-i\big(I\otimes P_k - P_k^{T}\otimes I\big),
\]
so $\|\widehat{\mathcal{H}}_k\|_{\infty}\le \|I\otimes P_k\|_{\infty}+\|P_k^T\otimes I\|_{\infty}=2$.
Hence, writing $\mathcal{H}=\sum_{P_k\in\HamiltStruct} h_k\,\mathcal{H}_k$ and using $|h_k|\le 1$,
\[
\|\mathcal{H}\|_{2\to2}=\|\widehat{\mathcal{H}}\|_{\infty}
\le \sum_{P_k\in\HamiltStruct}|h_k|\,\|\widehat{\mathcal{H}}_k\|_{\infty}
\le 2|\HamiltStruct|.
\]
Substituting these into the first inequality in \cref{eq:phys_bound} gives
$\|\Lindblad\|_{2\to2}\le 2M$.

\end{proof}

\section{Structure learning}\label{sec:structure_learning}

In this section we describe the Lindbladian structure learning algorithm -- the primitive to identify the operator content of both the Hamiltonian and the dissipator. First, we set up the problem. Express the Lindbladian in the Pauli basis:
\begin{equation} \label{eq:lindbladian_struct_learning_appendix}
        \Lindblad(\rho) = -i\underbrace{\sum_{P_k\in \HamiltStruct} h_k \comm{P_k}{\rho}}_{\mathrm{Hamiltonian}}+\underbrace{\sum_{P_k,P_j \in \DiagDissStruct} a_{kj}\!\left(P_k \rho P_j - \tfrac12\acomm{P_j P_k}{\rho}\right)}_{\mathrm{Dissipator}},
\end{equation}
 where $a\succeq 0$ is commonly called a Kossakowski matrix. We denote by $\eta_H := \min_{P_i \in \HamiltStruct} |h_i|$ and
$\eta_D := \min_{P_i,P_j \in \DiagDissStruct} |a_{ij}|$ the smallest nonzero Hamiltonian and dissipator
coefficients, respectively, and set $\eta := \min\{\eta_H,\eta_D\}$. In practice, one may treat $\eta$ as a
user-chosen resolution threshold satisfying
$\eta \le \min\!\left(\min_{h_{i}\neq 0}|h_i|,\ \min_{a_{ij}\neq 0}|a_{ij}|\right)$.
See \ref{remark:eta-heavy-support} for when $\eta$ is chosen to be larger than the true minimum nonzero magnitude.

WLOG, we rescale $\Lindblad$ such that the coefficients are normalized: $\max\Big\{\max_{P_k\in\mathcal{S}_H}|h_k|,\ \max_{\substack{P_m,P_n\in\mathcal{S}_D}}|a_{mn}|\Big\}=1.$

\begin{definition}[Hamiltonian structure]    \label{def:hamiltonian_structure}
    The Hamiltonian structure is the set $\HamiltStruct$ of all Pauli terms with nonzero coefficients in the Hamiltonian expansion:
    \begin{equation}
        \HamiltStruct := \{\, P_i \;\;:\;\; h_i \neq 0 \,\}.
    \end{equation}
\end{definition}

\begin{definition}[Dissipator structure]
    \label{def:dissipator_structure}
    The dissipator structure is the set $\DiagDissStruct$ of all Pauli terms with nonzero diagonal Kossakowski entries \footnote{For a positive semidefinite matrix $a$, one has $|a_{ij}|^2 \le a_{ii}a_{jj}$. Thus, if $a_{ii}=0$, it follows that $a_{ij}=a_{ji}=0$ for all $j$. Hence the dissipator support is fully captured by the nonzero diagonal entries.}:
    \begin{equation}
        \DiagDissStruct := \{\, P_i \;\;|\;\; a_{ii} \neq 0 \,\}.
    \end{equation}
\end{definition}

Let $M_H := |\HamiltStruct|$ and $M_D := |\DiagDissStruct|$ be the support sizes of the Hamiltonian and dissipator. Recall that we use $M$ to denote an upper bound on the number of non-zero Lindbladian coefficients.

\begin{problem}[Structure learning problem]
    \label{problem:structure_learning}
    Given black-box access to the  dynamics $e^{\mathcal{L}t}$ for arbitrary $t\ge0$, identify two sets of Pauli operators, $\widehat{\mathcal{S}}_H$ and $\widehat{\mathcal{S}}_D$, containing the true Hamiltonian and dissipator supports,
    \begin{equation}
        \HamiltStruct \subseteq \widehat{\mathcal{S}}_H,
        \qquad
        \DiagDissStruct \subseteq \widehat{\mathcal{S}}_D,
    \end{equation}
    with high probability, while ensuring that both sets have polynomial size,
    \begin{equation}
        |\widehat{\mathcal{S}}_H|,\, |\widehat{\mathcal{S}}_D| 
        = \mathrm{poly}(M,1/\eta).
    \end{equation}
    In words, the task is to recover, up to polynomial overhead, all Pauli terms with nonzero coefficients in the Hamiltonian and the dissipator.
\end{problem}

The central challenge of the \nameref{problem:structure_learning} is that even if the Lindbladian $\mathcal{L}$ is sparse in the Pauli basis, the corresponding channel $\mathcal{E}_t = e^{\mathcal{L}t}$ is, in general, dense in that basis. Therefore, the structure of $\mathcal{L}$ cannot be inferred straightforwardly from the structure of $\mathcal{E}_t$. To address that, we focus on the evolution near $t \approx 0$, where the channel can be approximated by its Taylor expansion 
\begin{equation} 
    \mathcal{E}_t \;\approx\;  \mathrm{id} + \mathcal{L}t + \mathcal{L}^2 \tfrac{t^2}{2}, 
\end{equation}

An analytically simpler object to work with is the $\chi$-matrix representation of the Lindbladian evolution channel:
\begin{equation}
    \mathcal{E}_t(\rho) \;=\; \sum_{i,j} \chi_{ij}(t)\, P_i \rho P_j.
\end{equation}

In this section we describe how to access and learn the Lindbladian structure from the time-derivatives of the diagonal entries $\{\chi_{ii}(t)\}$, which are commonly referred to as the Pauli error rates. In particular, this section proceeds with the following outline:

\begin{itemize}
    \item \textbf{Short-time derivative identities.}  
    We first show that the initial time-derivatives of the Pauli error rates $\{\chi_{ii}\}$ directly encode the Lindbladian structure. This is formalized in \cref{lemma:time_derivatives_pauli_rates}.
    
    \item \textbf{Estimation of Pauli error rates.}  
    Next, we establish how to efficiently estimate the diagonal entries $\{\chi_{ii}(t)\}$ of the $\chi$-matrix at arbitrary times without ancillas. This is achieved by the imported population-recovery procedure (\cref{theorem:population_recovery}, \cite{flammia2021pauli}).
    
    \item \textbf{Derivative estimation via Chebyshev interpolation.}  
    Given noisy estimates of $\chi_{ii}(t)$ at a set of Gauss–Chebyshev nodes, we fit a degree-$r$ Chebyshev interpolant and use its analytic derivatives at $t=0$ as estimators for $\chi_{ii}^{(1)}(0)$ and $\chi_{ii}^{(2)}(0)$.  
    Their total estimation error is quantified in \cref{theorem:chebyshev_derivative_total_error}, with required parameters summarized in \cref{cor:first_chebyshev_derivative_parameters,cor:second_chebyshev_derivative_parameters}.
\end{itemize}

We combine these ingredients into the Hamiltonian and Dissipator structure learning \cref{algorithm:dissipator_structure_learning,algorithm:hamiltonian_structure_learning}. We then prove their correctness and sample complexity in \cref{theorem:ancilla_free_structure_learning}.

\subsection{Pauli error rates derivative identities}\label{sec:pauli_rates_short_time_lindbladian}

The diagonal entries $\{\chi_{ii}(t)\}$ of the $\chi$-matrix, commonly referred to as the Pauli error rates, are central to our protocol for two reasons: (i) their time-derivatives at $t=0$ reveal the operator content of the Lindbladian, and (ii) they can be estimated efficiently (see \cref{sec:pauli_error_rates}).

To see why (i) holds, consider the Taylor expansion of the channel:
\begin{equation}
    \mathcal{E}_t(\rho) 
    \;=\; \sum_{k=0}^{\infty}\mathcal{L}^k(\rho) \,\frac{t^k}{k!}
    \;=\; \sum_{i,j} \chi_{ij}(t)\, P_i \rho P_j 
    \;=\; \sum_{k=0}^{\infty} \left( \sum_{i,j} \chi_{ij}^{(k)}\, P_i \rho P_j \right)\frac{t^k}{k!},
\end{equation}
where in the last equality we have expanded the coefficients as $\chi_{ij}(t)=\sum_{k=0}^\infty \chi_{ij}^{(k)} t^k/k!$, and defined $\chi_{ij}^{(k)}=\dv[k]{t}\chi_{ij}(0)$. By comparing terms, we see that $\mathcal{L}^k(\rho) \;=\; \sum_{i,j} \chi_{ij}^{(k)}\, P_i \rho P_j$, and in particular
\begin{equation}
    \mathcal{L}(\rho) \;=\; \sum_{i,j} \chi_{ij}^{(1)}\, P_i \rho P_j,
    \qquad
    \mathcal{L}^2(\rho) \;=\; \sum_{i,j} \chi_{ij}^{(2)}\, P_i \rho P_j.
\end{equation}
Comparing with the Lindblad form (\cref{eq:lindbladian_struct_learning_appendix}), we immediately see that the first derivatives of the diagonal coefficients correspond to the diagonal dissipator matrix elements: $\chi_{ii}^{(1)} \;=\; a_{ii}$. Thus, if we could estimate all nonzero $\chi_{ii}^{(1)}$ efficiently, we could immediately determine the dissipator structure.

The Hamiltonian structure is less immediate. Because $-i\comm{H}{\rho}$ only acts on one side of $\rho$, the Hamiltonian does not appear in the first derivatives of the diagonals. To ``see'' the Hamiltonian contributions, we need to go to the second order. The explicit form of the second time-derivative $\chi_{ii}^{(2)}$ is quite involved, as it includes quadratic terms in $a_{ij}$ and $h_i$. However, for the purpose of structure identification, the following result is sufficient:

\begin{lemma}[Time-derivatives of Pauli error rates of $e^{\mathcal{L}t}$]
    \label{lemma:time_derivatives_pauli_rates}
    Let $\mathcal{E}_t = e^{\mathcal{L}t}$ be the channel generated by an $n$-qubit Lindbladian $\mathcal{L}$, with $\chi$-matrix representation $\mathcal{E}_t(\rho) \;=\; \sum_{i,j} \chi_{ij}(t)\, P_i \rho P_j$. Then the first and second derivatives of the diagonal coefficients at $t=0$ satisfy
    \begin{align}
        \chi_{ii}^{(1)} &= a_{ii}, 
        && \text{if } a_{ii} > 0,\\[4pt]
        \chi_{ii}^{(1)} &= 0, \quad 
        \chi_{ii}^{(2)} \;\ge\; 2h_i^{2},
        && \text{if } a_{ii} = 0.
    \end{align}
\end{lemma}

\begin{proof}
    Expanding the channel,
\begin{equation}
    \label{eq:lindblad-to-channel-taylor-expansion}
    \mathcal{E}_t(\rho) 
    = \sum_{r=0}^{\infty}\frac{\mathcal{L}^r(\rho)\,t^r}{r!}
    = \sum_{i,j} \chi_{ij}(t)\, P_i \rho P_j,
\end{equation}
with $\mathcal{L}^0=\operatorname{id}$, the entries admit $\chi_{ij}(t)=\sum_{r=0}^{\infty}\chi_{ij}^{(r)} t^r/r!$. 

Throughout this proof we repeatedly match coefficients between the left- and right-hand sides of
\eqref{eq:lindblad-to-channel-taylor-expansion}: we match powers of $t$ using uniqueness of the Taylor
expansion, and within each order we match the coefficients of $P_i\rho P_j$ using that the superoperators
$\{\rho\mapsto P_i\rho P_j\}$ are linearly independent across $\{i,j\}$.

We now consider the term $P_k \rho P_k$ and determine the first two derivatives: $\chi_{kk}^{(1)}$ and $\chi_{kk}^{(2)}$. Clearly, $\chi_{kk}^{(0)}=\delta_{k0}$ since $\mathcal{E}_0=\mathrm{id}$. The first derivative $\chi_{kk}^{(1)}$ can be read off directly from the Lindbladian by matching the order-$t$ terms in \eqref{eq:lindblad-to-channel-taylor-expansion},
\begin{equation}
    \mathcal{L}(\rho) = -i[H,\rho] + \mathcal{D}(\rho)
    = -i\sum_i h_i \big(P_i \rho - \rho P_i\big)
    + \sum_{i,j} a_{ij}\left(P_i \rho P_j - \tfrac{1}{2}\{P_j P_i,\rho\}\right).
\end{equation}
Notice that the Hamiltonian terms $P_k \rho$ and $\rho P_k$ act on only one side of $\rho$, so they contribute only to off-diagonal $\chi$-entries $\chi_{0k}^{(1)}$ and $\chi_{k0}^{(1)}$ and do not affect $\chi_{kk}^{(1)}$. The dissipator, however, does generate diagonal contributions of the form $a_{kk}\,P_k \rho P_k$, which feed directly into $\chi_{kk}$. Hence, as stated in \cref{lemma:time_derivatives_pauli_rates}:
\begin{equation}
    \chi_{kk}^{(1)} = a_{kk} \quad (k \neq 0), 
    \qquad 
    \chi_{00}^{(1)} = -\sum_{i\ne0} a_{ii}.
\end{equation}
To ``see'' the Hamiltonian, we need to go to the second order in $t$:
\begin{equation}
    \mathcal{L}^2(\rho) = -\underbrace{\comm{H}{\comm{H}{\rho}}}_{\scalebox{0.7}{\boxed{1}}\;:\;\text{Hamiltonian}}
    \;-\; i\,\underbrace{\big(\comm{H}{\mathcal{D}(\rho)} + \mathcal{D}(\comm{H}{\rho})\big)}_{\scalebox{0.7}{\boxed{2}}\;:\;\text{Mixed terms}}
    \;+\; \underbrace{\mathcal{D}(\mathcal{D}(\rho))}_{\scalebox{0.7}{\boxed{3}}\;:\;\text{Dissipator}}.
    \label{eq:second_order_lindbladian}
\end{equation}
This decomposition separates the Hamiltonian, mixed, and dissipative contributions. In principle, all three can generate diagonal $\chi_{kk}(t)$ elements. For the Hamiltonian contribution, we compute
\begin{align}
    \boxed{1}\;:\; -\comm{H}{\comm{H}{\rho}}
    &= -\sum_i h_i \left(P_i \comm{H}{\rho} - \comm{H}{\rho} P_i\right) \\
    &= -\sum_{i,i'} h_i h_{i'}\left(P_i P_{i'} \rho - P_i \rho P_{i'} - P_{i'} \rho P_i + \rho P_{i'} P_i\right).
\end{align}
Among these, diagonal terms of the form $2h_i^2\, P_i \rho P_i$ appear explicitly. Hence, at order $t^2$, the Hamiltonian contributes
\begin{equation}
    \chi_{kk}^{(2)} \leftarrow 2h_k^2 \quad (k \neq 0), 
    \qquad 
    \chi_{00}^{(2)} \leftarrow -2\sum_i h_i^2.
\end{equation}
Thus the Hamiltonian structure, which is invisible at first order, becomes, at least in principle, detectable from the second-order expansion of the  error rates. Unfortunately, $\chi_{kk}^{(2)}$ are ``obscured'' by the mixed terms and second-order dissipator contributions. 
The key observation is that if $a_{kk} = 0$, i.e.\ $P_k$ is not part of the dissipator structure, then any spurious contributions either vanish or only increase the magnitude of $\chi_{kk}^{(2)}$. The crucial point is that whenever $a_{kk} = 0$, we also have $a_{ik} = a_{ki} = 0 \: \forall i$, since $a$ is positive semi-definite. In what follows, we only consider the $\chi_{kk}^{(2)}$ terms corresponding to $a_{kk} = 0$.

Let us first look at the mixed terms:
\begin{align}
    \boxed{2.1} \quad -i\,\comm{H}{\mathcal{D}(\rho)} &= -i \sum_{i} h_i \Big(P_i \mathcal{D}(\rho) - \mathcal{D}(\rho) P_i\Big) \\
    &= -i \sum_{i,m,n} h_i\,a_{mn}
    \Big(
     P_i P_m \rho P_n
   - P_m \rho P_n P_i
   - \tfrac12 P_i P_n P_m \rho \\
   &\hspace{2.5cm}-\tfrac12 P_i \rho P_n P_m
   + \tfrac12 P_n P_m \rho P_i
   + \tfrac12 \rho P_n P_m P_i
    \Big) \\[0.5cm]
    \boxed{2.2} \quad -i\,\mathcal{D}\left(\comm{H}{\rho}\right)
    &= -i \sum_{m,n} a_{mn} \left( P_m \comm{H}{\rho} P_n
       - \tfrac12 \acomm{P_n P_m}{\comm{H}{\rho}} \right) \\
    &= -i \sum_{i,m,n} h_i\,a_{mn}
    \Big(
      P_m P_i \rho P_n
      - P_m \rho P_i P_n
      - \tfrac12 P_n P_m P_i \rho \\
   &\hspace{2.5cm}+ \tfrac12 P_n P_m \rho P_i
      - \tfrac12 P_i \rho P_n P_m
      + \tfrac12 \rho P_i P_n P_m
    \Big)
\end{align}
We can now simplify. Terms where $P_n$ or $P_m$ act alone on one side of $\rho$ are irrelevant, since their coefficient must be $h_i a_{kn}$ (or $h_i a_{mk}$), which vanishes by the assumption $a_{mk} = a_{kn} = 0 \: \forall m,n$. Moreover, pairs of terms such as $\tfrac12 P_n P_m \rho P_i - \tfrac12 P_i \rho P_n P_m$ cancel whenever $P_i = P_k$ and $P_n P_m \propto P_k$. Altogether, this shows that when $a_{kk} = 0$, the mixed terms $\scalebox{0.9}{\boxed{2.1}}$ and $\scalebox{0.9}{\boxed{2.2}}$ do not contribute to $\chi_{kk}^{(2)}$.

Now consider the second-order dissipator terms $\scalebox{0.9}{\boxed{3}}$. It will be convenient to write the dissipator as 
\begin{equation}
    \mathcal{D}(\rho) = \Phi(\rho) - \tfrac12 J(\rho) \qq{with} 
    \Phi(\rho) = \sum_{i,j} a_{ij} P_i \rho P_j \qq{and} 
    J(\rho) = \sum_{i,j} a_{ij} \acomm{P_j P_i}{\rho},
\end{equation}
so that $\mathcal{D}^2(\rho) = \Phi^2(\rho) - \tfrac12 J(\Phi(\rho)) - \tfrac12 \Phi(J(\rho)) + \tfrac14 J^2(\rho)$. Again, we see that terms like $\tfrac12 J(\Phi(\rho))$ and $\tfrac12 \Phi(J(\rho))$ are irrelevant because they produce contributions where Pauli operators act only on one side of $\rho$. For instance,
\begin{equation}
    J(\Phi(\rho)) \;=\; \sum_{i,j,m,n} a_{ij}a_{mn}\,\big(P_m P_n P_i \rho P_j + P_i \rho P_j P_m P_n\big),
\end{equation}
and any diagonal contribution would necessarily come with a coefficient such as $a_{ik}$ or $a_{kj}$. But $a_{ik}=a_{kj}=0$ for all $i,j$, so these terms vanish. Therefore, $J(\Phi(\rho))$ never yields a $P_k \rho P_k$ term and can be ignored when computing $\chi_{kk}^{(2)}$.

Now let's consider the $J^2$ term:
\begin{equation}
    J^2(\rho) \;=\; \Big\{\sum_{i,j} a_{ij} P_j P_i,\ \Big\{\sum_{m,n} a_{mn} P_{n} P_{m},\ \rho\Big\}\Big\}.
\end{equation}
The relevant contributions to $\chi_{kk}^{(2)}$ are those where $\rho$ is sandwiched between $P_j P_i$ and $P_{n} P_{m}$, i.e.
\begin{equation}
    \sum_{i,j,m,n} a_{ij}a_{mn}\, \big(P_j P_i \,\rho\, P_{n} P_{m} + P_{n} P_{m} \,\rho\, P_j P_i\big).
\end{equation}
Now, recall that products of Paulis close up to a phase:
\begin{equation}
    P_j P_i = \omega_{ji} P_{j\oplus i}, \qquad \omega_{ji}\in\{\pm 1,\pm i\}, \qquad \omega_{ji}=\omega_{ij}^*,
\end{equation}
where $j\oplus i$ denotes the Pauli label of the product modulo phases (see \cref{sec:appendix_notation_prelim} for details on these phase factors). 
For a diagonal term $P_k \rho P_k$, we need simultaneously
\begin{equation}
    P_j P_i = \omega_{ji} P_k, \qquad P_{n} P_{m} = \omega_{nm} P_k,
\end{equation}
which enforces $j=i\oplus k$ and $n = m\oplus k$. Hence the relevant coefficient is
\begin{equation}
    2\cdot\sum_{i,m} a_{i,\,i\oplus k}\,a_{m,\,m\oplus k}\, \omega_{i\oplus k,i}\,\omega_{m\oplus k,m} \;\; P_k \rho P_k = 2\cdot\Big(\sum_i a_{i,\,i\oplus k}\,\omega_{i\oplus k,i}\Big)^2\,P_k \rho P_k
    \label{eq:J_squared_term}
\end{equation}
The quantity $S_k=\sum_i a_{i,i\oplus k}\,\omega_{i\oplus k,i}$ is real because $S_k^*=\sum_i a_{i,i\oplus k}^*\,\omega_{i\oplus k,i}^* = \sum_i a_{i\oplus k,i}\,\omega_{i,i\oplus k}$. We can permute the sum index $i \oplus k \mapsto i'$, which leaves the sum invariant since it runs over all $i$: $S_k^* = \sum_{i'} a_{i',i'\oplus k}\,\omega_{i'\oplus k,i'} = S_k$. Consequently, \cref{eq:J_squared_term} is nonnegative, since it is the square of a real quantity. As a result, the $J^2$ contribution to $\chi_{kk}^{(2)}$ is always nonnegative.

Finally, consider the $\Phi^2(\rho)$ term:
\begin{equation}
    \Phi^2(\rho) \;=\; \sum_{i,j,m,n} a_{ij}a_{mn} P_i P_m \rho P_n P_j
\end{equation}
For a diagonal term $P_k \rho P_k$, we need simultaneously
\begin{equation}
    P_i P_m = \omega_{im} P_k, \qquad P_{n} P_{j} = \omega_{nj} P_k,
\end{equation}
which enforces $m=i\oplus k$ and $n = j\oplus k$. Hence, the relevant coefficient is
\begin{equation}
    \sum_{i,j} a_{i,j}\,a_{i\oplus k,\,j\oplus k}\, \omega_{i,i\oplus k}\,\omega_{j\oplus k,j} \;\; P_k \rho P_k
    \label{eq:Phi_squared_term}
\end{equation}
Define the matrix $b=[b_{ij}]$ by
\begin{equation}
    b_{ij} \;:=\; a_{i,j}\,a_{i\oplus k,\,j\oplus k} \;=\; \big(a\odot (\pi_k^\top a \pi_k)\big)_{ij},
\end{equation}
where $\pi$ is the permutation matrix that maps $i \mapsto \pi_k(i) = i\oplus k$ and $\odot$ is a Hadamard product of matrices. Then, the permuted matrix $\pi_k^\top a \pi_k$ is positive semidefinite, and by the Schur product theorem, which states that the Hadamard product of two positive definite matrices is also a positive definite matrix, we have that $b\succeq 0$. Now define the phase vector $\omega^{(k)}\in\mathbb{C}^{4^n}$ with components $\omega^{(k)}_j:=\omega_{j\oplus k,j}$; note $\omega^{(k)*}_j=\omega_{j,j \oplus k}$. Hence,
\begin{equation}
    \sum_{i,j} a_{ij}\,a_{i\oplus k,\,j\oplus k}\,\omega_{i,i\oplus k}\,\omega_{j\oplus k,j}
    \;=\; \sum_{i,j} \omega^{(k)*}_i\,b_{ij}\, \omega^{(k)}_j\;=\; \omega^{(k)\dagger}b\,\omega^{(k)} \ge 0,
\end{equation}

because $b\succeq 0$. Therefore, the $\Phi^2$ contribution to the $P_k\rho P_k$ coefficient is nonnegative. 

As a result, we have that $\chi_{kk}^{(2)} \ge 2h_k^2$ whenever $a_{kk}=0$, thus proving \cref{lemma:time_derivatives_pauli_rates}.
\end{proof}

\begin{lemma}[Support of the short-time Pauli-rate derivatives]
\label{lemma:pauli_rate_derivative_support}
Let $P_0$ denote the identity Pauli, and define the set of nonidentity
pairwise products of dissipator Paulis, modulo phases, by
\begin{equation*}
    \mathcal S_D^{(2)}
    :=
    \left\{
        P_k\neq P_0 :
        P_k\propto P_iP_j,\quad
        P_i,P_j\in\mathcal S_D
    \right\}.
\end{equation*}
For every nonidentity
$P_k\notin\mathcal S_D\cup\mathcal S_D^{(2)}$, the Pauli error rates
of $e^{t\mathcal L}$ satisfy
\begin{equation*}
    \chi_{kk}^{(1)}=0,
    \qquad
    \chi_{kk}^{(2)}=2h_k^2.
\end{equation*}
In particular, both derivatives vanish for every nonidentity
$P_k\notin
\mathcal S_H\cup\mathcal S_D\cup\mathcal S_D^{(2)}$.
Moreover,
\begin{equation*}
    |\mathcal S_D^{(2)}|
    \le \binom{|\mathcal S_D|}{2}.
\end{equation*}
\end{lemma}

\begin{proof}
Fix a nonidentity
$P_k\notin\mathcal S_D\cup\mathcal S_D^{(2)}$.
Since $P_k\notin\mathcal S_D$, we have $a_{kk}=0$, and positive
semidefiniteness of $a$ implies $a_{ki}=a_{ik}=0$ for every $i$.
The first-derivative identity therefore gives $\chi_{kk}^{(1)}=0$.

For the second derivative, use the decomposition in the proof of
\cref{lemma:time_derivatives_pauli_rates}.
When $a_{kk}=0$, the mixed Hamiltonian--dissipator contribution
vanishes, as do the contributions from $J(\Phi(\rho))$ and
$\Phi(J(\rho))$.
The remaining contributions to $\chi_{kk}^{(2)}$ are $2h_k^2$
from the Hamiltonian part and those from
$\Phi^2(\rho)+\tfrac14 J^2(\rho)$.

Since $a_{ij}=0$ unless $P_i,P_j\in\mathcal S_D$, a diagonal
contribution from either of these dissipative terms requires
$P_k$ to be proportional to a product of two elements of
$\mathcal S_D$, as seen in
\cref{eq:Phi_squared_term,eq:J_squared_term}.
This is excluded by $P_k\notin\mathcal S_D^{(2)}$.
Hence $\chi_{kk}^{(2)}=2h_k^2$, which also vanishes when
$P_k\notin\mathcal S_H$.

Finally, $P_i^2=P_0$ and $P_iP_j\propto P_jP_i$, so every element
of $\mathcal S_D^{(2)}$ arises from an unordered pair of distinct
elements of $\mathcal S_D$.
Thus $|\mathcal S_D^{(2)}|\le\binom{|\mathcal S_D|}{2}$.
\end{proof}

\subsection{Estimating Pauli error rates}
\label{sec:pauli_error_rates}

\cref{lemma:time_derivatives_pauli_rates} establishes that the operator content of the Lindbladian is encoded in the time derivatives of $\{\chi_{ii}(t)\}$. In this section, we therefore focus on estimating the Pauli error rates $\{\chi_{ii}\}$ as a precursor to estimating the derivatives $\{\chi_{ii}^{(1)}(t=0), \chi_{ii}^{(2)}(t=0)\}$. 

The Pauli error rates of any quantum channel form a valid probability distribution: $\chi_{ii}\ge 0$ and $\sum_i \chi_{ii} = 1$. Our goal is to learn the distribution ${\chi_{ii}}$ to accuracy $\varepsilon$ in $\ell_\infty$ distance. In the entanglement-assisted setting, this task is straightforward: by measuring the Choi state of $\mathcal{E}$ in the Bell basis, one obtains i.i.d.\ samples from ${\chi_{ii}}$ directly, and standard distribution-learning results imply a sample complexity of $\Theta(\varepsilon^{-2})$, independent of system size.

By contrast, the entanglement-free setting is substantially more challenging. Without ancillas, one cannot directly sample from ${\chi_{ii}}$. (Indeed, if system-only operations allowed sampling an index $i$ with probability $\chi_{ii}$, then Pauli eigenvalues could be estimated with $\mathcal{O}(1/\varepsilon^2)$ measurements independent of $n$, contradicting known entanglement-free lower bounds of $\Omega(2^n/\varepsilon^2)$~\citep{chen2024tight,chen2022quantum}.)
Nevertheless, \citet{flammia2021pauli} showed that an ancilla-free approach is still possible by reducing Pauli error learning to the classical problem of population recovery.

\begin{theorem}[Learning Pauli error rates via population recovery~\cite{flammia2021pauli}]
\label{theorem:population_recovery}
Let $\mathcal{E}$ be an $n$-qubit channel with Pauli error rates $\{\chi_{ii}\}$ in the Pauli basis.
For any $0<\varepsilon,\delta<1$, there exists an ancilla-free procedure that uses
\begin{equation}
\label{eq:poprec_samples}
m \;=\; \mathcal{O}\!\left(\varepsilon^{-2}\log\frac{n}{\varepsilon\delta}\right)
\end{equation}
applications of $\mathcal{E}$ and prepares/measures only product single-qubit Pauli eigenstates.
With probability at least $1-\delta$, it outputs a hypothesis $\widehat\chi$ such that
\begin{equation}
\label{eq:poprec_accuracy}
\max_i |\widehat\chi_{ii}-\chi_{ii}|
\;\le\; \varepsilon.
\end{equation}

The hypothesis $\widehat\chi$ is supported on a set $\widehat S\subseteq\{0,\dots,4^n-1\}$
of size $|\widehat S|\;\le\;\frac{4}{\varepsilon}$ (i.e., $\widehat\chi_{ii}=0$ for $i\notin \widehat S$)

The classical post-processing runs in time $\mathcal{O}\!\left(mn/\varepsilon\right) = \widetilde{O}(n/\varepsilon^3)$, where $\widetilde{O}(\cdot)$ omits the $\operatorname{polylog}()\,$ factors in $(n,1/\varepsilon,1/\delta)$. 
\end{theorem}

\begin{remark}[Operational form]\label{rem:poprec_operational}
One may view the procedure in \cref{theorem:population_recovery} as first applying Pauli twirling so that the effective channel is a Pauli channel
(i.e. only the diagonal rates $\{\chi_{ii}\}$ of the channel survive), then collecting $mn$ single-qubit Pauli-basis measurement outcomes from product inputs and performing a population-recovery
post-processing step; see~\cite{flammia2021pauli} for details.
\end{remark}

\subsection{Estimating time-derivatives of Pauli error rates}
\label{sec:pauli_error_derivative_estimation}

Having established how to estimate the Pauli error rates $\{\chi_{ii}(t)\}$, we now turn to estimating their first and second time-derivatives at $t=0$, which are required to apply \cref{lemma:time_derivatives_pauli_rates} to structure learning.  
Our approach is based on Chebyshev interpolation; see \cref{sec:chebyshev_interpolation}. Specifically, we evaluate $\widehat{\chi}_{ii}(t)$ at a sequence of nodes $\{t_m\}_{m=0}^r \subset [0,\tau_{\max}]$, fit a degree-$r$ Chebyshev interpolant, and use its closed-form derivatives at $t=0$ as estimators for $\chi_{ii}^{(1)}(0)$ and $\chi_{ii}^{(2)}(0)$. In \cref{sec:chebyshev_interpolation} we show that these estimators achieve provably small bias and variance when the underlying function is sufficiently smooth (see \cref{sec:deriv-bounds-pauli-error-rates}).

The following lemmas formalize the sample complexities associated with estimating the derivatives of the Pauli error rates.

\begin{lemma}[First derivative estimation of Pauli error rates]
    \label{lemma:first_deriv_estimation_pauli_error_rates}
    Let $\mathcal{E}_t = e^{\mathcal{L}t}$ be an $n$-qubit channel generated by a Lindbladian $\mathcal{L}$ with at most $M$ nonzero Pauli coefficients. 
    Choose $\tau_{\max} = \Theta(1/M)$ and $r = \Theta(\log(M/\epsilon_1))$, and obtain estimates $\{\widehat{\chi}_{ii}(t_m)\}$ of the Pauli error rates of $\mathcal{E}_t$ at the $r{+}1$ Gauss–Chebyshev nodes $\{t_m\}_{m=0}^r \subset [0,\tau_{\max}]$ using a total of
    \begin{equation}
        m = \mathcal{O}\left(
        \frac{M^2}{\epsilon_1^2}\;
        \mathrm{polylog}\Big(n, \tfrac{M}{\epsilon_1}, \tfrac{1}{\delta}\Big)
        \right)
    \end{equation}
    applications of the channel $e^{\mathcal{L}t}$. 
    Then, from these estimates, one can construct a Chebyshev-interpolated estimator $\{\widehat{\chi}_{ii}^{(1)}\}$ for the first time-derivatives of all Pauli error rates such that, with probability at least $1-\delta$,
    \begin{equation}
        \big|\widehat{\chi}_{ii}^{(1)} - \chi_{ii}^{(1)}\big| \le \epsilon_1,\qquad \forall\, i.
    \end{equation}
\end{lemma}

\begin{proof}
    
By the Pauli-rate derivative bound \cref{cor:chi_derivative_bound_M_sparse}, the function
$f(t)=\chi_{ii}(t)$ satisfies
\[
\|f^{(k)}\|_\infty \le (2M)^k \qquad (k\ge 1),
\]
so \cref{cor:first_chebyshev_derivative_parameters} applies with $B=1$ and $\Lambda=2M$.
With the resulting choice of parameters,
\begin{equation}
\label{eq:parameters_first_derivative}
\tau_{\max}=\frac{1}{4M},
\qquad
r=\Big\lceil \log\Big(\frac{36M}{\epsilon_1}\Big)\Big\rceil,
\qquad
\varepsilon_s=\frac{\epsilon_1}{20Mr^{3}},
\end{equation}
the Chebyshev estimator $\widehat{\chi}^{(1)}_{ii}$ satisfies
$|\widehat{\chi}^{(1)}_{ii}-\chi^{(1)}_{ii}|\le \epsilon_1$ provided
$|\widehat\chi_{ii}(t_m)-\chi_{ii}(t_m)|\le \varepsilon_s$ uniformly over all nodes $t_m$.

Such estimates are guaranteed by the population-recovery procedure of \cref{theorem:population_recovery}, except with a failure probability of at most $\delta_s$ per node. Setting the per-node failure probability to $\delta_s = \delta/(r+1)$ and applying the union bound ensures the total failure probability across all nodes is at most $\delta$. Substituting the parameters from \eqref{eq:parameters_first_derivative}, we get that each time node requires
\begin{equation}
    \mathcal{O}\left(\varepsilon_s^{-2}\,\log\frac{n}{\varepsilon_s \delta_s}\right)
\end{equation}
applications of $e^{\mathcal{L}t}$, hence the total number of channel applications across all $r{+}1$ nodes is:
\begin{equation}
    m = \mathcal{O}\left(
    \frac{M^2}{\epsilon_1^2}\;
    \mathrm{polylog}\Big(n,\tfrac{M}{\epsilon_1}, \tfrac{1}{\delta}\Big)
    \right)
\end{equation}
concluding the proof of the lemma. 
\end{proof}

\begin{lemma}[Second derivative estimation of Pauli error rates]
    \label{lemma:second_deriv_estimation_pauli_error_rates}
    Let $\mathcal{E}_t = e^{\mathcal{L}t}$ be an $n$-qubit channel generated by a Lindbladian $\mathcal{L}$ with at most $M$ nonzero Pauli coefficients. 
    Choose $\tau_{\max} = \Theta(1/M)$ and $r = \Theta(\log(M^2/\epsilon_2))$, and obtain estimates $\{\widehat{\chi}_{ii}(t_m)\}$ of the Pauli error rates of $\mathcal{E}_t$ at the $r{+}1$ Gauss–Chebyshev nodes $\{t_m\}_{m=0}^r \subset [0,\tau_{\max}]$ using a total of
    \begin{equation}
        m = \mathcal{O}\left(
        \frac{M^4}{\epsilon_2^2}\;
        \mathrm{polylog}\Big(n, \tfrac{M^2}{\epsilon_2}, \tfrac{1}{\delta}\Big)
        \right)
    \end{equation}
    applications of the channel $e^{\mathcal{L}t}$. 
    Then, from these estimates, one can construct a Chebyshev-interpolated estimator $\{\widehat{\chi}_{ii}^{(2)}\}$ for the second time-derivative of all Pauli error rates such that, with probability at least $1-\delta$:
    \begin{equation}
        \big|\widehat{\chi}_{ii}^{(2)} - \chi_{ii}^{(2)}\big| \le \epsilon_2,
        \qquad \forall\, i.
    \end{equation}
\end{lemma}

\begin{proof}
    
By the Pauli-rate derivative bound \cref{cor:chi_derivative_bound_M_sparse}, the function
$f(t)=\chi_{ii}(t)$ satisfies
\[
\|f^{(k)}\|_\infty \le (2M)^k \qquad (k\ge 1),
\]
so \cref{cor:second_chebyshev_derivative_parameters} applies with $B=1$ and $\Lambda=2M$.
With the resulting choice of parameters,
\begin{equation}
\label{eq:parameters_second_derivative}
\tau_{\max}=\frac{1}{4M},
\qquad
r=\Big\lceil \log\Big(\frac{1280M^2}{\epsilon_2}\Big)\Big\rceil,
\qquad
\varepsilon_s=\frac{\epsilon_2}{32M^2r^{5}},
\end{equation}
the Chebyshev estimator $\widehat{\chi}^{(2)}_{ii}$ satisfies
$|\widehat{\chi}^{(2)}_{ii}-\chi^{(2)}_{ii}|\le \epsilon_2$ provided
$|\widehat\chi_{ii}(t_m)-\chi_{ii}(t_m)|\le \varepsilon_s$ uniformly over all nodes $t_m$.

Following the proof of \cref{lemma:first_deriv_estimation_pauli_error_rates}, set the per-node failure probability to $\delta_s = \delta/(r+1)$. Substituting the parameters from \eqref{eq:parameters_second_derivative} gives the total number of channel evaluations across all nodes:
\begin{equation}
    m = \mathcal{O}\left(
    \frac{M^4}{\epsilon_2^2}\;
    \mathrm{polylog}\Big(n,\tfrac{M^2}{\epsilon_2},\tfrac{1}{\delta}\Big)
    \right).
\end{equation}
completing the proof of \cref{lemma:second_deriv_estimation_pauli_error_rates}.
\end{proof}

\begin{remark}[First-derivative accuracy under second-derivative parameters]
\label{rem:pauli-rates-first-deriv-estimation-from-second-deriv-parameters}
Assume $M\ge 1$ and use the Chebyshev parameters $(\tau_{\max},r,\varepsilon_s)$ from
\cref{eq:parameters_second_derivative}. Then these parameters also meet the sufficient conditions of
\cref{cor:first_chebyshev_derivative_parameters} with target accuracy $\varepsilon=\epsilon_2/2$, and hence
\begin{equation*}
\bigl|\widehat{\chi}^{(1)}_{ii}-\chi^{(1)}_{ii}\bigr|\le \frac{\epsilon_2}{2}
\qquad \forall\, i.
\end{equation*}
\end{remark}

\begin{proof}
With $\tau_{\max}=1/(4M)=1/(2\Lambda)$ fixed, it suffices to check the $r$- and $\varepsilon_s$-conditions for
$\varepsilon=\epsilon_2/2$:
\[
r \ge \log\!\Big(\frac{72M}{\epsilon_2}\Big),\qquad
\varepsilon_s \le \frac{\epsilon_2}{40Mr^3}.
\]
The second-derivative choice has $r=\lceil \log(\tfrac{1280M^2}{\epsilon_2})\rceil\ge \log(\tfrac{72M}{\epsilon_2})$ for $M\ge 1$,
and $\varepsilon_s=\tfrac{\epsilon_2}{32M^2r^5}\le \tfrac{\epsilon_2}{40Mr^3}$ since $M\ge 1$ and $r\ge 2$.
Therefore \cref{cor:first_chebyshev_derivative_parameters} guarantees the claim.
\end{proof}

\begin{remark}[Replacing the crude sparsity scale $2M$ by a tighter derivative bound]
\label{rem:chi_derivative_estimation_tighter_bound}
In \cref{lemma:first_deriv_estimation_pauli_error_rates,lemma:second_deriv_estimation_pauli_error_rates} we set the smoothness scale to $\Lambda=2M$ via the crude bound $(2M)^k$ on $\|\chi_{ij}^{(k)}\|_\infty$.
More generally, one may take any $\Lambda$ satisfying $\Lambda \ge \|\mathcal{L}\|_{2\to2}$; in particular,
by \cref{lemma:operator_norm_bounds} one can choose
\[
\Lambda := \Delta(H) + 2\sum_{m,n}|a_{mn}|,
\]
which yields the tighter uniform derivative bound $|\chi_{ij}^{(k)}(t)| \le \Lambda^k$ and improves the resulting Chebyshev regression parameters.
\end{remark}

\subsection{Structure learning algorithm}\label{sec:structure_learning_algorithm}

Combining the results of the previous subsections, we now present two separate ancilla-free procedures that reconstruct the dissipator and  Hamiltonian structures of a sparse Lindbladian from short-time dynamics. The correctness and resources required by \cref{algorithm:dissipator_structure_learning,algorithm:hamiltonian_structure_learning} are proved in \cref{theorem:ancilla_free_structure_learning}. 

\begin{boxedalgorithm}[H]{Dissipator Structure Learning \label{algorithm:dissipator_structure_learning}}
    \DontPrintSemicolon
    \SetKwInOut{Input}{Inputs}
    \SetKwInOut{Output}{Output}
    \SetKwBlock{Schedule}{Schedule}{}
    
    \Input{\,Sparsity $M$; threshold $\eta$; failure probability $\delta$}
    \Output{\,A structure estimator $\widehat{\mathcal{S}}_D$ such that $\widehat{\mathcal{S}}_D = \mathcal{S}_D$ with probability $\ge 1-\delta$}
    
    \Schedule{
      $\tau_{\max} \gets \Theta(1/M)$ \tcp*{max evolution time}  
      $r \gets \Theta(\log(M/\eta))$ \tcp*{Chebyshev degree}  
      $\{t_m\}_{m=0}^{r} \gets$ Chebyshev–Gauss nodes on $[0,\tau_{\max}]$ \tcp*{sampling schedule}  
      $\varepsilon_s \gets \Theta(\eta/(Mr^{3}))$ \tcp*{accuracy per-node s.t. 1st deriv. accuracy is $\eta/2$.}  
      $\delta_s \gets \Theta(\delta/r)$ \tcp*{per-node failure prob.}  
      $\{\alpha^{(1)}_m\} \gets$ precomputed weights \tcp*{Chebyshev estimator}  
    }
    
    $\widehat{\mathcal{S}}, \widehat{\mathcal{S}}_D \gets \emptyset$ \;
    
    \For{$m \gets 0$ \KwTo $r$}{
        $\{\widehat\chi_{ii}(t_m)\} \gets \texttt{PopulationRecovery}\left(\mathcal{E}_{t_m},\, \varepsilon_s,\, \delta_s\right)$\;
        $S \gets S \cup \{\widehat\chi_{ii}(t_m)\}$ \tcp*{update candidate set}
    }
    
    \For{$i: \widehat\chi_{ii}(\cdot)\in S$}{
        $\widehat\chi_{ii}^{(1)} \gets \sum_{m=0}^{r} \alpha^{(1)}_m\,\widehat\chi_{ii}(t_m)$ \tcp*{1st derivative estimator} 
        \If{$\widehat\chi_{ii}^{(1)} > \tfrac{\eta}{2}$}{
            $\widehat{\mathcal{S}}_D \gets \widehat{\mathcal{S}}_D \cup \{P_i\}$ \tcp*{dissipator candidate after structure pruning}
        }
    }
    
    \Return $\widehat{\mathcal{S}}_D$\;
\end{boxedalgorithm}

\begin{boxedalgorithm}[H]{Hamiltonian Structure Learning \label{algorithm:hamiltonian_structure_learning}}
    \DontPrintSemicolon
    \SetKwInOut{Input}{Inputs}
    \SetKwInOut{Output}{Output}
    \SetKwBlock{Schedule}{Schedule}{}
    
    \Input{\,Sparsity $M$; threshold $\eta$; failure probability $\delta$}
    \Output{\,A superset $\widehat{\mathcal{S}}_H$ such that $\mathcal{S}_H \subseteq \widehat{\mathcal{S}}_H$ with probability $\ge 1-\delta$}
    
    \Schedule{
      $\tau_{\max} \gets \Theta(1/M)$ \tcp*{max evolution time}  
      $r \gets \Theta(\log(M^2/\eta^{2}))$ \tcp*{Chebyshev degree}  
      $\{t_m\}_{m=0}^{r} \gets$ Chebyshev–Gauss nodes on $[0,\tau_{\max}]$ \tcp*{sampling schedule}  
      $\varepsilon_s \gets \Theta(\eta^{2}/(M^{2}r^{5}))$ \tcp*{accuracy per-node s.t. 2nd deriv. accuracy is $\eta^2$.} 
      $\delta_s \gets \Theta(\delta/r)$ \tcp*{per-node failure prob.}  
      $\{\alpha^{(1)}_m\}, \{\alpha^{(2)}_m\} \gets$ precomputed weights \tcp*{Chebyshev estimator}  
    }
    
    $\widehat{\mathcal{S}},\widehat{\mathcal{S}}_H \gets \emptyset$ \;
    
    \For{$m \gets 0$ \KwTo $r$}{
        $\{\widehat\chi_{ii}(t_m)\} \gets \texttt{PopulationRecovery}\left(\mathcal{E}_{t_m},\, \varepsilon_s,\, \delta_s\right)$\;
        $S \gets S \cup \{\widehat\chi_{ii}(t_m)\}$ \tcp*{update candidate set}
    }
    
    \For{$i: \widehat\chi_{ii}(\cdot)\in S$}{
        $\widehat\chi_{ii}^{(1)} \gets \sum_{m=0}^{r} \alpha^{(1)}_m\,\widehat\chi_{ii}(t_m)$ \tcp*{1st derivative estimator} 
        $\widehat\chi_{ii}^{(2)} \gets \sum_{m=0}^{r} \alpha^{(2)}_m\,\widehat\chi_{ii}(t_m)$ \tcp*{2nd derivative estimator}  
        \If{$\widehat\chi_{ii}^{(2)} > \eta^2$  \text{   or   }  $\widehat\chi_{ii}^{(1)} > \tfrac{\eta}{2}$}{
            $\widehat{\mathcal{S}}_H \gets \widehat{\mathcal{S}}_H \cup \{P_i\}$ \tcp*{Hamiltonian candidate after structure pruning}
        }
    }
    
    \Return $\widehat{\mathcal{S}}_H$\;
\end{boxedalgorithm}

\begin{theorem}[Ancilla-free structure learning]
\label{theorem:ancilla_free_structure_learning}
Let $\mathcal{E}_t=e^{t\mathcal{L}}$ be the channel generated by an $n$-qubit Lindbladian $\mathcal{L}$ whose Pauli expansion
is $M$-sparse with the smallest coefficient magnitude $\ge \eta$ (cf.~\cref{eq:lindbladian_struct_learning_appendix}), and with Hamiltonian and dissipator structures $\mathcal{S}_H$ and $\mathcal{S}_D$
as in \cref{def:hamiltonian_structure,def:dissipator_structure}. Then there is an ancilla-free procedure
(\cref{algorithm:dissipator_structure_learning,algorithm:hamiltonian_structure_learning}) which, given parameters
$0<\eta,\delta<1$, outputs candidate sets $\widehat{\mathcal{S}}_D$ and $\widehat{\mathcal{S}}_H$ such that the
following guarantees hold:

\begin{enumerate}
\item \textit{(Coverage and size.)} With probability $\ge1-\delta$, the output sets satisfy:
$\widehat{\mathcal{S}}_D =\mathcal{S}_D$ and $\mathcal{S}_H \subseteq \widehat{\mathcal{S}}_H,$ and 
\begin{equation*}
    |\widehat{\mathcal S}_D|=|\mathcal S_D|,
    \qquad
    |\widehat{\mathcal S}_H|\le M^2.
\end{equation*}

\item \textit{(Number of experiments.)} The total number of channel uses is
\[
N_{\exp}
=
\mathcal{O}\!\left(
\frac{M^4}{\eta^4}\;
\mathrm{polylog}\Big(n,\tfrac{M}{\eta},\tfrac{1}{\delta}\Big)
\right).
\]
and is dominated by identifying the Hamiltonian set $\widehat{\mathcal{S}}_H$. Identifying the $\widehat{\mathcal{S}}_D$ uses nearly quadratically fewer queries. 

\item \textit{(Time resolution.)} It only ever applies $e^{\Lindblad t}$ for $t \ge t_{\mathrm{res}}$, where
\[
t_{\mathrm{res}}=\Omega\Bigl(M^{-1}\,\mathrm{polylog}(M/\eta)^{-1}\Bigr) .
\]

\item \textit{(Classical overhead.)} The total classical post-processing time is
\[
\widetilde{\mathcal{O}}\!\left(\frac{n\,M^6}{\eta^6}\right),
\]
where the $\widetilde{\mathcal{O}}(\cdot)$ hides polylogarithmic factors in $(M,n,1/\eta,1/\delta)$.

\item \textit{(Total evolution time.)} The procedure only queries $\mathcal{E}_t$ at times $t\in(0,\tau_{\max})$
with $\tau_{\max}=\Theta(1/M)$, and therefore the total evolution time satisfies
\[
t_{\mathrm{tot}}:=\sum_{\text{calls}} t \;\le\; N_{\exp}\,\tau_{\max}
=
\mathcal{O}\!\left(
\frac{M^3}{\eta^4}\;
\mathrm{polylog}\Big(n,\tfrac{M}{\eta},\tfrac{1}{\delta}\Big)
\right).
\]
and is also dominated by identifying the Hamiltonian set $\widehat{\mathcal{S}}_H$. Identifying the $\widehat{\mathcal{S}}_D$ takes total time $\widetilde{O}(M/\eta^2)$.
\end{enumerate}

\noindent
The procedure prepares only \emph{product} $n$-qubit input states in which each qubit is a Pauli eigenstate, and measures each output qubit
in a single-qubit Pauli eigenbasis.
\end{theorem}

\begin{proof}
Both \cref{algorithm:dissipator_structure_learning,algorithm:hamiltonian_structure_learning} follow the same template.
First, they invoke the ancilla-free population-recovery procedure
(\cref{theorem:population_recovery}) to estimate the Pauli error rates $\{\chi_{ii}(t)\}$
of $\mathcal{E}_t=e^{t\Lindblad}$ at a set of Gauss--Chebyshev nodes $\{t_m\}_{m=0}^r\subset(0,\tau_{\max}]$.
Second, they apply Chebyshev interpolation to these nodewise estimates to obtain estimators for the short-time derivatives
$\{\chi_{ii}^{(1)}(0)\}$ and $\{\chi_{ii}^{(2)}(0)\}$, with accuracy guarantees provided by
\cref{lemma:first_deriv_estimation_pauli_error_rates,lemma:second_deriv_estimation_pauli_error_rates}.
Finally, they use the structural identities linking Lindbladian coefficients to Pauli-rate derivatives
(\cref{lemma:time_derivatives_pauli_rates}) to prune the initial candidate pools and output the supersets
$\widehat{\mathcal{S}}_D$ and $\widehat{\mathcal{S}}_H$.

\textbf{Correctness.}
The Pauli error rates derivative guarantees are conditioned on the event that the population-recovery subroutine succeeds at every Gauss--Chebyshev node of the algoritms. In \cref{algorithm:hamiltonian_structure_learning} the schedule parameters $(\tau_{\max},r,\varepsilon_s,\delta_s)$ are chosen as in
\cref{lemma:second_deriv_estimation_pauli_error_rates}, so with probability at least $1-\delta/2$ the Chebyshev estimator satisfies
\begin{equation}
\label{eq:good_event_second}
\bigl|\widehat{\chi}^{(2)}_{ii}-\chi^{(2)}_{ii}\bigr|\le \eta^{2},
\qquad \forall\, i.
\end{equation}

\smallskip
In \cref{algorithm:dissipator_structure_learning} the schedule parameters are chosen as in
\cref{lemma:first_deriv_estimation_pauli_error_rates}, so with probability at least $1-\delta/2$,
\begin{equation}
\label{eq:good_event_first}
\bigl|\widehat{\chi}^{(1)}_{ii}-\chi^{(1)}_{ii}\bigr|\le \frac{\eta}{2},
\qquad \forall\, i.
\end{equation}
The same bound \eqref{eq:good_event_first} also holds for \cref{algorithm:hamiltonian_structure_learning}
(by \cref{rem:pauli-rates-first-deriv-estimation-from-second-deriv-parameters} under its parameter choice),
again with probability $\ge 1-\delta/2$. 

By a union bound, \eqref{eq:good_event_second} and \eqref{eq:good_event_first} hold for both \cref{algorithm:dissipator_structure_learning,algorithm:hamiltonian_structure_learning}  with probability at least $1-\delta$.
Conditioning on this event, we now prove coverage for $\widehat{\mathcal{S}}_D$ and $\widehat{\mathcal{S}}_H$.

\smallskip
\noindent\emph{Dissipator coverage.}
Fix $P_k$. If $P_k\in\mathcal{S}_D$, then by \cref{lemma:time_derivatives_pauli_rates} we have
$\chi^{(1)}_{kk}=a_{kk}\ge \eta$.
Combined with \eqref{eq:good_event_first}, this implies $ \widehat{\chi}^{(1)}_{kk}\ge \chi^{(1)}_{kk}-\eta/2 \ge \eta/2$, 
so the selection rule in \cref{algorithm:dissipator_structure_learning} includes $P_k$. Conversely, if $P_k\notin\mathcal{S}_D$, then $\chi^{(1)}_{kk}=a_{kk}=0$, and \eqref{eq:good_event_first} gives
$ \widehat{\chi}^{(1)}_{kk}\le \eta/2,$
so the selection rule excludes $P_k$.
Therefore, $\widehat{\mathcal{S}}_D=\mathcal{S}_D$.

\smallskip
\noindent\emph{Hamiltonian coverage.}
Fix $P_k\in\mathcal{S}_H$.
If $P_k\in\mathcal{S}_H\setminus \mathcal{S}_D$, then $a_{kk}=0$ and $|h_k|\ge \eta$.
By \cref{lemma:time_derivatives_pauli_rates},
$\chi^{(2)}_{kk}\ge 2h_k^2\ge 2\eta^2$.
Combining with \eqref{eq:good_event_second} yields
$\widehat{\chi}^{(2)}_{kk}\ge 2\eta^2-\eta^2=\eta^2$,
so the second-derivative selection rule of \cref{algorithm:hamiltonian_structure_learning} includes $P_k$.
If instead $P_k\in\mathcal{S}_H\cap\mathcal{S}_D$, then $a_{kk}\ge \eta$ implies
$\chi^{(1)}_{kk}=a_{kk}\ge \eta$ by \cref{lemma:time_derivatives_pauli_rates}.
Combining with \eqref{eq:good_event_first} gives
$\widehat{\chi}^{(1)}_{kk}\ge \eta-\eta/2=\eta/2$,
so $P_k$ is included by the first-derivative selection rule.
Therefore $\mathcal{S}_H\subseteq \widehat{\mathcal{S}}_H$.

\medskip
\textbf{Sample complexity.}
The total number of channel uses needed to estimate all values $\{\widehat\chi_{ii}(t_m)\}$ at the Gauss--Chebyshev nodes
to the uniform accuracy $\varepsilon_s$ is determined by the derivative-accuracy requirement. For dissipator structure learning, only first-derivative estimates are needed. Taking $\epsilon_1\le \eta/2$ in
\cref{lemma:first_deriv_estimation_pauli_error_rates} (so that \eqref{eq:good_event_first} holds) yields
\begin{equation}
    m_D = \mathcal{O}\left(
    \frac{M^2}{\eta^2}\;
    \mathrm{polylog}\Big(n,\tfrac{M}{\eta},\tfrac{1}{\delta}\Big)
    \right).
\end{equation}

For Hamiltonian structure learning, second-derivative estimates are required (and under the same schedule the
first-derivative accuracy condition also holds). Taking $\epsilon_2\le \eta^2$ in
\cref{lemma:second_deriv_estimation_pauli_error_rates} (so that \eqref{eq:good_event_second} and \eqref{eq:good_event_first} hold) gives
\begin{equation}
    m_H = \mathcal{O}\left(
    \frac{M^4}{\eta^4}\;
    \mathrm{polylog}\Big(n,\tfrac{M}{\eta},\tfrac{1}{\delta}\Big)
    \right).
\end{equation}

\medskip
\textbf{Superset sizes.}
Condition on the same good estimation event used above.
Since $\widehat{\mathcal S}_D=\mathcal S_D$, clearly
$|\widehat{\mathcal S}_D|=|\mathcal S_D|$.

Now fix a nonidentity
$P_k\notin
\mathcal S_H\cup\mathcal S_D\cup\mathcal S_D^{(2)}$.
By \cref{lemma:pauli_rate_derivative_support},
$\chi_{kk}^{(1)}=\chi_{kk}^{(2)}=0$.
The error bounds
\eqref{eq:good_event_first} and \eqref{eq:good_event_second}
therefore imply $\widehat{\chi}_{kk}^{(1)}\le \frac{\eta}{2}, \widehat{\chi}_{kk}^{(2)}\le \eta^2.$
Thus neither selection condition in
\cref{algorithm:hamiltonian_structure_learning} is satisfied,
and $P_k\notin\widehat{\mathcal S}_H$.
Consequently,
\begin{equation*}
    \widehat{\mathcal S}_H\setminus\{P_0\}
    \subseteq
    \mathcal S_H\cup\mathcal S_D\cup\mathcal S_D^{(2)}.
\end{equation*}
Writing $M_H=|\mathcal S_H|$ and $M_D=|\mathcal S_D|$, we obtain
\begin{equation*}
    |\widehat{\mathcal S}_H\setminus\{P_0\}|
    \le M_H+M_D+\binom{M_D}{2}
    = M_H+\frac{M_D(M_D+1)}{2}
    \le M_H+M_D^2
    \le M^2.
\end{equation*}

\medskip
\textbf{Time resolution.}
Per \cref{lemma:time_resolution_chebyshev_nodes}, the smallest time resolution (as dictated by the Gauss-Chebyshev nodes) satisfies:
\begin{equation*}
        t_{\mathrm{res}} \coloneq \min_{0\le m<r}\bigl|t_{m+1}-t_m\bigr|
        \;=\; \Omega\left(\tau_{\max}/r^2\right) \;=\;
        \Omega\left(M^{-1}\,\mathrm{polylog}(M/\eta)^{-1}\right).
\end{equation*}
since $\tau_{\max}=\Theta(1/M)$ and $r=\mathrm{polylog}(M/\eta)$ in both schedules. 

\medskip
\textbf{Classical post-processing.}
The classical post-processing time is dictated by the population recovery algorithm. By \cref{theorem:population_recovery}, the classical cost per node is $\widetilde{O}(n/\varepsilon_s^3)$, hence the total classical cost is upper bounded by $O(nr/\varepsilon_s^3)$ and is dominated by the Hamiltonian structure learning. Plugging $\varepsilon_s$ from \cref{eq:parameters_second_derivative} corresponding to $\epsilon_2 = \eta^2$, the total classical cost is:
\begin{equation*}
    O(nr/\varepsilon_s^3) = \widetilde{O}(nM^6/\eta^6)
\end{equation*}

This completes the proof of \cref{theorem:ancilla_free_structure_learning} (the total evolution time is evident from the theorem statement).
\end{proof}

\begin{remark}[False positives]
Evidently, the Hamiltonian candidate set $\widehat{\mathcal S}_H$ may include dissipator terms that are not
present in the true Hamiltonian support. Moreover, $\widehat{\mathcal S}_H$ can contain additional spurious
Pauli strings arising from higher-order effects in the short-time expansion---most notably, second-order
dissipative couplings that can induce $\widehat{\chi}^{(2)}_{ii}>\eta^2$ even when the corresponding Pauli
does not appear in the Hamiltonian. Such false positives are benign: in the coefficient-learning stage they
are treated as additional candidate targets, and the stable reconstruction
returns estimates close to their true zero coefficients, effectively pruning
them.
\end{remark}

\begin{remark}[Obtaining $\eta$-heavy support] \label{remark:eta-heavy-support}
If $\eta$ is chosen larger than the true minimum nonzero magnitude, the structure-learning routine
recovers the support of the $\eta$-heavy part of the dissipator. Under the same scheduling, it also
identifies (almost all of) the $\eta$-heavy Hamiltonian support, returning a superset
$\widehat{\mathcal S}_H$ that contains every $P_j\in\mathcal S_H$ with $|h_j|\ge \eta$, except possibly in
the narrow regime where the associated diagonal dissipator coefficient satisfies
$0<a_{jj}<\Theta(\eta^2)$ (as implied by \cref{lemma:time_derivatives_pauli_rates}). We leave a finer
analysis of this corner case for future work.

This $\eta$-heavy guarantee is sufficient when one only seeks the dominant terms. In contrast, our
coefficient-learning guarantees (\cref{algorithm:coefficient_learning}, \cref{sec:coefficient_learning})
assume that the returned supports are supersets of the true ones, i.e.,
$\mathcal S_D \subseteq \widehat{\mathcal S}_D$ and $\mathcal S_H \subseteq \widehat{\mathcal S}_H$.
\end{remark}

\begin{remark}[Reduced resource requirements via tighter operator norm bound]
\label{rem:theorem7_replace_M_by_Lambda}
All resource bounds in \cref{theorem:ancilla_free_structure_learning} are stated in terms of the crude smoothness scale
$\Lambda=2M$, coming from the uniform derivative bound $|\chi^{(k)}_{ij}(t)|\le (2M)^k$ used in
\cref{lemma:first_deriv_estimation_pauli_error_rates,lemma:second_deriv_estimation_pauli_error_rates}.
More generally, the same analysis goes through with any $\Lambda$ satisfying $\Lambda\ge \|\mathcal{L}\|_{2\to 2}$.
Consequently, every occurrence of $M$ in the resource guarantees of \cref{theorem:ancilla_free_structure_learning}
(e.g.\ $N_{\exp}$, $t_{\mathrm{res}}$, classical cost, $t_{\mathrm{tot}}$, and the $\tau_{\max}=\Theta(1/M)$)
can be replaced by $\Lambda/2$ .

In particular, by \cref{lemma:operator_norm_bounds} one may take
\[
\Lambda := \Delta(H) + 2\sum_{m,n}|a_{mn}|,
\]
which can be significantly smaller than $2M$ and yields uniformly improved resource bounds.
\end{remark}

\section{Coefficient learning}
\label{sec:coefficient_learning}

We now turn to the second stage of the learning procedure: estimating
the numerical Lindbladian coefficients once the candidate Pauli
structure has been identified. Recall that
\begin{equation}
\label{eq:coefficient_learning_lindbladian}
    \mathcal L^\dagger(O)
    =
    i\sum_{P_p\in\mathcal S_H} h_p[P_p,O]
    +
    \sum_{P_i,P_j\in\mathcal S_D}
    a_{ij}
    \left(
        P_jOP_i-\frac12\{P_jP_i,O\}
    \right),
\end{equation}
where the Kossakowski matrix $a$ is Hermitian.

Throughout this section we condition on the successful
structure-learning event
\begin{equation}
\label{eq:coefficient_learning_structure_event}
    \widehat{\mathcal S}_D=\mathcal S_D,
    \qquad
    \mathcal S_H\subseteq\widehat{\mathcal S}_H
    \subseteq
    \mathcal S_D
    \cup
    \{P_p\neq I:\chi^{(2)}_{pp}(0)>0\}.
\end{equation}
Thus the dissipative support is known exactly, while
$\widehat{\mathcal S}_H$ may contain additional Pauli strings whose
true Hamiltonian coefficient is zero.
All candidate Pauli labels are taken to be nonidentity phaseless
Hermitian Paulis, and we write $P_0:=I$ when convenient.
We assume a known sparsity bound $M$ such that $\mathcal L$
contains at most $M$ nonzero Pauli components, and a known integer
upper bound $\mathfrak d\ge1$ on the maximum degree of the dual
interaction graph; a degree-zero model is covered by taking the
upper bound $\mathfrak d=1$.

Our goal is to reconstruct every candidate coefficient to additive
accuracy $\varepsilon$ using only product-state preparations,
native evolution under $e^{t\mathcal L}$, and Pauli measurements.
The main issue is that different global Lindbladian terms can be
locally indistinguishable on the natural support of a target
coefficient. We resolve these ambiguities by using a support-enlargement procedure adapted from \citet{montana2025efficiently}, 
and then construct explicit stable reconstruction formulas for the
individual coefficients.

The section proceeds as follows:
\begin{itemize}
    \item \textbf{Local identifiability and support enlargement.}
    In \cref{sec:stable_local_reconstruction}, we characterize when
    distinct Lindbladian terms induce the same local sandwich
    coefficient and remove these ambiguities by adding at most
    $\mathfrak d$ witness qubits.

    \item \textbf{Deterministic targeted probe construction.}
    In \cref{sec:deterministic_targeted_probes}, we construct a small
    target-dependent family of Pauli probes whose unwanted
    contributions cancel exactly. The resulting reconstruction has absolute weight sum at most one and uses $\mathcal O(\mathfrak d)$ probes for Hamiltonian and off-diagonal dissipative coefficients and $\mathcal O(\min\{\mathfrak d^2,\widehat M_D\})$ probes for diagonal dissipative coefficients. Altogether, only $\mathcal O(\widehat M\mathfrak d)$ distinct derivatives are required.

    \item \textbf{Derivative estimation.}
    In \cref{sec:pauli_derivative_estimation_main}, we show how to
    estimate the selected short-time derivatives from finite-time
    experimental data using Gauss--Chebyshev interpolation.
    To control the interpolation error, we first show that every
selected observable overlaps at most
$\mathcal O(\min\{\mathfrak d^2,M\})$
genuine Lindbladian components, which gives the required uniform
bounds on its higher time derivatives.

    \item \textbf{Optional parallel estimation.}
    In \cref{sec:parallel_derivative_estimation}, we show that when
    the selected probes are sufficiently local, their expectation
    values can instead be estimated simultaneously using shadow
    process tomography.

    \item \textbf{Complete algorithm and guarantees.}
    In \cref{sec:coefficient_learning_algorithm}, we combine these
    ingredients into an ancilla-free coefficient-learning algorithm.
    Estimating the selected probes separately requires
    \[
        \widetilde{\mathcal O}\left(
            \widehat M\,\mathfrak d\,
            \frac{\Gamma^2}{\varepsilon^2}
        \right),
        \qquad
        \widehat M
        :=
        |\widehat{\mathcal S}_H|
        +
        |\widehat{\mathcal S}_D|^2,
        \qquad
        \Gamma = \mathcal O(\min\{\mathfrak d^2,M\})
    \]
    channel applications and does not require bounded locality. 
\end{itemize}

\subsection{Local identifiability and support enlargement}
\label{sec:stable_local_reconstruction}

We first determine a region on which each Lindbladian coefficient
can be distinguished from the other candidate terms. Restricting
probes to the natural support of a term can make different
coefficients locally indistinguishable. We resolve these
ambiguities by adding witness qubits, adapting the support
enlargement of Lemma~3 in
Monta\~n\`a-L\'opez \emph{et al.}~\cite{montana2025efficiently}.

\paragraph{Local probe data and sandwich coefficients.}
For a set of qubits $T$, let
$\mathcal P_T:=\{P\in\mathcal P_n:\supp{P}\subseteq T\}$
and let $\Pi_T$ denote the Hilbert--Schmidt projection onto
$\operatorname{span}(\mathcal P_T)$.
The Pauli-input derivatives available from probes supported on
$T$ are
\begin{equation}
\label{eq:local_derivative_data}
    d_{O,Q}
    :=
    \left.\dv{}{t}
    \tr{O\,e^{t\mathcal L}(2^{-n}Q)}
    \right|_{t=0}
    =
    2^{-n}\tr{Q\,\mathcal L^\dagger(O)},
    \qquad O,Q\in\mathcal P_T.
\end{equation}
Since $O$ and $Q$ are supported on $T$, these derivatives are
unchanged if $\mathcal L^\dagger$ is replaced by
$\mathcal L_T^\dagger:=\Pi_T\circ\mathcal L^\dagger\circ\Pi_T$.
On $\operatorname{span}(\mathcal P_T)$, this restricted map has
the unique sandwich-basis expansion
\begin{equation}
\label{eq:local_sandwich_expansion}
    \mathcal L_T^\dagger(O)
    =
    \sum_{P_u,P_v\in\mathcal P_T}
    \xi_{uv}^{(T)}P_vOP_u.
\end{equation}
Thus the local derivative data determine the local sandwich
coefficients $\xi_{uv}^{(T)}$.

On the full system, write
$\xi_{uv}:=\xi_{uv}^{([n])}$.
For nonidentity $P_i,P_j$, the coefficient of
$O\mapsto P_jOP_i$ is exactly $\xi_{ij}=a_{ij}$:
the Hamiltonian and anticommutator terms each have an identity
sandwich factor on one side.
On a smaller region, however, other dissipative pairs can
contribute to the same local sandwich coefficient.

\paragraph{Local aliasing.}
For a target dissipative coefficient $a_{ij}$, the natural
region is $T_{ij}:=\supp{P_i}\cup\supp{P_j}$.
We now characterize the candidate pairs that contribute to
$\xi_{ij}^{(T_{ij})}$.

For any region $T$, write $\overline T:=[n]\setminus T$ and
denote the restriction of a Pauli string $P$ to $T$ by $P|_T$.
For $O\in\mathcal P_T$, tracing out the exterior Pauli factors
gives
\begin{equation}
\label{eq:projected_sandwich_map}
    \Pi_T(P_vOP_u)
    =
    \begin{cases}
        \bigl(P_v|_T\,O|_T\,P_u|_T\bigr)\otimes I_{\overline T},
        & P_u|_{\overline T}=P_v|_{\overline T},\\
        0,
        & \text{otherwise}.
    \end{cases}
\end{equation}
A surviving term therefore contributes to the target local
sandwich coefficient only if its two local factors also match
those of the target.

Consequently, a distinct candidate pair $(P_u,P_v)$ aliases
$(P_i,P_j)$ on $T_{ij}$ precisely when
\begin{equation}
\label{eq:common_exterior_alias}
    P_u=P_iR,\qquad
    P_v=P_jR,\qquad
    \supp{R}\cap T_{ij}=\emptyset,
\end{equation}
for a nonidentity Pauli $R$.
Indeed, the local restrictions must match $P_i,P_j$, while
\cref{eq:projected_sandwich_map} requires a common exterior factor.
Conversely, every $O\in\mathcal P_{T_{ij}}$ commutes with $R$,
so
$P_vOP_u=P_jROP_iR=P_jOP_i$.
Hence the target and the conflicting coefficient contribute
to the same local sandwich map.

\paragraph{Illustrative example.}
For $a,b\ge0$, consider
\begin{equation}
\label{eq:minimal_alias_example}
    \mathcal L(\rho)
    =
    a(Z_1\rho Z_1-\rho)
    +
    b(Z_1Z_2\rho Z_1Z_2-\rho).
\end{equation}
The natural support of the first term is $\{1\}$.
For every observable $O$ supported on qubit $1$,
$Z_1Z_2\,O\,Z_1Z_2=Z_1OZ_1$, so all derivative data on this
region depend only on $a+b$.
In particular, $d_{X_1,X_1}=-2(a+b)$.
The second term has the common exterior factor $R=Z_2$.
Including qubit $2$ distinguishes the terms: $X_1X_2$
anticommutes with $Z_1$ but commutes with $Z_1Z_2$, giving
$d_{X_1X_2,X_1X_2}=-2a$.

\paragraph{Removing aliases by support enlargement.}
For every conflicting candidate pair in
\cref{eq:common_exterior_alias}, choose one witness qubit
$s_{uv}\in\supp{R}$ and define
\begin{equation}
\label{eq:def_enlarged_dissipator_patch}
    \widetilde T_{ij}
    :=
    T_{ij}\cup
    \{s_{uv}:(P_u,P_v)\text{ conflicts with }(P_i,P_j)\}.
\end{equation}
The conflicts are determined from all candidate pairs
$P_u,P_v\in\widehat{\mathcal S}_D$; their coefficients need
not be known.

\begin{lemma}[Support enlargement; adapted from Lemma~3 of Ref.~\cite{montana2025efficiently}]
\label{lemma:support_enlargement_coefficient_learning}
For every $P_i,P_j\in\mathcal S_D$, at most $\mathfrak d$
witness qubits are required, so that
\begin{equation}
\label{eq:enlarged_dissipator_patch_size}
    |\widetilde T_{ij}|
    \le
    |\supp{P_i}\cup\supp{P_j}|+\mathfrak d.
\end{equation}
No candidate pair distinct from $(P_i,P_j)$ contributes to
its local sandwich coefficient on $\widetilde T_{ij}$.
In particular,
$\xi_{ij}^{(\widetilde T_{ij})}=a_{ij}$.
\end{lemma}

\begin{proof}
A conflicting pair is fixed by its first member $P_u=P_iR$,
since $P_v=P_jR$ is then determined.
Because $R$ is nonidentity and supported outside $T_{ij}$,
$P_u\neq P_i$ and the two strings overlap on $\supp{P_i}$.
Both diagonal dissipative components $\mathcal A^D_{uu}$ and
$\mathcal A^D_{ii}$ are present because $P_u,P_i\in\mathcal S_D$.
They are therefore adjacent in the dual interaction graph,
giving at most $\mathfrak d$ possible choices of $P_u$.

At the chosen witness $s_{uv}$,
\[
    P_i|_{s_{uv}}=P_j|_{s_{uv}}=I,
    \qquad
    P_u|_{s_{uv}}=P_v|_{s_{uv}}=R|_{s_{uv}}\neq I.
\]
Thus the competing pair no longer has the same local sandwich
factors as the target.
Any pair aliasing the target on $\widetilde T_{ij}$ would
also alias it on $T_{ij}$: its factors must be identity on
the added sites, and its exterior factors must agree.
It would therefore already belong to the list of conflicts
removed above.

Finally, Hamiltonian and anticommutator terms retain an identity
sandwich factor on one side after projection. Since
$P_i,P_j\neq I$, they cannot contribute to the target coefficient.
Hence $\xi_{ij}^{(\widetilde T_{ij})}=a_{ij}$, including when
the target coefficient is zero.
\end{proof}

\paragraph{Hamiltonian coefficients.}
A dissipative sandwich term can also locally imitate a
one-sided Hamiltonian map.
Let $H:=\sum_{P_p\in\mathcal S_H}h_pP_p$ and define
$B:=\sum_{u,v\in\mathcal S_D}a_{uv}P_vP_u$.
Hermiticity of the Kossakowski matrix gives
\[
    B^\dagger
    =
    \sum_{u,v}a_{uv}^*P_uP_v
    =
    \sum_{u,v}a_{uv}P_vP_u
    =
    B.
\]
Thus $B=\sum_p b_pP_p$ with $b_p\in\mathbb R$.
Writing
\[
    \mathcal L^\dagger(O)
    =
    iHO-iOH+\sum_{u,v}a_{uv}P_vOP_u
    -\tfrac12BO-\tfrac12OB,
\]
we see that the global coefficient of $O\mapsto P_pO$,
for $P_p\neq I$, is
\begin{equation}
\label{eq:hamiltonian_one_sided_coefficient}
    \xi_{0p}=ih_p-\tfrac12b_p,
    \qquad
    h_p=\Im(\xi_{0p}).
\end{equation}
Taking the imaginary part removes the anticommutator contribution.
On a smaller region, however, dissipative sandwich terms can add
an extra imaginary contribution to the same coefficient. For a candidate $P_p\in\widehat{\mathcal S}_H$, set
$T_p:=\supp{P_p}$.
By \cref{eq:projected_sandwich_map}, a dissipative pair aliases
$O\mapsto P_pO$ on $T_p$ precisely when
\begin{equation}
\label{eq:hamiltonian_common_exterior_alias}
    P_u=R,\qquad
    P_v=P_pR,\qquad
    R\neq I,\qquad
    \supp{R}\cap T_p=\emptyset.
\end{equation}
Indeed, every $O\in\mathcal P_{T_p}$ commutes with $R$, so
\[
    a_{uv}P_vOP_u
    =
    a_{uv}P_pROR
    =
    a_{uv}P_pO.
\]
Thus this pair adds $a_{uv}$ to the target local sandwich
coefficient. With a single such conflict,
\[
    \xi_{0p}^{(T_p)}
    =
    ih_p-\tfrac12b_p+a_{uv},
    \qquad
    \Im\left(\xi_{0p}^{(T_p)}\right)
    =
    h_p+\Im(a_{uv}).
\]
With several conflicts, their contributions add.
Therefore taking the imaginary part alone does not isolate
$h_p$ on $T_p$.
As before, choose one witness from each exterior factor $R$
and add these witnesses to $T_p$, obtaining $\widetilde T_p$.

If there are $m$ conflicting pairs, their strings $P_v=P_pR$
all agree with $P_p$ on $T_p$ and therefore pairwise overlap.
Each $P_v$ also overlaps its partner $P_u=R$ on $\supp{R}$.
This partner is distinct from every string in the first group,
since it is identity on $T_p$.
Hence the diagonal component associated with any $P_v$ has
at least $(m-1)+1=m$ neighbors.
It follows that $m\le\mathfrak d$ and
\begin{equation}
\label{eq:enlarged_hamiltonian_patch_size}
    |\widetilde T_p|
    \le
    |\supp{P_p}|+\mathfrak d.
\end{equation}

The witness argument removes every dissipative alias of the
one-sided target. A different global one-sided Pauli map either
vanishes under projection or retains a different Pauli label.
Consequently,
\begin{equation}
\label{eq:identified_hamiltonian_coefficient}
    \xi_{0p}^{(\widetilde T_p)}
    =
    ih_p-\tfrac12b_p,
    \qquad
    h_p=\Im\left(\xi_{0p}^{(\widetilde T_p)}\right).
\end{equation}
This also applies to false-positive Hamiltonian candidates,
for which $h_p=0$.

\paragraph{From local identifiability to reconstruction.}
We have therefore associated every candidate coefficient with
a known region on which its target sandwich coefficient is
unambiguous.
Let
\begin{equation}
    k:=\max\{
    \max_{P_p\in\widehat{\mathcal S}_H}|\supp{P_p}|,
    \max_{P_i,P_j\in\widehat{\mathcal S}_D}
    |\supp{P_i}\cup\supp{P_j}|
    \},
\end{equation}
with maxima over empty families taken to be zero.
Every enlarged region constructed above has size at most
$k+\mathfrak d$.

One possible reconstruction is the full local Pauli inversion
of Lemma~4 in Ref.~\cite{montana2025efficiently}.
For $P_i,P_j\in\mathcal P_T$, write
$P_jOP_i=\omega_{ij}(O)Q_{ij}(O)$, where
$\omega_{ij}(O)\in\{\pm1,\pm i\}$ and
$Q_{ij}(O)\in\mathcal P_T$ is a phaseless Hermitian Pauli.
The inversion is
\begin{equation}
\label{eq:stable_local_pauli_inverse}
    \xi_{ij}^{(T)}
    =
    \frac1{4^{|T|}}
    \sum_{O\in\mathcal P_T}
    \omega_{ij}(O)^*d_{O,Q_{ij}(O)}.
\end{equation}
Its weights have absolute sum one, but it uses
$4^{|T|}\le4^{k+\mathfrak d}$ derivatives per target.
The next subsection constructs a small target-dependent
probe family directly, without forming this exhaustive list.

\subsection{Deterministic targeted probe construction}
\label{sec:deterministic_targeted_probes}

The preceding subsection identifies an enlarged region on which
each target coefficient is distinguishable from the other
candidate terms. We now construct a small family of Pauli probes
on that region whose responses cancel all unwanted contributions.
The construction uses only the known candidate Pauli labels and
does not require exhaustive local inversion.

\paragraph{Identifying the contributions to cancel.}
For a dissipative target $a_{ij}$, take $T=\widetilde T_{ij}$ and
reconstruct $\xi_{ij}^{(T)}=a_{ij}$.
For a Hamiltonian target $h_p$, take $(i,j)=(0,p)$ and
$T=\widetilde T_p$, so that $h_p=\Im(\xi_{0p}^{(T)})$.
Thus in either case we first reconstruct the local coefficient
$\xi_{ij}^{(T)}$ of $O\mapsto P_jOP_i$.

For each $O\in\mathcal P_T$, use the previously defined
factorization $P_jOP_i=\omega_{ij}(O)Q_{ij}(O)$ and define
the phase-corrected response
\begin{equation}
\label{eq:targeted_phase_corrected_response}
    z_{ij}(O)
    :=
    \omega_{ij}(O)^*d_{O,Q_{ij}(O)}
    =
    2^{-n}\tr{
        (P_jOP_i)^\dagger\mathcal L_T^\dagger(O)
    }.
\end{equation}
The target sandwich term contributes exactly $\xi_{ij}^{(T)}$
to this expression, independently of $O$.

Consider another local sandwich term
$\xi_{uv}^{(T)}P_vOP_u$ in
\cref{eq:local_sandwich_expansion}.
Since $P_vOP_u\propto P_vP_uO$ and
$P_jOP_i\propto P_jP_iO$, Pauli orthogonality gives
\begin{equation}
\label{eq:targeted_product_sector}
    2^{-n}\tr{(P_jOP_i)^\dagger P_vOP_u}\neq0
    \quad\Longrightarrow\quad
    P_vP_u\propto P_jP_i.
\end{equation}
We call the pairs satisfying this product condition the
\emph{product sector} of the target.

For a pair in this sector, let $R_{uv}\in\mathcal P_T$ be
the phaseless Pauli proportional to $P_uP_i$.
The product condition also gives $P_jP_v\propto R_{uv}$.
For any Pauli operators $R,S$ on the same qubits, define their
commutation sign $s_R(S)\in\{+1,-1\}$ by $RS=s_R(S)SR$.
Cyclicity of the trace and $O^2=I$ then give
\begin{equation}
\label{eq:targeted_sandwich_sign}
\begin{aligned}
    2^{-n}\tr{(P_jOP_i)^\dagger P_vOP_u}
    &=
    2^{-n}\tr{P_uP_iOP_jP_vO}\\
    &=
    s_{R_{uv}}(O)\,
    2^{-n}\tr{P_uP_iP_jP_v}\\
    &=
    \gamma_{uv}s_{R_{uv}}(O),
\end{aligned}
\end{equation}
where
$\gamma_{uv}:=2^{-n}\tr{P_uP_iP_jP_v}\in\{\pm1,\pm i\}$
is a known phase for the fixed target $(i,j)$.
The target itself has $R_{ij}=I$ and $\gamma_{ij}=1$.
Every distinct local pair in the same product sector has
$R_{uv}\neq I$.

Let $\mathcal C_{ij}(T)$ denote the set of local labels $(u,v)\neq(i,j)$
with $P_u,P_v\in\mathcal P_T$ and $P_vP_u\propto P_jP_i$ that arise, via
\cref{eq:projected_sandwich_map}, from a sandwich term of a candidate
component ($P_p\in\widehat{\mathcal S}_H$, or $P_u,P_v\in\widehat{\mathcal S}_D$).
Since $\mathcal S_H\subseteq\widehat{\mathcal S}_H$ and
$\mathcal S_D=\widehat{\mathcal S}_D$, every $(u,v)\neq(i,j)$ in the product
sector with $\xi^{(T)}_{uv}\neq0$ lies in $\mathcal C_{ij}(T)$; the remaining
labels have $\xi^{(T)}_{uv}=0$ and do not affect
\cref{eq:targeted_response_pair_expansion}. By \cref{eq:targeted_sandwich_sign},
\begin{equation}
\label{eq:targeted_response_pair_expansion}
    z_{ij}(O)
    =
    \xi_{ij}^{(T)}
    +
    \sum_{(u,v)\in\mathcal C_{ij}(T)}
        \gamma_{uv}\xi_{uv}^{(T)}
        s_{R_{uv}}(O).
\end{equation}
For fixed target $(i,j)$, the map
$(u,v)\mapsto R_{uv}$ is one-to-one within this product sector:
$R_{uv}$ determines $P_u\propto R_{uv}P_i$, and the product
condition then determines $P_v$ uniquely.
Thus the nuisance strings
$\{R_{uv}:(u,v)\in\mathcal C_{ij}(T)\}$
are distinct nonidentity Paulis.

The set $\mathcal C_{ij}(T)$ is constructed entirely classically.
Expand each candidate Hamiltonian or dissipative component into
its one- and two-sided sandwich terms, apply
\cref{eq:projected_sandwich_map} to discard terms that vanish on
$T$, and retain only the resulting local pairs whose Pauli product
is proportional to $P_jP_i$. In particular, for a dissipative
component
\[
    P_vOP_u-\tfrac12P_vP_uO-\tfrac12OP_vP_u,
\]
all three sandwich terms lie in the same product sector
$P_vP_u$. No knowledge of the coefficient values is required.

The reconstruction problem has therefore reduced to choosing
observables whose average cancels every known sign
$s_{R_a}(O)$ while leaving the constant target contribution
$\xi_{ij}^{(T)}$ unchanged.

\paragraph{Choosing the cancellation generators.}
If a Pauli $G$ anticommutes with $R_a$, then
$s_{R_a}(GO)=-s_{R_a}(O)$.
Pairing the observables $O$ and the phaseless product $GO$
therefore cancels that nuisance contribution.
We next construct a small collection of such generators that
covers the entire nuisance list.

\begin{lemma}[Deterministic choice of cancellation generators]
\label{lemma:deterministic_cancellation_generators}
Given $m\ge1$ nonidentity Pauli strings $R_1,\ldots,R_m$ on a
patch $T$ of size $\ell$, one can deterministically construct
$G_1,\ldots,G_s\in\mathcal P_T$ such that every $R_a$
anticommutes with at least one $G_j$, with
\begin{equation}
\label{eq:deterministic_generator_count}
    s\le\lceil\log_2(m+1)\rceil.
\end{equation}
Using explicit Pauli labels, the construction takes
$\mathcal O(\ell m)$ classical operations.
\end{lemma}

\begin{proof}
We first construct one Pauli $G$ that anticommutes with at least
half of any given set $\mathcal U$ of nonidentity nuisance strings.

Order the qubits of $T$ as $1,\ldots,\ell$.
For each $R\in\mathcal U$, let
\[
    q(R):=\max\supp{R}
\]
be its last nonidentity qubit, and group the strings according to
this value:
\[
    \mathcal U_q:=\{R\in\mathcal U:q(R)=q\}.
\]
These sets partition $\mathcal U$. We construct
\[
    G=g_1\otimes\cdots\otimes g_\ell,
    \qquad g_q\in\{I,X,Y,Z\},
\]
from left to right.
Suppose $g_1,\ldots,g_{q-1}$ have already been chosen, and
consider the strings $R\in\mathcal U_q$.
For each such $R$, the commutation sign contributed by the first
$q-1$ sites is already fixed. Moreover, $R|_q\neq I$, while
$R|_p=I$ for every $p>q$.

Thus the choice of $g_q$ completely determines whether $R$
commutes or anticommutes with the final $G$.
Since $R|_q\neq I$, exactly two of the four choices
$g_q\in\{I,X,Y,Z\}$ give each of the two possible commutation
signs. Hence each $R\in\mathcal U_q$ is made to anticommute with
$G$ by exactly two of the four possible choices of $g_q$.

Suppose $g_1,\ldots,g_{q-1}$ have already been fixed.
For each possible choice
$\sigma\in\{I,X,Y,Z\}$ at site $q$, define
\[
    C_q(\sigma)
    :=
    \bigl|
        \{R\in\mathcal U_q:
        s_R(g_1\otimes\cdots\otimes g_{q-1}\otimes
            \sigma\otimes I^{\otimes(\ell-q)})=-1\}
    \bigr|.
\]
Because every $R\in\mathcal U_q$ is identity on all sites
$p>q$, $C_q(\sigma)$ is exactly the number of strings in
$\mathcal U_q$ that will anticommute with the final $G$ if we
choose $g_q=\sigma$; later choices cannot change this.

Now fix one $R\in\mathcal U_q$.
The contribution from sites $1,\ldots,q-1$ has already fixed a
sign $\eta_R\in\{\pm1\}$.
Since $R|_q\neq I$, among the four choices
$\sigma\in\{I,X,Y,Z\}$, exactly two commute with $R|_q$ and
exactly two anticommute with it.
Therefore exactly two choices of $\sigma$ satisfy
\[
    \eta_R\,s_{R|_q}(\sigma)=-1,
\]
and hence make this particular $R$ anticommute with $G$.

Summing this statement over all $R\in\mathcal U_q$ gives the
exact identity
\[
    C_q(I)+C_q(X)+C_q(Y)+C_q(Z)
    =
    2|\mathcal U_q|.
\]
Therefore
\[
    \max_{\sigma\in\{I,X,Y,Z\}}C_q(\sigma)
    \ge
    \frac{
        C_q(I)+C_q(X)+C_q(Y)+C_q(Z)
    }{4}
    =
    \frac{|\mathcal U_q|}{2}.
\]
Choose $g_q$ to be the first maximizer of $C_q(\sigma)$ in the order $I,X,Y,Z$.
Then at least $|\mathcal U_q|/2$ strings in $\mathcal U_q$
anticommute with the final $G$.

Do this successively for $q=1,\ldots,\ell$.
For each $q$, let
\[
    \mathcal A_q
    :=
    \{R\in\mathcal U_q:s_R(G)=-1\}.
\]
By the construction above, $|\mathcal A_q|\ge\tfrac{|\mathcal U_q|}{2}$.
The sets $\mathcal U_1,\ldots,\mathcal U_\ell$ partition
$\mathcal U$, and hence the sets $\mathcal A_q$ are disjoint.
Therefore the total number of strings that anticommute with $G$
is at least
\[
    \sum_{q=1}^{\ell}|\mathcal A_q|
    \ge
    \frac12\sum_{q=1}^{\ell}|\mathcal U_q|
    =
    \frac{|\mathcal U|}{2}.
\]

We now repeat this construction on the strings that were not
canceled.  Let $m_t$ be the number remaining before round $t$.
The generator constructed in that round anticommutes with at
least $m_t/2$ of them, so
\[
    m_{t+1}
    \le
    m_t-\left\lceil\frac{m_t}{2}\right\rceil
    =
    \left\lfloor\frac{m_t}{2}\right\rfloor.
\]
Inductively, $m_t\le\left\lfloor\frac{m}{2^t}\right\rfloor$.
Hence $m_t=0$ once $2^t>m$, which occurs for
\[
    t=\lceil\log_2(m+1)\rceil.
\]
Thus at most this many generators are needed, and every nuisance
string anticommutes with at least one of them.

Finally, in a round with $u$ surviving strings, processing all
$\ell$ sites and updating the four counts costs $\mathcal O(\ell u)$.
Since the number of surviving strings halves each round, $\sum_t m_t<2m$.
Hence the total classical cost is
\[
    \mathcal O\!\left(\ell\sum_t m_t\right)
    =
    \mathcal O(\ell m).
\]
\end{proof}

\paragraph{Exact cancellation and stability.}
Apply \cref{lemma:deterministic_cancellation_generators} to
the nuisance strings $\{R_{uv}:(u,v)\in\mathcal C_{ij}(T)\}$ of \cref{eq:targeted_response_pair_expansion}.
For each binary vector
$\vb b=(b_1,\ldots,b_s)\in\{0,1\}^s$, let
$O_{\vb b}\in\mathcal P_T$ be the phaseless Pauli satisfying
\[
    O_{\vb b}\propto G_1^{b_1}\cdots G_s^{b_s}.
\]
The generators need not commute: commutation signs are unaffected by overall phases and satisfy
$s_R(GO)=s_R(G)s_R(O)$.
Hence, for every nuisance string,
\begin{equation}
\label{eq:deterministic_sign_cancellation}
\begin{aligned}
    \frac1{2^s}\sum_{\vb b\in\{0,1\}^s}
        s_{R_a}(O_{\vb b})
    &=
    \frac1{2^s}\sum_{\vb b\in\{0,1\}^s}
        \prod_{j=1}^s s_{R_a}(G_j)^{b_j}\\
    &=
    \prod_{j=1}^s\frac{1+s_{R_a}(G_j)}2
    =0.
\end{aligned}
\end{equation}
The last equality holds because at least one generator
anticommutes with $R_a$.
Averaging \cref{eq:targeted_response_pair_expansion} therefore
gives
\begin{equation}
\label{eq:deterministic_targeted_reconstruction}
    \xi_{ij}^{(T)}
    =
    \frac1{2^s}\sum_{\vb b\in\{0,1\}^s}
        \omega_{ij}(O_{\vb b})^*
        d_{O_{\vb b},Q_{ij}(O_{\vb b})}.
\end{equation}
For $m\ge1$, there are at most $2^s\le2m$ probes, and the
absolute reconstruction weights sum to one.
Repeated probes can be merged; let $\mathcal R_{ij}$ denote
the resulting set of distinct pairs
$(O_{\vb b},Q_{ij}(O_{\vb b}))$.

Replacing each derivative by an estimate defines
$\widehat\xi_{ij}^{(T)}$.
If every required derivative has error at most $\varepsilon_d$,
then
\[
    |\widehat\xi_{ij}^{(T)}-\xi_{ij}^{(T)}|
    \le
    \frac1{2^s}\sum_{\vb b}
    |\omega_{ij}(O_{\vb b})|
    \left|
        \widehat d_{O_{\vb b},Q_{ij}(O_{\vb b})}
        -
        d_{O_{\vb b},Q_{ij}(O_{\vb b})}
    \right|
    \le
    \frac1{2^s}\sum_{\vb b}\varepsilon_d
    =
    \varepsilon_d,
\]
For a dissipative target this estimates the complex coefficient
$a_{ij}$ directly.
For a Hamiltonian target, taking the imaginary part with
$(i,j)=(0,p)$ estimates $h_p$ with the same error bound.
Taking real or imaginary parts of the reconstruction weights
also gives real-coordinate formulas with absolute weight sum
at most one.

\paragraph{Illustrative example.}
For $a,b,c>0$, consider
\[
    \mathcal L(\rho)
    =
    a(Z_1\rho Z_1-\rho)
    +b(Z_2\rho Z_2-\rho)
    +c(Z_1Z_2\rho Z_1Z_2-\rho),
\]
and reconstruct the diagonal coefficient $a_{ii}=a$ associated
with $P_i=Z_1$ on $T=\{1,2\}$.
Here $Q_{ii}(O)=O$ and $\omega_{ii}(O)=s_{Z_1}(O)$.
Writing $s_1:=s_{Z_1}(O)$ and $s_2:=s_{Z_2}(O)$ gives
\[
\begin{aligned}
    d_{O,O}
    &=a(s_1-1)+b(s_2-1)+c(s_1s_2-1),\\
    z_{ii}(O)
    &=a+b s_1s_2+c s_2-(a+b+c)s_1.
\end{aligned}
\]
Thus the nuisance strings are $Z_1$, $Z_2$, and $Z_1Z_2$;
the last term includes the anticommutator contributions.

The last-site groups are
$\mathcal B_1=\{Z_1\}$ and
$\mathcal B_2=\{Z_2,Z_1Z_2\}$.
At qubit $1$, the rule chooses $g_1=X$.
At qubit $2$, every choice cancels one of the two strings in
$\mathcal B_2$, so the tie rule chooses $g_2=I$.
Hence $G_1=X_1$, which anticommutes with $Z_1$ and $Z_1Z_2$.
The sole survivor is $Z_2$, giving $G_2=X_2$ in the next round.

The selected observables are therefore
$I,X_1,X_2,X_1X_2$, and
\cref{eq:deterministic_targeted_reconstruction} becomes
\begin{equation}
\label{eq:deterministic_cancellation_example}
    a=
    \frac14\left(
        d_{I,I}-d_{X_1,X_1}
        +d_{X_2,X_2}-d_{X_1X_2,X_1X_2}
    \right).
\end{equation}
Indeed, the four derivatives are respectively
$0$, $-2(a+c)$, $-2(b+c)$, and $-2(a+b)$, so the right-hand
side equals $a$.
The known zero derivative $d_{I,I}$ need not be measured.

\paragraph{Number of required probes.}
We now bound the nuisance list directly using the candidate
Pauli strings and the dual interaction graph.

\begin{lemma}[Targeted probe counts]
\label{lemma:targeted_probe_counts}
The deterministic reconstruction uses at most $4\mathfrak d+6$
probes for a Hamiltonian or off-diagonal dissipative coefficient,
and at most $2(\mathfrak d+1)^2$ probes for a diagonal dissipative
coefficient.
A diagonal coefficient also requires at most
$2|\widehat{\mathcal S}_D|$ probes.
\end{lemma}

\begin{proof}
First consider a Hamiltonian or off-diagonal target, for which
$P_jP_i\not\propto I$.
Fix a qubit $q\in\supp{P_jP_i}$ and define
$\mathcal D_q:=\{P_u\in\widehat{\mathcal S}_D:
q\in\supp{P_u}\}$.
The diagonal components associated with these strings all
overlap on $q$ and form a clique in the dual interaction graph,
so $|\mathcal D_q|\le\mathfrak d+1$.

Every global dissipative pair satisfying
$P_vP_u\propto P_jP_i$ has at least one member in $\mathcal D_q$.
Either member determines the other uniquely, so there are at
most $\mathfrak d+1$ unordered pairs in this product sector.
Their two orientations give at most $2(\mathfrak d+1)$
two-sided sandwich terms.
All Hamiltonian and anticommutator contributions in the sector
have one of the two one-sided labels associated with the
phaseless Pauli proportional to $P_jP_i$.
Projection onto $T$ cannot increase the number of labels.
Thus there are at most $2\mathfrak d+4$ local sandwich labels,
including the target, and hence $m\le2\mathfrak d+3$.
The construction therefore uses at most $4\mathfrak d+6$ probes.

For a diagonal target $a_{ii}$, the product condition is
$P_vP_u\propto I$, so only diagonal dissipative components
contribute.
Let
$\mathcal D(T):=\{P_u\in\widehat{\mathcal S}_D:
\supp{P_u}\cap T\neq\emptyset\}$.
Components outside this set vanish on every observable
supported on $T$.

Write $T=\widetilde T_{ii}
=\supp{P_i}\cup\{s_1,\ldots,s_t\}$, where $t\le\mathfrak d$.
At most $\mathfrak d+1$ strings in
$\widehat{\mathcal S}_D$ intersect $\supp{P_i}$, since their
diagonal components are either $\mathcal A^D_{ii}$ itself
or adjacent to it.
At most $\mathfrak d+1$ strings act on each witness $s_r$,
since their diagonal components form a clique at that qubit.
Consequently,
\[
    |\mathcal D(T)|
    \le
    (t+1)(\mathfrak d+1)
    \le
    (\mathfrak d+1)^2.
\]
Each string in $\mathcal D(T)$ gives one local diagonal
sandwich label, while all anticommutator contributions share
the identity sandwich map.
Excluding the target leaves at most $|\mathcal D(T)|$
nuisance labels.
Hence the probe count is at most
$2|\mathcal D(T)|\le2(\mathfrak d+1)^2$.
The additional bound follows from
$|\mathcal D(T)|\le|\widehat{\mathcal S}_D|$.
\end{proof}

Let $\widehat M_D:=|\widehat{\mathcal S}_D|$ and
$\widehat M:=|\widehat{\mathcal S}_H|+\widehat M_D^2$.
Let $\mathcal R$ be the union of the selected probe pairs over
all Hamiltonian and dissipative targets.
Even counting both orientations of every off-diagonal
coefficient, \cref{lemma:targeted_probe_counts} gives
\begin{equation}
\label{eq:total_selected_probe_bound}
\begin{aligned}
    |\mathcal R|
    &\le
    2\widehat M_D^2
    +
    (\widehat M-\widehat M_D)(4\mathfrak d+6)\\
    &\le
    (4\mathfrak d+8)\widehat M.
\end{aligned}
\end{equation}
A probe shared by several reconstructions needs to be estimated
only once.

\paragraph{Classical construction cost.}
For each target, the relevant global pairs can be found by
scanning $P_u\in\widehat{\mathcal S}_D$:
the product condition uniquely determines the possible partner
$P_v$, which is checked by dictionary lookup.
For a diagonal target, a scan identifies the strings intersecting
its enlarged patch.
Using length-$n$ Pauli labels and a sorted dictionary, these
scans, together with the support enlargement and local
restrictions, cost
$\widetilde{\mathcal O}(n(\widehat M_D+1))$ per target.

Once the nuisance list is available,
\cref{lemma:deterministic_cancellation_generators} constructs
the generators in $\mathcal O(|T|m)$ operations.
All products defining the observables can then be generated
by successively adjoining each generator to the current list,
using $\mathcal O(|T|2^s)$ operations.
Computing their input Paulis and reconstruction phases has
the same cost.
Thus the remaining work is $\mathcal O(n(m+1))$ per target.

For Hamiltonian and off-diagonal targets,
$m+1=\mathcal O(\mathfrak d)$.
For diagonal targets, $m\le\widehat M_D$.
Summing over all targets therefore gives
\begin{equation}
\label{eq:targeted_probe_construction_cost}
    t_{\mathrm{pre}}
    =
    \widetilde{\mathcal O}\left(
        n\widehat M(\widehat M_D+\mathfrak d+1)
    \right).
\end{equation}
The selected probe families therefore have polynomial
construction cost, without enumerating the full local Pauli
basis or solving a reconstruction linear system.

\subsection{Estimating the selected time derivatives}
\label{sec:pauli_derivative_estimation_main}

Having constructed a deterministic set of probes for each
coefficient, it remains to estimate their derivatives from
short-time experimental data.
Each selected probe $(O,Q)$ requires
\begin{equation}
\label{eq:pauli_input_derivative}
    d_{O,Q}
    =
    \left.
    \dv{}{t}
    \tr{O\,\mathcal E_t(2^{-n}Q)}
    \right|_{t=0},
    \qquad
    \mathcal E_t=e^{t\mathcal L}.
\end{equation}
For $Q\neq I$, let
$\rho_{\pm Q}:=2^{-n}(I\pm Q)$.
Since
$2^{-n}Q=\frac12(\rho_{+Q}-\rho_{-Q})$,
\begin{equation}
\label{eq:pauli_input_physical_realization}
    \tr{O\,\mathcal E_t(2^{-n}Q)}
    =
    \frac12\left(
        \tr{O\,\mathcal E_t(\rho_{+Q})}
        -
        \tr{O\,\mathcal E_t(\rho_{-Q})}
    \right).
\end{equation}
The states $\rho_{\pm Q}$ can be prepared as mixtures of product
Pauli eigenstates. For $Q=I$, the input is simply the maximally
mixed state. Thus every selected Pauli-input probe is accessible
using product-state preparations and Pauli measurements.

\paragraph{Derivative bounds for the selected observables.}
The derivative bound of
\cref{thm:observable_derivative_bounds} depends on
$\mathfrak d_O$, the number of genuine Lindbladian components
whose supports overlap $O$.
Although a selected observable may act on many qubits, the
support-enlargement construction allows us to bound
$\mathfrak d_O$ solely in terms of the dual-graph degree.

\begin{lemma}[Neighborhood size of selected observables]
\label{lemma:enlarged_patch_neighborhood}
Let $O$ be any Pauli observable selected by the deterministic
construction of \cref{sec:deterministic_targeted_probes}.
Then
\begin{equation}
\label{eq:enlarged_patch_neighborhood}
    \mathfrak d_O
    \le
    (\mathfrak d+2)(\mathfrak d+1).
\end{equation}
\end{lemma}

\begin{proof}
Every selected observable is supported on the enlarged region of
its target. We first show that each such region is contained in
the union of the supports of at most $\mathfrak d+2$ genuine
Lindbladian components.

For a dissipative target $a_{ij}$, the natural region satisfies
\[
    T_{ij}
    =
    \supp{P_i}\cup\supp{P_j}
    =
    \supp{\mathcal A^D_{ii}}
    \cup
    \supp{\mathcal A^D_{jj}},
\]
where both diagonal components are genuine because
$P_i,P_j\in\mathcal S_D$.
If the enlargement adds witness qubits
$s_1,\ldots,s_m$, with $m\le\mathfrak d$, each $s_r$ was chosen
from the support of a conflicting string
$P_{u_r}\in\mathcal S_D$.
Hence
\[
    \widetilde T_{ij}
    \subseteq
    \supp{\mathcal A^D_{ii}}
    \cup
    \supp{\mathcal A^D_{jj}}
    \cup
    \bigcup_{r=1}^m
    \supp{\mathcal A^D_{u_ru_r}},
\]
so $\widetilde T_{ij}$ is covered by at most
$\mathfrak d+2$ genuine components.

For a Hamiltonian candidate $P_p$, the natural region
$T_p=\supp{P_p}$ is covered by one genuine component whenever
$P_p\in\mathcal S_H\cup\mathcal S_D$.
Otherwise, by the structure-learning output property,
$\chi^{(2)}_{pp}(0)>0$.
Since $\chi^{(2)}_{pp}(0)$ is the coefficient of
$P_p\,\cdot\,P_p$ in $\mathcal L^2$, some composition of two
genuine Lindbladian components
$\mathcal A_1,\mathcal A_2$ contributes to it.
A composition of these components acts trivially outside
$\supp{\mathcal A_1}\cup\supp{\mathcal A_2}$, and therefore
\[
    \supp{P_p}
    \subseteq
    \supp{\mathcal A_1}\cup\supp{\mathcal A_2}.
\]
Adding the at most $\mathfrak d$ witness components again gives
a cover of $\widetilde T_p$ by at most $\mathfrak d+2$ genuine
components.

Thus, in either case, there exist genuine components
$\mathcal A_1,\ldots,\mathcal A_s$, with
$s\le\mathfrak d+2$, such that
\[
    \supp{O}
    \subseteq
    \bigcup_{r=1}^s\supp{\mathcal A_r}.
\]
Now let $\mathcal B$ be any genuine component counted by
$\mathcal N(O)$.
Then $\mathcal B$ overlaps $O$ on some qubit, which belongs to
$\supp{\mathcal A_r}$ for at least one $r$.
Hence $\mathcal B$ either equals $\mathcal A_r$ or is adjacent
to it in the dual interaction graph. Therefore
\[
    \mathcal N(O)
    \subseteq
    \bigcup_{r=1}^s
    \left(
        \{\mathcal A_r\}\cup N_G(\mathcal A_r)
    \right),
\]
and consequently
\[
    \mathfrak d_O
    \le
    s(\mathfrak d+1)
    \le
    (\mathfrak d+2)(\mathfrak d+1).
\]
\end{proof}

For a selected probe $(O,Q)$, define
\begin{equation}
\label{eq:selected_probe_time_trace}
    f_{O,Q}(t)
    :=
    \tr{O\,e^{t\mathcal L}(2^{-n}Q)}.
\end{equation}
Since $\|2^{-n}Q\|_1=1$,
\cref{thm:observable_derivative_bounds,lemma:enlarged_patch_neighborhood}
imply
\[
    \left|
        f_{O,Q}^{(\ell)}(t)
    \right|
    \le
    \bigl[
        2(\mathfrak d+2)(\mathfrak d+1)
    \bigr]^\ell
    \ell!.
\]
Independently, the sparsity-only bound of
\ref{remark:observable_derivative_bounds_M_sparse} gives
\[
    \left|
        f_{O,Q}^{(\ell)}(t)
    \right|
    \le
    (2M)^\ell
    \le
    (2M)^\ell \ell!.
\]
Combining the two bounds, we may therefore define $\Gamma
    :=
    \min\left\{
        (\mathfrak d+2)(\mathfrak d+1),\, M
    \right\},$
so that, uniformly in $t$,
\begin{equation}
\label{eq:enlarged_patch_derivative_bound}
    \left|
        f_{O,Q}^{(\ell)}(t)
    \right|
    \le
    (2\Gamma)^\ell \ell!,
    \qquad \ell\ge1.
\end{equation}
Thus the relevant smoothness scale is
$\Gamma=\mathcal O(\min\{\mathfrak d^2,M\})$.

\paragraph{Chebyshev derivative estimation.}
We now apply the first-derivative Chebyshev estimator from
\cref{cor:first_chebyshev_derivative_parameters_factorial}.

\begin{lemma}[Estimation of a selected Pauli derivative]
\label{lemma:pauli_expectation_derivative_estimation}
Let $(O,Q)$ be any probe selected in
\cref{sec:deterministic_targeted_probes}.
For any $\varepsilon,\delta\in(0,1)$, one can estimate
$d_{O,Q}=f_{O,Q}'(0)$ to additive accuracy $\varepsilon$ with
probability at least $1-\delta$ using
\begin{equation}
\label{eq:single_probe_derivative_complexity}
    \widetilde{\mathcal O}\left(
        \frac{\Gamma^2}{\varepsilon^2}
    \right),
    \qquad
    \Gamma
    :=
    \min\left\{
        (\mathfrak d+2)(\mathfrak d+1),\,M
    \right\}
    =
    \mathcal O\!\left(
        \min\{\mathfrak d^2,M\}
    \right).
\end{equation}
applications of $e^{t\mathcal L}$.
It suffices to use the $r+1$ Gauss--Chebyshev nodes on
$[0,\tau_{\max}]$, with
\begin{equation}
\label{eq:cheb_params_pauli_derivative}
    \tau_{\max}
    =
    \frac{1}{4\Gamma},
    \qquad
    r
    =
    \max\left\{
        16,\,
        \left\lceil
        4\log_2\left(
            \frac{16\Gamma}{\varepsilon}
        \right)
        \right\rceil
    \right\}.
\end{equation}
\end{lemma}

\begin{proof}
The derivative bound
\cref{eq:enlarged_patch_derivative_bound} has the form required
by
\cref{cor:first_chebyshev_derivative_parameters_factorial}
with
\[
    B=1,
    \qquad
    \Lambda
    = 2\Gamma.
\]
The choices in
\cref{eq:cheb_params_pauli_derivative} therefore guarantee
additive derivative error at most $\varepsilon$ provided that
every sampled node value is estimated to accuracy
\begin{equation}
\label{eq:coefficient_node_accuracy}
    \varepsilon_s
    =
    \frac{\varepsilon}
    {20\Gamma r^3}.
\end{equation}

For $Q\neq I$, estimate each of the two expectation values in
\cref{eq:pauli_input_physical_realization} to accuracy
$\varepsilon_s$, with failure probability at most
$\delta_s/2$.
Their half-difference then estimates $f_{O,Q}(t)$ to accuracy
$\varepsilon_s$ with failure probability at most $\delta_s$.
The case $Q=I$ requires only one expectation value.
Since Pauli measurement outcomes lie in $\{-1,+1\}$,
Hoeffding's inequality requires
\[
    m_s
    =
    \mathcal O\left(
        \varepsilon_s^{-2}
        \log\frac{1}{\delta_s}
    \right)
\]
channel applications per Chebyshev node, up to a constant factor
for the two physical inputs.

Taking $\delta_s=\delta/(r+1)$ and applying a union bound over
the $r+1$ nodes gives
\begin{equation}
\label{eq:single_probe_derivative_complexity_explicit}
    m
    =
    \mathcal O\left(
        \frac{\Gamma^2}{\varepsilon^2}
        r^7
        \log\frac{r}{\delta}
    \right).
\end{equation}
Since
$r=\mathcal O(\log(\Gamma/\varepsilon))$, this gives
\cref{eq:single_probe_derivative_complexity}.
\end{proof}

\paragraph{Estimating all selected probes separately.}
Let $\widehat M_H:=|\widehat{\mathcal S}_H|$,
$\widehat M_D:=|\widehat{\mathcal S}_D|$, and
$\widehat M:=\widehat M_H+\widehat M_D^2$.
By \cref{eq:total_selected_probe_bound}, the deterministic
construction selects
\[
    |\mathcal R|
    =
    \mathcal O(\widehat M\,\mathfrak d)
\]
distinct probes in total.
Estimating them independently, with failure probability
$\delta/|\mathcal R|$ allocated to each derivative, therefore
gives
\begin{equation}
\label{eq:all_selected_probe_derivatives_sample_complexity}
N_{\mathrm{ser}}
=
\widetilde{\mathcal O}\left(
    \widehat M\,\mathfrak d\,
    \frac{\Gamma^2}{\varepsilon^2}
\right)
=
\widetilde{\mathcal O}\left(
    \widehat M\,
    \frac{
        \mathfrak d\,
        \min\{(\mathfrak d+1)^4,M^2\}
    }{\varepsilon^2}
\right).
\end{equation}
channel applications.
This already yields a complete coefficient-learning procedure
with no bounded-locality assumption.
When the selected probes are sufficiently local, the next
subsection shows how to estimate them in parallel.

Since $\Gamma\le M$, the resulting schedule always satisfies
\[
    \tau_{\max}\ge \frac{1}{4M}.
\]
Thus coefficient learning never requires a parametrically shorter
evolution-time scale than the sparsity-only $M^{-1}$ scale.

\subsection{Optional parallel estimation via shadow process tomography}
\label{sec:parallel_derivative_estimation}

The probe-by-probe procedure of
\cref{sec:pauli_derivative_estimation_main} already estimates all
selected derivatives without any bounded-locality assumption.
When the selected Pauli probes are sufficiently local, however,
their values at each Chebyshev node can be estimated simultaneously
using shadow process tomography.

Define the Hamiltonian and dissipative candidate localities
\[
    k_H
    :=
    \max_{P\in\widehat{\mathcal S}_H}|\supp{P}|,
    \qquad
    k_D
    :=
    \max_{P_i,P_j\in\widehat{\mathcal S}_D}
    |\supp{P_i}\cup\supp{P_j}|,
\]
with maxima over empty families taken to be zero.
By the support-enlargement construction,
every selected Hamiltonian probe $(O,Q)$ satisfies
\[
    |\supp{O}|,\ |\supp{Q}|
    \le k_H+\mathfrak d,
\]
while every selected dissipative probe satisfies
\[
    |\supp{O}|,\ |\supp{Q}|
    \le k_D+\mathfrak d.
\]

We use the following shadow-process primitive.

\begin{lemma}[Shadow process tomography, \cite{stilck2024efficient}]
\label{lemma:shadow_process_tomography}
Let $\Phi$ be an $n$-qubit quantum channel, and let
$\{(P^{(a)},Q^{(a)})\}_{a=1}^K$ be prescribed Pauli pairs satisfying
\[
    |\supp{P^{(a)}}|\le\omega_P,
    \qquad
    |\supp{Q^{(a)}}|\le\omega_Q.
\]
For any $\varepsilon,\delta\in(0,1)$, all transfer-matrix elements
\[
    2^{-n}\tr{P^{(a)}\Phi(Q^{(a)})},
    \qquad a=1,\ldots,K,
\]
can be estimated simultaneously to additive accuracy
$\varepsilon$, with probability at least $1-\delta$, using
\begin{equation}
\label{eq:shadow_process_tomography_complexity}
    \mathcal O\left(
        3^{\omega_P+\omega_Q}
        \varepsilon^{-2}
        \log\frac{K}{\delta}
    \right)
\end{equation}
samples.
The protocol uses only random product Pauli state preparations
and single-qubit Pauli measurements.
\end{lemma}

\begin{proof}
The result follows from Supplementary Note~8 of
\citet{stilck2024efficient}.
Each experimental record gives, for any requested Pauli pair
$(P,Q)$, an unbiased estimator $X_{P,Q}$ satisfying
\[
    \mathbb E[X_{P,Q}]
    =
    2^{-n}\tr{P\Phi(Q)},
    \qquad
    \mathbb E[X_{P,Q}^2]
    \le
    3^{|\supp{P}|+|\supp{Q}|}.
\]
A median-of-means estimator and a union bound over the $K$
requested pairs give
\cref{eq:shadow_process_tomography_complexity}.
\end{proof}

We now apply this primitive to the selected coefficient-learning
probes.

\begin{lemma}[Parallel estimation of the selected derivatives]
\label{lemma:parallel_patchwise_derivative_estimation}
For any $\varepsilon,\delta\in(0,1)$, all derivatives selected by
\cref{sec:deterministic_targeted_probes} can be estimated
simultaneously to additive accuracy $\varepsilon$, with probability
at least $1-\delta$, using
\begin{equation}
\label{eq:parallel_selected_derivative_complexity}
    \widetilde{\mathcal O}\left(
        \left(
            9^{k_H+\mathfrak d}
            +
            9^{k_D+\mathfrak d}
        \right)
        \frac{\Gamma^2}{\varepsilon^2}
    \right)
\end{equation}
applications of $e^{t\mathcal L}$.
\end{lemma}

\begin{proof}
Let $\mathcal R_H$ and $\mathcal R_D$ denote the selected probe
pairs associated with Hamiltonian and dissipative targets,
respectively, after removing duplicates, and let
$\mathcal R:=\mathcal R_H\cup\mathcal R_D$.
By \cref{eq:total_selected_probe_bound},
\[
    |\mathcal R|
    =
    \mathcal O(\widehat M\,\mathfrak d).
\]

By \cref{lemma:pauli_expectation_derivative_estimation}, all selected
derivatives can be estimated using the same $r+1$
Gauss--Chebyshev nodes on $[0,\tau_{\max}]$, with
\[
    \tau_{\max}
    =
    \frac{1}{4\Gamma},
    \qquad
    r
    =
    \mathcal O\left(
        \log\frac{\Gamma}{\varepsilon}
    \right).
\]
provided that every required node value
\[
    f_{O,Q}(t_m)
    =
    2^{-n}\tr{O\,\mathcal E_{t_m}(Q)}
\]
is estimated to accuracy
\begin{equation}
\label{eq:parallel_node_accuracy}
    \varepsilon_s
    =
    \frac{\varepsilon}
    {20\Gamma r^3}.
\end{equation}

Fix one node $t_m$.
For $(O,Q)\in\mathcal R_H$,
\[
    |\supp{O}|,\ |\supp{Q}|
    \le k_H+\mathfrak d,
\]
so \cref{lemma:shadow_process_tomography} estimates all corresponding
values simultaneously using
\[
    \mathcal O\left(
        9^{k_H+\mathfrak d}
        \varepsilon_s^{-2}
        \log\frac{|\mathcal R|}{\delta_s}
    \right)
\]
channel applications.
Likewise, all values associated with $\mathcal R_D$ can be
estimated using
\[
    \mathcal O\left(
        9^{k_D+\mathfrak d}
        \varepsilon_s^{-2}
        \log\frac{|\mathcal R|}{\delta_s}
    \right)
\]
applications.
Allocating failure probability $\delta_s/2$ to each nonempty
family changes only the constant inside the logarithm.

Set $\delta_s=\delta/(r+1)$ and apply a union bound over all
Chebyshev nodes. The resulting query complexity is
\[
    \mathcal O\left(
        (r+1)
        \left(
            9^{k_H+\mathfrak d}
            +
            9^{k_D+\mathfrak d}
        \right)
        \varepsilon_s^{-2}
        \log\frac{|\mathcal R|r}{\delta}
    \right).
\]
Substituting \cref{eq:parallel_node_accuracy} gives
\[
    \mathcal O\left(
        \left(
            9^{k_H+\mathfrak d}
            +
            9^{k_D+\mathfrak d}
        \right)
        \frac{\Gamma^2}{\varepsilon^2}
        r^7
        \log\frac{|\mathcal R|r}{\delta}
    \right).
\]
Since
$r=\mathcal O(\log(\Gamma/\varepsilon))$
and $|\mathcal R|=\mathcal O(\widehat M\mathfrak d)$,
this yields
\cref{eq:parallel_selected_derivative_complexity}.
\end{proof}

\begin{remark}[Classical post-processing]
\label{rem:classical-cost-shadow-process-tomography}
For a prescribed collection of $K$ Pauli pairs of support at most
$\omega$, a direct implementation updates all requested shadow
estimators in $\mathcal O(K\omega)$ operations per experimental
record, in addition to the $\mathcal O(n)$ cost of processing the
full record.
Thus the selected probe family can be processed in
\[
    \mathcal O\left(
        n+|\mathcal R|(k+\mathfrak d)
    \right)
\]
operations per sample.
The subsequent Chebyshev sums and coefficient reconstructions
require only
$\mathcal O(|\mathcal R|r)$ additional arithmetic operations.
\end{remark}

The parallel procedure is useful when $k_H$ and $k_D$ are small.
For highly nonlocal candidate terms, the probe-by-probe procedure
of \cref{sec:pauli_derivative_estimation_main} avoids the
exponential locality factor; the final algorithm may therefore
use whichever implementation gives the smaller query complexity.

\subsection{Coefficient-learning algorithm}
\label{sec:coefficient_learning_algorithm}

We now combine the deterministic probe construction with the
short-time derivative estimator developed above.
For each candidate coefficient, the algorithm constructs an
enlarged region on which the target sandwich coefficient is
unambiguous, deterministically selects a small family of Pauli
probes that isolates this coefficient, and estimates the required
short-time derivatives probe by probe.

\begin{boxedalgorithm}[H]{Ancilla-free Coefficient Learning
\label{algorithm:coefficient_learning}}
  \DontPrintSemicolon
  \SetKwInOut{Input}{Inputs}
  \SetKwInOut{Output}{Output}
  \SetKwBlock{Design}{Probe selection}{}
  \SetKwBlock{Schedule}{Schedule}{}

  \Input{
    Candidate supports satisfying
    \cref{eq:coefficient_learning_support_assumptions};
    accuracy $\varepsilon$;
    sparsity bound $M$;
    dual-graph degree upper bound $\mathfrak d$;
    failure probability $\delta$
  }
  \Output{
    Coefficient estimates
    $\{\widehat h_p\}$ and $\{\widehat a_{ij}\}$
  }

  \Design{
    \For{each dissipative target $a_{ij}$ with $i\le j$}{
      Construct the enlarged region $\widetilde T_{ij}$
      \tcp*{\cref{sec:stable_local_reconstruction}}
      Construct the targeted probe family
      $\mathcal R_{ij}$ and its reconstruction weights
      \tcp*{\cref{sec:deterministic_targeted_probes}}
    }
    \For{each Hamiltonian target $h_p$}{
      Construct the enlarged region $\widetilde T_p$
      \tcp*{\cref{sec:stable_local_reconstruction}}
      Construct the targeted probe family
      $\mathcal R_p$ and its reconstruction weights
      \tcp*{\cref{sec:deterministic_targeted_probes}}
    }
    $\mathcal R\gets
      \bigcup_{i\le j}\mathcal R_{ij}
      \cup\bigcup_p\mathcal R_p$
    \tcp*{remove duplicate probe pairs}
  }

  \Schedule{
    Choose $\tau_{\max}$, $r$, and $\varepsilon_s$
    according to
    \cref{eq:cheb_params_pauli_derivative,eq:coefficient_node_accuracy}
    with target derivative accuracy $\varepsilon$\;
    $\delta_d\gets\delta/|\mathcal R|$
    \tcp*{failure probability per derivative}
  }

  \For{$(O,Q)\in\mathcal R$}{
    Estimate $d_{O,Q}$ to accuracy $\varepsilon$ and
    failure probability $\delta_d$
    using the Gauss--Chebyshev procedure
    \tcp*{\cref{lemma:pauli_expectation_derivative_estimation}}
  }

  \For{each dissipative target $a_{ij}$ with $i\le j$}{
    Evaluate $\widehat a_{ij}$ from its targeted reconstruction
    \tcp*{\cref{eq:deterministic_targeted_reconstruction}}
    $\widehat a_{ji}\gets\widehat a_{ij}^{\,*}$\;
  }

  \For{each Hamiltonian target $h_p$}{
    Evaluate $\widehat\xi_{0p}$ from its targeted reconstruction\;
    $\widehat h_p\gets\Im(\widehat\xi_{0p})$\;
  }

  \Return $\{\widehat h_p\},\{\widehat a_{ij}\}$\;
\end{boxedalgorithm}

\begin{theorem}[Ancilla-free coefficient learning]
\label{theorem:ancilla_free_coefficient_learning}
Let $\mathcal E_t=e^{t\mathcal L}$ be generated by an $n$-qubit
Lindbladian with $|h_i|,|a_{ij}|\le1$ and at most $M$ nonzero
Pauli components, and let $\mathfrak d\ge1$ be a known integer
upper bound on its dual-interaction-graph degree.
Let $\widehat{\mathcal S}_H,\widehat{\mathcal S}_D$ be the
candidate supports returned by the structure-learning procedure
on its successful event, so that
\begin{equation}
\label{eq:coefficient_learning_support_assumptions}
    \widehat{\mathcal S}_D=\mathcal S_D,
    \qquad
    \mathcal S_H\subseteq\widehat{\mathcal S}_H
    \subseteq
    \mathcal S_D\cup
    \{P_p\neq I:\chi^{(2)}_{pp}(0)>0\}.
\end{equation}
Define
\[
    \widehat M_H:=|\widehat{\mathcal S}_H|,
    \qquad
    \widehat M_D:=|\widehat{\mathcal S}_D|,
    \qquad
    \widehat M:=\widehat M_H+\widehat M_D^2.
\]
Define also
\[
    \Gamma
    :=
    \min\left\{
        (\mathfrak d+2)(\mathfrak d+1),\,M
    \right\}.
\]
For any $0<\varepsilon,\delta<1$,
\cref{algorithm:coefficient_learning} has the following guarantees,
conditional on
\cref{eq:coefficient_learning_support_assumptions}.

\begin{enumerate}

\item \textit{(Accuracy.)}
With probability at least $1-\delta$,
\begin{equation}
\label{eq:coefficient_learning_accuracy}
    |\widehat h_p-h_p|\le\varepsilon,
    \qquad
    |\widehat a_{ij}-a_{ij}|\le\varepsilon
\end{equation}
for every candidate coefficient, including candidates whose
true coefficient is zero.

\item \textit{(Number of experiments.)}
The total number of applications of $\mathcal E_t$ is
\begin{equation}
\label{eq:coefficient_learning_quantum_complexity}
    N_{\exp}
    =
    \widetilde{\mathcal O}\left(
        \widehat M\,\mathfrak d\,
        \frac{\Gamma^2}{\varepsilon^2}
    \right).
\end{equation}

\item \textit{(Time resolution.)}
All queries use evolution times at most
\begin{equation}
\label{eq:coefficient_learning_tau_max}
    \tau_{\max}
    =
    \frac{1}{4\Gamma}.
\end{equation}
and every queried time and the separation between distinct
queried times are at least
\begin{equation}
\label{eq:coefficient_learning_time_resolution}
    t_{\mathrm{res}}
    =
    \Omega\left(
        \frac{1}
        {\Gamma\log^2(\Gamma/\varepsilon)}
    \right),
\end{equation}
up to constant changes inside the logarithm.

\item \textit{(Total evolution time.)}
The total evolution time satisfies
\begin{equation}
\label{eq:coefficient_learning_total_time}
    t_{\mathrm{tot}}
    :=
    \sum_{\text{calls}}t
    =
    \widetilde{\mathcal O}\left(
        \widehat M\,\mathfrak d\,
        \frac{\Gamma}{\varepsilon^2}
    \right).
\end{equation}

\item \textit{(Classical overhead.)}
The deterministic probe families can be constructed in
\begin{equation}
\label{eq:coefficient_learning_preprocessing}
    t_{\mathrm{pre}}
    =
    \widetilde{\mathcal O}\left(
        n\widehat M(\widehat M_D+\mathfrak d+1)
    \right)
\end{equation}
classical operations.
Including processing of the experimental data gives
\begin{equation}
\label{eq:coefficient_learning_classical_complexity}
    t_{\mathrm{classical}}
    =
    t_{\mathrm{pre}}
    +
    \widetilde{\mathcal O}\left(
        n\widehat M\,\mathfrak d\,
        \frac{\Gamma^2}{\varepsilon^2}
    \right).
\end{equation}

\end{enumerate}

Each experimental run uses only product Pauli state preparation,
native evolution under $\mathcal E_t$, and single-qubit Pauli
measurements. No ancillas or controls during the evolution are
required.
\end{theorem}

\begin{proof}
\noindent\textbf{Correctness.}
For each dissipative target $a_{ij}$,
\cref{eq:deterministic_targeted_reconstruction} gives an exact
reconstruction
\[
    a_{ij}
    =
    \sum_{\alpha\in\mathcal R_{ij}}
        w_{ij,\alpha}d_\alpha,
    \qquad
    \sum_{\alpha\in\mathcal R_{ij}}
        |w_{ij},\alpha|
    \le1.
\]
For each Hamiltonian target, the same construction gives
\[
    h_p
    =
    \Im\left(
        \sum_{\alpha\in\mathcal R_p}
        w_{p,\alpha}d_\alpha
    \right),
    \qquad
    \sum_{\alpha\in\mathcal R_p}
        |w_{p,\alpha}|
    \le1.
\]
Consequently, if every selected derivative is estimated to
accuracy $\varepsilon$, then
\[
    |\widehat a_{ij}-a_{ij}|
    \le
    \sum_{\alpha\in\mathcal R_{ij}}
        |w_{ij,\alpha}|
        |\widehat d_\alpha-d_\alpha|
    \le\varepsilon,
\]
and similarly
$|\widehat h_p-h_p|\le\varepsilon$.
Hermitian completion
$\widehat a_{ji}=\widehat a_{ij}^*$
preserves the same error bound.

By \cref{lemma:targeted_probe_counts}, each Hamiltonian or
off-diagonal dissipative target uses
$\mathcal O(\mathfrak d)$ probes.
A diagonal target uses at most
$2\widehat M_D$ probes.
Thus, after removing duplicates,
\begin{align}
\label{eq:coefficient_learning_probe_count}
    |\mathcal R|
    &\le
    2\widehat M_D^2
    +
    (\widehat M-\widehat M_D)(4\mathfrak d+6)
    \nonumber\\
    &\le
    (4\mathfrak d+8)\widehat M
    =
    \mathcal O(\widehat M\mathfrak d).
\end{align}

Allocate failure probability
$\delta/|\mathcal R|$ to each selected derivative.
By \cref{lemma:pauli_expectation_derivative_estimation}, each
derivative requires
\[
    \widetilde{\mathcal O}\left(
        \frac{\Gamma^2}{\varepsilon^2}
    \right)
\]
channel applications, where the logarithmic dependence on
$|\mathcal R|/\delta$ is absorbed into
$\widetilde{\mathcal O}$.
Multiplying by
\cref{eq:coefficient_learning_probe_count} gives
\cref{eq:coefficient_learning_quantum_complexity}, and a union
bound gives the accuracy guarantee with probability at least
$1-\delta$.

\medskip
\noindent\textbf{Time resolution and total evolution time.}
Let $r$ be the Chebyshev degree from
\cref{eq:cheb_params_pauli_derivative}.
Ordering the Gauss--Chebyshev nodes increasingly, the smallest
node is
\[
    t_0
    =
    \tau_{\max}
    \sin^2\left(
        \frac{\pi}{4(r+1)}
    \right)
    \ge
    \frac{\tau_{\max}}{4(r+1)^2}.
\]
Adjacent nodes are separated by at least
$2\tau_{\max}/(r+1)^2$.
Since
$\tau_{\max}=1/(4\Gamma)$ and
$r=\mathcal O(\log(\Gamma/\varepsilon))$,
this proves
\cref{eq:coefficient_learning_time_resolution}.
Furthermore,
\[
    t_{\mathrm{tot}}
    \le
    N_{\exp}\tau_{\max},
\]
which together with
\cref{eq:coefficient_learning_quantum_complexity}
gives
\cref{eq:coefficient_learning_total_time}.

\medskip
\noindent\textbf{Classical overhead.}
For each target, the relevant conflicts and product-sector pairs
can be identified by scanning
$P_u\in\widehat{\mathcal S}_D$:
once $P_u$ is fixed, the product constraint determines the
possible partner $P_v$, which is checked by dictionary lookup.
Using length-$n$ Pauli labels, constructing the enlarged region
and its nuisance list costs
$\widetilde{\mathcal O}(n(\widehat M_D+1))$
per target.
Over at most $\widehat M$ targets this contributes
\[
    \widetilde{\mathcal O}\left(
        n\widehat M(\widehat M_D+1)
    \right).
\]

Once the nuisance list is known,
\cref{lemma:deterministic_cancellation_generators} constructs
its cancellation generators and the resulting probe labels and
weights in
$\mathcal O(|T|(m+1))$ operations.
The probe-count argument gives
\[
    \sum_{\text{targets}}(m+1)
    =
    \mathcal O(\widehat M\mathfrak d).
\]
Since every enlarged region satisfies $|T|\le n$, this contributes
$\mathcal O(n\widehat M\mathfrak d)$.
Combining the two contributions gives
\cref{eq:coefficient_learning_preprocessing}.

Finally, a direct implementation requires at most
$\mathcal O(n)$ classical work per experimental record.
Combining this with
\cref{eq:coefficient_learning_quantum_complexity}
and the lower-order interpolation and reconstruction costs gives
\cref{eq:coefficient_learning_classical_complexity}.
\end{proof}

\begin{remark}[Optimality of the accuracy dependence]
In the worst case over dissipative generators, coefficient
estimation requires $\Omega(\varepsilon^{-2})$ channel uses,
consistent with standard-quantum-limit bounds for noisy parameter
estimation~\cite{huelga1997improvement,demkowicz2014using,
escher2011general,zhou2018achieving}.
Thus the $\widetilde{\mathcal O}(\varepsilon^{-2})$ dependence
of \cref{algorithm:coefficient_learning} is optimal up to
logarithmic factors.
\end{remark}

\begin{corollary}[Optional shadow-process parallelization]
\label{cor:coefficient_learning_parallelization}
Under the assumptions of
\cref{theorem:ancilla_free_coefficient_learning}, let
\[
    k_H:=
    \max_{P\in\widehat{\mathcal S}_H}|\supp{P}|,
    \qquad
    k_D:=
    \max_{P_i,P_j\in\widehat{\mathcal S}_D}
        |\supp{P_i}\cup\supp{P_j}|,
\]
with maxima over empty families taken to be zero.
Using the shadow-process procedure of
\cref{sec:parallel_derivative_estimation}, all selected
derivatives can instead be estimated using
\begin{equation}
\label{eq:coefficient_learning_parallel_complexity}
    N_{\exp}^{\mathrm{par}}
    =
    \widetilde{\mathcal O}\left(
        \left(
            9^{k_H+\mathfrak d}
            +
            9^{k_D+\mathfrak d}
        \right)
        \frac{\Gamma^2}{\varepsilon^2}
    \right).
\end{equation}
channel applications.
In particular, with $k:=\max\{k_H,k_D\}$,
\[
    N_{\exp}^{\mathrm{par}}
    =
    \widetilde{\mathcal O}\left(
        9^{k+\mathfrak d}
        \frac{\Gamma^2}{\varepsilon^2}
    \right).
\]
The corresponding total evolution time is
\[
    t_{\mathrm{tot}}^{\mathrm{par}}
    =
    \widetilde{\mathcal O}\left(
        9^{k+\mathfrak d}
        \frac{\Gamma}{\varepsilon^2}
    \right).
\]
\end{corollary}

\begin{remark}[Choosing between the two implementations]
The serialized and parallel procedures use the same selected
probe families and the same Chebyshev schedule.
One may therefore choose the less expensive implementation,
giving
\begin{equation}
\label{eq:coefficient_learning_best_complexity}
    N_{\exp}
    =
    \widetilde{\mathcal O}\left(
        \min\left\{
            \widehat M\mathfrak d,\,
            9^{k+\mathfrak d}
        \right\}
        \frac{\Gamma^2}{\varepsilon^2}
    \right).
\end{equation}
The parallel implementation is useful when the selected probes
are sufficiently local; the serialized guarantee remains
polynomial without any bounded-locality assumption.
\end{remark}
\section{Optimality of time resolution}
\label{sec:appendix-optimality-of-time-resoltion}

Here we formalize and prove the time-resolution optimality statement introduced in
the main Section~V. Concretely, we show that if the minimum accessible evolution
time (time resolution) is at least
$t_{0,\kappa}=\Theta_\kappa(M_\kappa^{2/\kappa-1})$ for some fixed integer
$\kappa\geq2$ and all sufficiently large $n$, then even the \emph{simpler} task of
learning only the dissipator structure becomes information-theoretically hard:
any protocol may require an exponential number of samples. Here, the minimum
accessible evolution time is a lower bound on the duration of every use of the
evolution channel; it is distinct from the spacing between two sampling nodes.

We establish this optimality by proving a sample-complexity lower bound for
structure learning of Lindbladians containing only dissipator terms, under a
single-copy experimental access model that is at least as permissive in state
preparation and measurement as the one assumed by our structure-learning
algorithm. We use phase-free Pauli representatives
\begin{align}
\mathcal P_n:=\{I,X,Y,Z\}^{\otimes n}.
\end{align}

\begin{definition}[Practical Lindbladian learning model]
\label{def:practical_lindbladian_learning_lower}
Consider learning an $M$-sparse Lindbladian with single-copy inputs, a minimum
evolution time $t_0$ (time resolution), and single-copy measurements.
Specifically, we restrict the Lindbladians we consider and model the learning
protocols as follows:
\begin{itemize}
\item We consider Lindbladians with the following expansion in the Pauli basis:
\begin{equation}
\label{eq:lower_Pauli_exp}
\mathcal{L}(\rho)=-i\sum_{P_i\in \mathcal{S}_H}h_i[P_i,\rho]
+\sum_{P_j,P_k\in \mathcal{S}_D}a_{jk}
\Big(P_j\rho P_k-\tfrac12\{P_kP_j,\rho\}\Big),
\end{equation}
where $h_i\in\mathbb R$, $A=(a_{jk})=A^\dagger\succeq0$,
$|h_i|\leq1$, and $|a_{jk}|\leq1$. We call the Lindbladian $M$-sparse
when the number of nonzero scalar coefficients is at most $M$:
\begin{align}
\bigl|\{i:h_i\neq0\}\bigr|
+\bigl|\{(j,k):a_{jk}\neq0\}\bigr|\leq M.
\label{eq:lower-sparsity}
\end{align}
We take $\mathcal S_H,\mathcal S_D\subseteq\mathcal P_n\setminus\{I\}$
and interpret coefficients indexed outside these sets as zero.
Here, we consider the case when $M=\mathrm{poly}(n)$.

\item We consider protocols with the following features as in
\Cref{theorem:ancilla_free_structure_learning}:
\begin{enumerate}
    \item Prepare an $n$-qubit (single-copy) state
    $\rho_{\mathrm{in}}$ as input.
    \item Evolve the Lindbladian for time $t\geq t_0$ on
    $\rho_{\mathrm{in}}$.
    \item Measure the output state. Here, we relax the protocols we consider
    to allow arbitrary measurements on the $n$ qubits. Each round ends with
    this measurement, and only its classical outcome is retained. The input
    and measurement act only on the $n$ system qubits; no ancilla or quantum
    register is retained between rounds.
\end{enumerate}
There can be adaptivity in preparing $\rho_{\mathrm{in}}$ and choosing the
measurement to perform based on the history of measurement outcomes so far.

\item The goal is to identify the structure of the Lindbladian up to parameter
$\eta$. In particular, for the diagonal dissipative instances below, the
protocol must correctly decide which coefficients $a_{PP}$ are at least $\eta$,
with high probability.
\end{itemize}
\end{definition}

We prove the following lower bound on the number of samples. Since every sample
uses an evolution time at least $t_0$, this also gives a lower bound on the
total evolution time.
\begin{theorem}\label{thm:lower}
Consider the practical Lindbladian learning model in
\Cref{def:practical_lindbladian_learning_lower}. Fix a constant integer
$\kappa\geq2$. For every sufficiently large $n$ (so that
$\kappa\leq\lfloor n/2\rfloor$), set
\begin{align}
M_\kappa&:=3^\kappa\binom{n}{\kappa},\\
N_{\star,\kappa}&:=\frac49 3^\kappa\binom{n-1}{\kappa-1},\\
t_{0,\kappa}&:=\frac{3n\ln2+2}{4(N_{\star,\kappa}-1)}
=\frac{9n(3n\ln2+2)}{16\kappa M_\kappa-36n}.
\label{eq:lower-time-threshold}
\end{align}
There are two purely dissipative Lindbladians in the class of at most
$M_\kappa$-sparse Lindbladians, with distinct dissipator supports, such that
any possibly adaptive protocol which identifies the dissipator structure above
threshold $\eta=1/3$ with success probability at least $2/3$ on both instances
requires
\begin{align}
m\geq 3\cdot2^{n-2}
\label{eq:lower-sample-bound}
\end{align}
samples whenever $t_0\geq t_{0,\kappa}$. Moreover, for fixed $\kappa$,
\begin{align}
t_{0,\kappa}=\Theta_\kappa\bigl(M_\kappa^{2/\kappa-1}\bigr),
\qquad M_\kappa=\Theta_\kappa(n^\kappa).
\label{eq:lower-time-scaling}
\end{align}
Consequently, the total evolution time is at least
$3\cdot2^{n-2}t_0$.
\end{theorem}

\subsection{Overview of the proof}
We now explain the high-level idea behind the \emph{time-resolution lower bound}
in \Cref{thm:lower}. The core idea is to construct two purely dissipative
Lindbladians $\mathcal L_{0,\kappa}$ and $\mathcal L_{1,\kappa}$ with distinct
dissipator supports, such that once the evolution time is restricted to
$t\geq t_{0,\kappa}$, the resulting output states become nearly
indistinguishable. This shows that enforcing a coarser time resolution can
render structure learning exponentially hard.

\begin{itemize}
\item \emph{Reduction to distinguishing.}
Any protocol that learns the dissipator structure with threshold $\eta=1/3$ can
in particular distinguish between the following two cases:
\begin{enumerate}
    \item The underlying Lindbladian is the
    $M_\kappa=\Theta_\kappa(n^\kappa)$-sparse Lindbladian
    $\mathcal L_{0,\kappa}$ with only dissipator terms;
    \item The underlying Lindbladian is the alternative
    $(M_\kappa-1)$-sparse Lindbladian $\mathcal L_{1,\kappa}$ with only
    dissipator terms.
\end{enumerate}
Both instances belong to the promised class of at most $M_\kappa$-sparse
Lindbladians, but their dissipator supports are different.

\item \emph{Rapid mixing at coarse time resolution.}
After evolving either Lindbladian on any single-copy input state
$\rho_{\mathrm{in}}$ for time $t\geq t_{0,\kappa}$, the output state is
exponentially close to the maximally mixed state $I/2^n$. Concretely, for both
hypotheses, $\rho_{\ell,\kappa}(t):=
e^{\mathcal L_{\ell,\kappa}t}(\rho_{\mathrm{in}})$ with
$\ell\in\{0,1\}$ satisfies
\[
\left\|\rho_{\ell,\kappa}(t)-\frac{I}{2^n}\right\|_1 \leq 2^{-n}.
\]

\item \emph{Indistinguishability and exponential samples.}
For an adaptive protocol, condition on the preceding measurement history. At
each sample, the distribution of the next measurement outcome is close to the
distribution obtained by measuring $I/2^n$. A sequential hybrid over the
classical transcript then shows that the two hypotheses have total variation
distance at most $m2^{-n}$ after $m$ samples. Standard binary hypothesis
testing therefore requires an exponential number of samples to achieve
constant success probability.

\item \emph{Conclusion.}
Therefore, for every fixed integer $\kappa\geq2$, if the minimum evolution time
satisfies $t_0\geq
t_{0,\kappa}=\Theta_\kappa(M_\kappa^{2/\kappa-1})$, then there are
$M_\kappa=\Theta_\kappa(n^\kappa)$-sparse Lindbladians with only dissipator
terms for which learning the dissipator structure above threshold $\eta=1/3$
requires
$\exp(\Omega(n))$ samples.
\end{itemize}

In the remainder of this section we provide the construction of
$\mathcal L_{0,\kappa}$ and $\mathcal L_{1,\kappa}$ and prove the rapid-mixing
and indistinguishability claims stated above.

\subsection{Construction of $\mathcal{L}_{0,\kappa}$}\label{sec:L0}
We start with the construction of $\mathcal L_{0,\kappa}$. Fix integers
$n\geq\kappa\geq2$. Define the balanced weight-$\kappa$ Pauli set
\begin{align}
\mathcal T_\kappa:=\{P\in\mathcal P_n:\ |\mathrm{supp}(P)|=\kappa,\
\text{and on its support }P_i\in\{X,Y,Z\}\}.
\end{align}
By a counting argument, we have
$|\mathcal T_\kappa|=3^\kappa\binom{n}{\kappa}=M_\kappa$. We define
\begin{align}
\mathcal L_{0,\kappa}(\rho):=
\sum_{P\in\mathcal T_\kappa}\bigl(P\rho P-\rho\bigr).
\label{eq:lower_L0}
\end{align}
We express $\mathcal L_{0,\kappa}$ in the Pauli-basis expansion in
\eqref{eq:lower_Pauli_exp}. We note that the $h_i$ vanish, indicating that
$\mathcal L_{0,\kappa}$ is a Lindbladian with only dissipator terms. The
Kossakowski matrix is diagonal with $a_{PP}=1$ for
$P\in\mathcal T_\kappa$, and $a_{jk}=0$ otherwise. Thus all coefficients
satisfy $|a_{jk}|\leq1$ and the sparsity of $\mathcal L_{0,\kappa}$ is
$M_\kappa$.

\paragraph{Exponential decay given any single-copy input state.}
Given any single-copy input state $\rho_{\mathrm{in}}$, we now show that the
output state $\rho_{0,\kappa}(t):=
e^{\mathcal L_{0,\kappa}t}(\rho_{\mathrm{in}})$ is exponentially close to the
maximally mixed state when $t\geq t_{0,\kappa}$.

Let $\mathcal L_{0,\kappa}^\dagger$ denote the Heisenberg adjoint:
\begin{align}
\operatorname{Tr}\bigl(O\mathcal L_{0,\kappa}(\rho)\bigr)
=\operatorname{Tr}\bigl((\mathcal L_{0,\kappa}^\dagger(O))\rho\bigr).
\end{align}
Since
$\operatorname{Tr}(O(P\rho P-\rho))=\operatorname{Tr}((POP-O)\rho)$, we have
\begin{align}
\mathcal L_{0,\kappa}^\dagger(O)
=\mathcal L_{0,\kappa}(O)
=\sum_{P\in\mathcal T_\kappa}(POP-O).
\end{align}
For any Pauli observable $Q\in\mathcal P_n$, a term indexed by
$P\in\mathcal T_\kappa$ contributes $PQP-Q$. If $[P,Q]=0$, then $PQP=Q$ and
this term contributes $0$. If $\{P,Q\}=0$, then $PQP=-Q$ and this term
contributes $-2Q$. Define
\begin{align}
N_{\mathrm{ac}}^{(\kappa)}(Q):=
\bigl|\{P\in\mathcal T_\kappa:\{P,Q\}=0\}\bigr|.
\end{align}
Consequently, writing
$\langle Q\rangle_t:=\operatorname{Tr}(Q\rho_{0,\kappa}(t))$ for
$\rho_{0,\kappa}(t)=e^{t\mathcal L_{0,\kappa}}(\rho_{\mathrm{in}})$,
\begin{align}
\frac{d}{dt}\langle Q\rangle_t
&=-2N_{\mathrm{ac}}^{(\kappa)}(Q)\langle Q\rangle_t,
&
\langle Q\rangle_t
&=e^{-2N_{\mathrm{ac}}^{(\kappa)}(Q)t}\langle Q\rangle_0.
\end{align}
We first show a uniform lower bound on $N_{\mathrm{ac}}^{(\kappa)}(Q)$ for any
nonidentity $Q$.

\begin{lemma}
For every nonidentity Pauli string $Q\in\mathcal P_n$,
\begin{align}
N_{\mathrm{ac}}^{(\kappa)}(Q)
\geq \frac49 3^\kappa\binom{n-1}{\kappa-1}.
\end{align}
\end{lemma}

\begin{proof}
Pick any site $i$ with $Q_i\neq I$. Consider all supports $S\subset[n]$ of
size $\kappa$ containing $i$; there are $\binom{n-1}{\kappa-1}$ such supports.
Fix such an $S$. Let $r:=|\mathrm{supp}(Q)\cap S|$. Since
$i\in\mathrm{supp}(Q)\cap S$, we have $r\geq1$.

We denote $\mathcal T_\kappa(S)$ as the set of strings supported exactly on
$S$, and count how many anticommute with $Q$. On each of the $r$ sites where
$Q$ is nonidentity, among the three choices $\{X,Y,Z\}$, exactly two
anticommute with $Q$ and one commutes; on the other $\kappa-r$ sites, all
three commute with $I$. Let $A$ be the number of anticommuting sites among
those $r$ sites. Then $P$ anticommutes with $Q$ iff $A$ is odd. Therefore,
\begin{align}
\bigl|\{P\in\mathcal T_\kappa(S):\{P,Q\}=0\}\bigr|
&=3^{\kappa-r}\sum_{a\ \mathrm{odd}}\binom{r}{a}2^a\\
&=\frac{3^\kappa}{2}\bigl(1-(-1/3)^r\bigr)
\geq\frac49 3^\kappa.
\end{align}
The final inequality holds for every $r\geq1$. Summing over the
$\binom{n-1}{\kappa-1}$ supports containing $i$ gives the lemma.
\end{proof}

Given the above lemma, denote
\begin{align}
N_{\star,\kappa}:=\frac49 3^\kappa\binom{n-1}{\kappa-1}.
\end{align}
For every $Q\neq I$, we have
$N_{\mathrm{ac}}^{(\kappa)}(Q)\geq N_{\star,\kappa}$. We use the slightly
weaker rate $N_{\star,\kappa}-1$, which will apply to both hard instances.
Thus
\begin{align}
|\langle Q\rangle_t|
\leq e^{-2(N_{\star,\kappa}-1)t}|\langle Q\rangle_0|.
\end{align}
Expanding in the Pauli basis,
\begin{align}
\rho_{0,\kappa}(t)=\frac{1}{2^n}
\left(I+\sum_{Q\neq I}\langle Q\rangle_tQ\right),
\end{align}
and using Pauli orthogonality gives
\begin{align}
\left\|\rho_{0,\kappa}(t)-\frac{I}{2^n}\right\|_2^2
&=\frac{1}{2^n}\sum_{Q\neq I}\langle Q\rangle_t^2\\
&\leq \frac{e^{-4(N_{\star,\kappa}-1)t}}{2^n}
\sum_{Q\neq I}\langle Q\rangle_0^2.
\end{align}
For any state $\rho_{\mathrm{in}}$,
\begin{align}
\sum_{Q\neq I}\langle Q\rangle_0^2
=2^n\operatorname{Tr}(\rho_{\mathrm{in}}^2)-1
\leq 2^n.
\end{align}
Hence
\begin{align}
\left\|\rho_{0,\kappa}(t)-\frac{I}{2^n}\right\|_1
\leq2^{n/2}e^{-2(N_{\star,\kappa}-1)t}.
\end{align}
For $t\geq t_{0,\kappa}$, the right-hand side is at most
$e^{-1}2^{-n}\leq2^{-n}$.

Concluding the above derivation, we obtain the following lemma.
\begin{lemma}\label{lem:lower_L0_decay}
Consider the $M_\kappa=\Theta_\kappa(n^\kappa)$-sparse Lindbladian consisting
only of dissipator terms in \eqref{eq:lower_L0}. Given any single-copy input
state $\rho_{\mathrm{in}}$, the output state
$\rho_{0,\kappa}(t):=
e^{\mathcal L_{0,\kappa}t}(\rho_{\mathrm{in}})$ satisfies
\begin{align}
\left\|\rho_{0,\kappa}(t)-\frac{I}{2^n}\right\|_1
\leq e^{-1}2^{-n}\leq2^{-n}
\end{align}
whenever
\begin{align}
t\geq t_{0,\kappa}
=\frac{9n(3n\ln2+2)}{16\kappa M_\kappa-36n}
=\Theta_\kappa(M_\kappa^{2/\kappa-1}).
\end{align}
\end{lemma}

\subsection{Construction of $\mathcal{L}_{1,\kappa}$}\label{sec:L1}
We then provide the construction of $\mathcal L_{1,\kappa}$. Let
\begin{align}
P_\star:=X^{\otimes\kappa}\otimes I^{\otimes(n-\kappa)},
\end{align}
and define the alternative
weight-$\kappa$ Pauli set
\begin{align}
\mathcal T_\kappa':=\mathcal T_\kappa\setminus\{P_\star\}.
\end{align}
By a counting argument,
$|\mathcal T_\kappa'|=|\mathcal T_\kappa|-1
=3^\kappa\binom{n}{\kappa}-1$. We define
\begin{align}
\mathcal L_{1,\kappa}(\rho):=
\sum_{P\in\mathcal T_\kappa'}\bigl(P\rho P-\rho\bigr).
\label{eq:lower_L1}
\end{align}
Again, the Hamiltonian coefficients vanish and the Kossakowski matrix is
diagonal with $a_{PP}=1$ for $P\in\mathcal T_\kappa'$ and $a_{jk}=0$
otherwise. Thus all coefficients satisfy $|a_{jk}|\leq1$ and
$\mathcal L_{1,\kappa}$ has $M_\kappa-1$ nonzero coefficients; in particular,
it belongs to the class of at most $M_\kappa$-sparse Lindbladians.

For every nonidentity $Q$, deleting one Pauli string can lower the
anticommuting count by at most one. Therefore,
\begin{align}
\bigl|\{P\in\mathcal T_\kappa':\{P,Q\}=0\}\bigr|
\geq N_{\star,\kappa}-1.
\end{align}
Following the same derivation as for \Cref{lem:lower_L0_decay}, we have the
following lemma.
\begin{lemma}\label{lem:lower_L1_decay}
Given any single-copy input state $\rho_{\mathrm{in}}$, the output state
$\rho_{1,\kappa}(t):=
e^{\mathcal L_{1,\kappa}t}(\rho_{\mathrm{in}})$ satisfies
\begin{align}
\left\|\rho_{1,\kappa}(t)-\frac{I}{2^n}\right\|_1
\leq e^{-1}2^{-n}\leq2^{-n}
\end{align}
whenever $t\geq t_{0,\kappa}$.
\end{lemma}

\subsection{Reducing to the state distinguishing problem}
Finally, we wrap everything up into a proof of \Cref{thm:lower}.

\begin{proof}[Proof of \Cref{thm:lower}]
Suppose that a practical protocol in
\Cref{def:practical_lindbladian_learning_lower} identifies the dissipator
structure above threshold $\eta=1/3$ with success probability at least $2/3$
on both $\mathcal L_{0,\kappa}$ and $\mathcal L_{1,\kappa}$. Given its output,
we can solve the following Lindbladian distinguishing task:
\begin{problem}[Lindbladian distinguishing task]
\label{prob:lindbladian_distinguish}
We are guaranteed that either of the following holds for an unknown $n$-qubit
Lindbladian:
\begin{enumerate}
\item The underlying Lindbladian is $\mathcal L_{0,\kappa}$ defined in
\eqref{eq:lower_L0};
\item The underlying Lindbladian is $\mathcal L_{1,\kappa}$ defined in
\eqref{eq:lower_L1}.
\end{enumerate}
The goal is to distinguish between the two cases with success probability at
least $2/3$.
\end{problem}
Indeed, $a_{P_\star P_\star}=1$ for $\mathcal L_{0,\kappa}$ and
$a_{P_\star P_\star}=0$ for $\mathcal L_{1,\kappa}$. Thus, correct threshold
support recovery distinguishes the two cases.

We now account for the allowed classical adaptivity. Let
\begin{align}
\tau:=\frac{I}{2^n},
\qquad
\varepsilon:=e^{-1}2^{-n},
\end{align}
where $\operatorname{TV}(p,q):=\frac12\|p-q\|_1$.
Consider an arbitrary protocol using at most $m$ samples. If the protocol
stops earlier, append deterministic dummy outcomes $\bot$ to its transcript;
this does not change its decision or distinguishing power. Any private
randomization can be included in the classical transcript. Conditional on a
nonterminal transcript history $h$ before sample $i$, the protocol selects an input state
$\rho_{\mathrm{in}}^{(i)}(h)$, evolves it for an allowed time
$t_i\geq t_{0,\kappa}$, and chooses an arbitrary measurement
$\mathcal M_i(h)$ on the resulting $n$-qubit output. Let
$p_{\ell,i}(\cdot\mid h)$ be the conditional distribution of the outcome under
$\mathcal L_{\ell,\kappa}$, and let $q_i(\cdot\mid h)$ be the conditional
distribution obtained by applying the same measurement to $\tau$. At a terminal
history, both conditional distributions are the point mass at $\bot$. By
\Cref{lem:lower_L0_decay,lem:lower_L1_decay} and contractivity of trace
distance under measurements,
\begin{align}
\operatorname{TV}\bigl(p_{\ell,i}(\cdot\mid h),q_i(\cdot\mid h)\bigr)
&\leq \frac12\left\|
e^{t_i\mathcal L_{\ell,\kappa}}
\bigl(\rho_{\mathrm{in}}^{(i)}(h)\bigr)-\tau\right\|_1\nonumber\\
&\leq \frac{\varepsilon}{2},
\qquad \ell\in\{0,1\}.
\label{eq:lower-one-round-TV}
\end{align}

Let $\mathsf P_\ell^{(i)}$ denote the distribution of the private randomness
together with the first $i$ measurement outcomes under hypothesis $\ell$, and let
$\mathsf P_{\mathrm{mix}}^{(i)}$ be the corresponding reference-transcript
distribution obtained by using $q_i(\cdot\mid h)$ at every sample. The
sequential hybrid bound gives
\begin{align}
\operatorname{TV}\bigl(\mathsf P_\ell^{(i)},
\mathsf P_{\mathrm{mix}}^{(i)}\bigr)
\leq
\operatorname{TV}\bigl(\mathsf P_\ell^{(i-1)},
\mathsf P_{\mathrm{mix}}^{(i-1)}\bigr)
+\frac{\varepsilon}{2}.
\end{align}
Hence, by induction,
\begin{align}
\operatorname{TV}\bigl(\mathsf P_\ell^{(m)},
\mathsf P_{\mathrm{mix}}^{(m)}\bigr)
\leq\frac{m\varepsilon}{2},
\qquad \ell\in\{0,1\}.
\end{align}
Consequently,
\begin{align}
\operatorname{TV}\bigl(\mathsf P_0^{(m)},\mathsf P_1^{(m)}\bigr)
\leq m\varepsilon.
\label{eq:lower-transcript-TV}
\end{align}
Any final decision is a post-processing of this classical transcript. For
equal priors, its success probability is therefore at most
\begin{align}
p_{\mathrm{succ}}
&\leq\frac12\left(1+
\operatorname{TV}\bigl(\mathsf P_0^{(m)},\mathsf P_1^{(m)}\bigr)\right)
\nonumber\\
&\leq\frac12+\frac{m\varepsilon}{2}.
\end{align}
Thus success probability at least $2/3$ requires
\begin{align}
m\geq\frac{1}{3\varepsilon}=\frac{e}{3}2^n
\geq3\cdot2^{n-2}.
\end{align}
This proves the sample lower bound. Since every sample has evolution time at
least $t_0$, the total-evolution-time lower bound follows immediately.
\end{proof}

Since $\kappa$ can be chosen as any fixed constant, the exponent
$2/\kappa-1$ can be made arbitrarily close to $-1$. More precisely, for every
fixed $\theta>0$, choosing a fixed integer $\kappa>2/\theta$ gives
$t_{0,\kappa}=o(M_\kappa^{-1+\theta})$. Hence,
$t_0=\Omega_\theta(M_\kappa^{-1+\theta})$ implies
$t_0\geq t_{0,\kappa}$ for all sufficiently large $n$, and the theorem
applies. This gives the claimed $M^{-1}$ polynomial exponent for worst-case
constant-threshold dissipator-support identification in this single-copy model.
\section{SPAM robustness}
\label{sec:spam_robustness}

In realistic experimental platforms, noise associated with state preparation and measurement (SPAM) is often a substantial source of error. For the learning guarantees derived above to be experimentally meaningful, our algorithm must therefore be robust to SPAM noise. Since the algorithm consists of a structure-learning stage followed by coefficient learning, both stages must tolerate such imperfections.

We consider a standard and widely used SPAM model in which errors are described by single-qubit depolarizing channels acting independently on each qubit. Concretely, for an intended single-qubit state $\rho$, the actual prepared state is
\begin{equation}
    \rho \;\longmapsto\; \mathcal E_{\mathrm P}(\rho)
    :=
    r_{\mathrm P}\rho + (1-r_{\mathrm P})\tfrac{I}{2},
\end{equation}
and immediately before the measurement, the single-qubit state undergoes
\begin{equation}
    \rho \;\longmapsto\; \mathcal E_{\mathrm M}(\rho)
    :=
    r_{\mathrm M}\rho + (1-r_{\mathrm M})\tfrac{I}{2}.
\end{equation}
We assume the retention parameters $r_{\mathrm P}$ and $r_{\mathrm M}$ are known from prior calibration and write
$r=r_{\mathrm P}r_{\mathrm M}$. In our access model, an experiment therefore implements the sequence
\begin{equation}
\label{eq:spam_access_model}
    \rho
    \;\longmapsto\;
    \mathcal E_{\mathrm M}^{\otimes n}
    \circ e^{t\mathcal L}
    \circ \mathcal E_{\mathrm P}^{\otimes n}(\rho),
\end{equation}
after which the measurement is performed.

First, consider the SPAM-robustness of the structure-learning step, which relies on Pauli channel estimation via population recovery. It was shown recently that Pauli error estimation can be made robust to SPAM noise of the above form~\cite{o2025spam}. Writing $\delta_s:=1-r$ to quantify the magnitude of the overall SPAM error, the SPAM-robust algorithm operates in two regimes. In the low-SPAM regime $\delta_s \ll \tfrac{1}{n}\log{\tfrac{1}{\varepsilon_s}}$, the sample complexity matches the SPAM-free setting up to a polynomial factor $\operatorname{poly}\bigl(1/\varepsilon_s\bigr)$ in terms of the per-node measurement accuracy $\varepsilon_s$. Consequently, since the guarantees of \Cref{result:structure_learning_informal} rely on $\varepsilon_s \le \widetilde{\Theta}(\eta^2/M^2)$, they continue to hold with an additional $\operatorname{poly}\bigl(M/\eta \bigr)$ sample overhead. This regime, however, requires SPAM rates that decrease with system size and may be unrealistic experimentally. In the general regime, where $\delta_s \gg \tfrac{1}{n}\log{\tfrac{1}{\varepsilon_s}}$, the sample complexity incurs a subexponential overhead $\exp(\Theta(n^{1/3}))$. Moreover, \citet{o2025spam} provide evidence that tolerating even a tiny amount of SPAM in Pauli error estimation must incur an $\exp(n^{1/3})$ overhead, implying that any structure-learning method based on Pauli error estimation necessarily inherits this scaling. While asymptotically unfavorable, this overhead grows slowly with system size: $e^{n^{1/3}}\approx 40$ at $n=50$ and $\approx 2\times10^{4}$ at $n=1000$, up to constants in the exponent.

We now turn to the coefficient-learning stage. Under the depolarizing SPAM model above, the preparation and measurement noise channels act diagonally in the Pauli basis. In particular, for any $n$-qubit Pauli operator $P$,
\begin{equation}
    \mathcal E_{\mathrm P}^{\otimes n}(P)=r_{\mathrm P}^{w(P)}P,
    \qquad
    \bigl(\mathcal E_{\mathrm M}^{\otimes n}\bigr)^\dagger(P)=r_{\mathrm M}^{w(P)}P,
\end{equation}
where $w(P)$ denotes the Pauli weight. As a consequence, the short-time Pauli derivative data accessed by the
algorithm are damped relative to the noiseless case. Concretely, instead of $\tfrac{1}{2^n}\tr\bigl(Q\,\mathcal L^\dagger(O)\bigr)$, we obtain
\begin{equation}
\label{eq:spam_damped_derivative_data}
    \tfrac{1}{2^n}\tr\Bigl(
        \mathcal E_{\mathrm P}^{\otimes n}(Q)\,
        \mathcal L^\dagger\bigl((\mathcal E_{\mathrm M}^{\otimes n})^\dagger(O)\bigr)
    \Bigr)
    =
    r_{\mathrm P}^{w(Q)}\,r_{\mathrm M}^{w(O)}\,
    \tfrac{1}{2^n}\tr\bigl(Q\,\mathcal L^\dagger(O)\bigr).
\end{equation}
Recovering the noiseless derivative therefore amounts to rescaling the measured data by the known factor $r_{\mathrm P}^{-w(Q)}r_{\mathrm M}^{-w(O)}$, which amplifies variance and therefore increases the required sampling precision. Specifically, every selected probe is supported on the enlarged region of its target, so $w(Q),w(O)\le k+\mathfrak d$ with $k$ the locality of \cref{sec:appendix_notation_prelim} evaluated on the candidate structures, and the worst-case variance amplification scales as $r^{-2(k+\mathfrak d)}$ with $r=r_{\mathrm P}r_{\mathrm M}$. Consequently, coefficient learning remains SPAM robust under known single-qubit depolarizing noise at the cost of increasing the total sample complexity of \Cref{result:coefficient_learning_informal} by a factor of $r^{-2(k+\mathfrak d)}$, analogous to the analysis of~\cite{stilck2024efficient, franca2025learning}.

\end{document}